\documentclass[journal,onecolumn
]{IEEEtran}
\usepackage{amsthm}
\usepackage{amsmath}
\usepackage{mathrsfs}
\usepackage{amssymb} 
\usepackage{bbm,dsfont} 
\usepackage{multirow} 
\usepackage{units}
\usepackage{stmaryrd}   
\usepackage{verbatim}
\usepackage{enumerate} 
\usepackage[table]{xcolor}
\usepackage{slashbox} 
\usepackage{ wasysym }
\usepackage{tikz}
\usetikzlibrary{calc}
\usepackage{tabularx}

\usepackage{lipsum}
\usepackage{multicol}
\usepackage{float}

\usepackage{ifthen}

\usepackage{xcolor}
\usepackage{hyperref} 
\usepackage
{hyperref} 
\hypersetup{
    colorlinks,
    linkcolor={blue!80!black},
    citecolor={green!50!black},
    urlcolor={blue!80!black}
}

\newcommand\blfootnote[1]{%
  \begingroup
  \renewcommand\thefootnote{}\footnote{#1}%
  \addtocounter{footnote}{-1}%
  \endgroup
}

\usepackage{graphicx}
\graphicspath{{images/}}

\title{The Capacity Region of the Arbitrarily Varying MAC: With and Without Constraints}
\author{\IEEEauthorblockN{Uzi Pereg and Yossef Steinberg}\\
\IEEEauthorblockA{Department of Electrical Engineering, 
Technion, Haifa 32000, Israel.\\
Email: {\tt uzipereg@campus.technion.ac.il}, {\tt ysteinbe@ee.technion.ac.il}
 }
}

\usepackage[square,numbers]{natbib}
\definecolor{light-gray}{gray}{0.8}

\definecolor{dark-gray}{gray}{0.3}
\usepackage{accents}
\newlength{\dhatheight}

\newcommand{\bieee}{\begin{IEEEeqnarray}{rCl}}
\newcommand{\eieee}{\end{IEEEeqnarray}}
\newcommand{\prob}[1]{\Pr\left(#1\right)}
\newcommand{\given}{\mid}
\newcommand{\cprob}[2]{\Pr\left(#1\given #2\right)}
\newcommand{\E}{\mathbb{E}}
\newcommand{\var}{\mathbb{V}\mathrm{ar}}
\newcommand{\eps}{\varepsilon}

\newcommand{\ie}{\emph{i.e.} }
\newcommand{\eg}{\emph{e.g.} }
\newcommand{\etal}{\emph{et al.} }
\newcommand{\cf}{\emph{cf.} }

\newcommand{\tm}{\widetilde{m}}	
																				
\newcommand{\tq}{\widetilde{q}}

\newcommand{\tX}{\widetilde{X}}
\newcommand{\tY}{\widetilde{Y}}
\newcommand{\tS}{\widetilde{S}}
\newcommand{\tZ}{\widetilde{Z}}

\newcommand{\tx}{\tilde{x}}
\newcommand{\ts}{\widetilde{s}}

\newcommand{\tLambda}{\widetilde{\Lambda}}
\newcommand{\oS}{\overline{S}}

\newcommand{\oq}{\overline{q}}

\newcommand{\hm}{\hat{m}}

\newcommand{\hP}{\hat{P}}

\newcommand{\Aset}{\mathcal{A}}

\newcommand{\Dset}{\mathcal{D}}
\newcommand{\Fset}{\mathcal{F}}

\newcommand{\Uset}{\mathcal{U}}
\newcommand{\Vset}{\mathcal{V}}
\newcommand{\Qset}{\mathcal{Q}}

\newcommand{\Sset}{\mathcal{S}}
\newcommand{\Wset}{\mathcal{W}}
\newcommand{\Xset}{\mathcal{X}}
\newcommand{\Yset}{\mathcal{Y}}

\newcommand{\Eset}{\mathcal{E}}
\newcommand{\markovC}[1]{%
\begin{tikzpicture}[#1]%
\draw (0,0.3ex) -- (1ex,0.3ex);%
\draw (0.5ex,0.3ex) circle (0.2ex);
\draw[white] (0.2ex,0) -- (0.5ex,0);%
\end{tikzpicture}%
}
\newcommand{\Cbar}{\markovC{scale=2}}

\newcommand{\interior}[1]{\text{int}\hspace{-0.01cm}\big( #1 \big)}

\theoremstyle{remark}	\newtheorem{theorem}{Theorem}
\theoremstyle{remark}	\newtheorem{lemma}[theorem]{Lemma}
\theoremstyle{remark}	\newtheorem{coro}[theorem]{Corollary}
\theoremstyle{remark}	\newtheorem{proposition}[theorem]{Proposition}
\theoremstyle{remark} \newtheorem{definition}{Definition}
\theoremstyle{remark} \newtheorem{remark}{Remark}
\theoremstyle{remark} \newtheorem{example}{Example}

\newcommand{\channel}{W_{Y|X,S}}
\newcommand{\avc}{\Wset}																		
\newcommand{\opC}{\mathbb{C}}																
\newcommand{\inC}{\mathsf{C}}															 	
\newcommand{\inR}{\mathsf{R}}

\newcommand{\pSpace}{\mathcal{P}}														

\newcommand{\dB}{\mathsf{B}}	
\newcommand{\dE}{\mathsf{E}}																
	
\newcommand{\dM}{\mathsf{M}}															 	

\newcommand{\enc}{f}																				
\newcommand{\dec}{g}																			 	
\newcommand{\code}{\mathscr{C}}															
\newcommand{\gcode}{\mathscr{C}^{\,\Gamma}}									

\newcommand{\cerr}{P_{e|s^n}^{(n)}}													
\newcommand{\err}{P_e^{(n)}}															
\newcommand{\cost}{\phi}																		
\newcommand{\plimit}{\Omega}																			
\newcommand{\tset}{\Aset^{\delta}}													
\newcommand{\Tset}{\mathcal{T}}												
\newcommand{\qn}{q}
\newcommand{\tQ}{\hat{\Qset}_n}														

\newcommand{\rstarC}{																			  
\, \hspace{-0.3cm} \text{ $$ \mbox{  
\hspace{-0.1cm} 
\small $\star$   
} $$ }
\hspace{-0.25cm}}

\newcommand{\apLSpaceS}{\overline{\pSpace}_\Lambda(\Sset)}			
\newcommand{\pLSpaceS}{\overline{\pSpace}_\Lambda(\Sset)}			
\newcommand{\pLSpaceSn}{\pSpace_\Lambda(\Sset^n)}		

\newcommand{\mac}{W_{Y|X_1,X_2,S}}
\newcommand{\avmac}{\mathscr{A}}
\newcommand{\Mcompound}{\avmac^\Qset} 

\newcommand{\MrCav}{\opC^{\rstarC}\hspace{-0.1cm}(\avmac)} 			
\newcommand{\MCavc}{\opC(\avmac)} 					  																	

\newcommand{\MrICav}{\inC^{\rstarC}\hspace{-0.1cm}(\avmac)} 		
\newcommand{\MICavc}{\inC(\avmac)} 																							

\newcommand{\MrCcompound}{\opC^{\rstarC}\hspace{-0.1cm}(\avmac^\Qset)} 			
\newcommand{\MCcompound}{\opC(\avmac^\Qset)} 					  																	

\newcommand{\MICcompound}{\inC(\avmac^\Qset)} 																							

\newcommand{\cMICavc}{\overline{\inR}(\avmac)} 		

\newcommand{\MCavcf}{\opC(\avmac_{\text{free}})} 					  																	

\newcommand{\MdrCav}{\opC^{\rstarC \rstarC}\hspace{-0.1cm}(\avmac)} 			
\newcommand{\MdrICav}{\inC^{\rstarC \rstarC}\hspace{-0.1cm}(\avmac)} 		

\newcommand{\MdrCcompound}{\opC^{\rstarC \rstarC}\hspace{-0.1cm}(\avmac^\Qset)} 			

\newcommand{\MdrCavf}{\opC^{\rstarC \rstarC}\hspace{-0.1cm}(\avmac_{\text{free}})} 			
\newcommand{\MdrICavf}{\inC^{\rstarC \rstarC}\hspace{-0.1cm}(\avmac_{\text{free}})} 		

\begin{document}
\maketitle

{}

\begin{abstract} 
We determine both the random code capacity region and the deterministic code capacity region of the arbitrarily varying multiple access channel (AVMAC) under input and state constraints. The underlying assumption is that zero-rate transmission can be arbitrary in the deterministic setting as well, where there is no shared randomness.
As opposed to the random code capacity region, the deterministic code capacity region can be non convex.
For the AVMAC without constraints, the characterization due to Ahlswede and Cai is complete except for two cases, pointed out in the literature as an open problem. The missing piece is obtained as a special case of our results.
\end{abstract}


\blfootnote{
This work was supported by the Israel Science Foundation (grant No. 1285/16).
}

\section{Introduction}
The arbitrarily  varying multiple access channel (AVMAC) without constraints was first considered by Jahn \cite{Jahn:79t,Jahn:81p}, to describe a communication network with unknown statistics, that may change over time.
It  is especially relevant to uplink communication in the presence of an 
adversary, or a \emph{jammer}, attempting to disrupt communication. Another scenario is that one of multiple users becomes adversarial and attacks the other users
\cite{SangwanBakshiDeyPrabhakaran:19a, SangwanBakshiDeyPrabhakaran:19c}.  
Jahn  established the `divided-randomness capacity region' \cite{Jahn:79t,Jahn:81p}, namely the capacity region achieved when each encoder shares randomness with the decoder independently, 
and showed that the AVMAC inherits 
 some of the properties of its single user counterpart. In particular, the divided-randomness capacity region is not necessarily achievable using deterministic codes \cite{BBT:60p}. Furthermore, 
 Jahn showed that 
the deterministic code capacity region  either coincides with the 
divided-randomness capacity region or else, it has an empty interior \cite{Jahn:79t,Jahn:81p}. This phenomenon is an analogue of Ahlswede's dichotomy property \cite{Ahlswede:78p}. 
Therefore, in order to calculate the deterministic code capacity region, it is essential to confirm that the capacity region has a non empty interior, \ie positive rates are achievable using deterministic codes. Gubner \cite{Gubner:90p} presented three computable conditions which are necessary for this to hold, and conjectured that they are also sufficient.
Then, Ahlswede and Cai \cite{AhlswedeCai:99p} confirmed Gubner's conjecture \cite{Gubner:90p}, implying that Gubner's conditions are both necessary and sufficient for a non empty capacity region.
As Wiese and Boche recognized,
the case where exactly one of the users has zero capacity has remained an open problem
\cite[Remark 9]{WieseBoche:13p}, until now.

Furthermore, 
 constraints are known to have a drastic effect on the behavior of the single user AVC
\cite{
CsiszarNarayan:88p}, while the effect on the AVMAC has never been established.
%
 Csisz{\'{a}}r and Narayan \cite{
CsiszarNarayan:88p} considered the single user AVC when input and state constraints are imposed on the user and the jammer, respectively. 
Such constraints are often due to 
power limitations of the transmitter and the jamming signal.
Not only the constrained setting provokes serious technical difficulties analytically, but also, as shown in \cite{CsiszarNarayan:88p}, there is 
a significant effect on the behavior of the deterministic code capacity. Specifically, it is shown in \cite{CsiszarNarayan:88p} that dichotomy in the notion of \cite{Ahlswede:78p} no longer holds when state constraints are imposed on the jammer. That is, the deterministic code capacity can be lower than the random code capacity, and yet non-zero. As for the AVMAC under constraints, Gubner and Hughes \cite{GubnerHughes:95p} determined the divided-randomness capacity region. Results on the Gaussian AVMAC were recently presented in a talk \cite{HosseinigokiKosut:19c}.
Solved  examples can be found in \cite{Gubner:91p,Gubner:92p} as well.

Other relevant settings include the AVMAC with conferencing encoders  \cite{WieseBoche:13p,Wiese:13z,BocheSchaefer:14c}, list codes \cite{Nitinawarat:13p,BocheSchaefer:14c,Cai:16p}, fading \cite{ShafieeUlukus:05c,Budkuley:15c}, and an eavesdropper \cite{HeKhistiYener:13p,ArumugamBloch:16c,ChouYener:17c}. 
Among the models of channel uncertainty are also the compound multiple access channel 
\cite{CsiszarKorner:82b,MYK:05c,WBBJ:11p,ZAAA:14c,MitraAsnaniPillai:19c} and the random parameter multiple access with side information
\cite{DasNarayan:00c,CemalSteinberg:05p,SomekhBaruchShamaiVerdu:08p,LapidothSteinberg:13p1,LapidothSteinberg:13p2}. After the publication of this work, Sangwan \etal \cite{SangwanBakshiDeyPrabhakaran:19a} considered a multiple access channel with three users, where one of the users is possibly adversarial, yet the identity of the jamming user is not known in advance (see also \cite{SangwanBakshiDeyPrabhakaran:19c}).

In this work, we consider the AVMAC when input and state constraints are imposed on the users and the jammer, respectively. 
We give full characterization for  both the random code capacity region and the deterministic code capacity region. 
The underlying assumption is that zero-rate transmission can be arbitrary in the deterministic setting as well, where there is no shared randomness.
In particular, 
the encoder can simulate a random transmission,  as long as there is no shared randomness between the parties. This assumption is generally considered to be natural for real-life communication systems (see \eg remark in \cite[p. 282]{CsiszarKorner:82b}), 
whereas shared randomness 
is often impractical \cite{Jahn:81p,BBCM:95p}. 
Nonetheless, the analysis without shared randomness is a lot more challenging.
When state constraints are imposed, the operational time sharing argument does not apply to the AVMAC.  Roughly speaking, using a code over a part of the blocklength effectively increases the state constraint and loosens the restriction on the jammer over this period of time.
Thus, it becomes essential to replace the operational time sharing argument with coded time sharing \cite{HanKobayashi:81p}. 
Our decoder is then  a coded time sharing variant of Ahlswede and Cai's decoding rule \cite{AhlswedeCai:99p},
while the time sharing sequence is deterministic and known to the jammer as well.
Yet, a fundamental difference between our coding scheme and the one in \cite{AhlswedeCai:99p} arises from the dichotomy discrepancy, since
Ahlswede and Cai only showed achievability of positive rates $R_1=R_2=\eps>0$, proving that the capacity region has a non-empty interior.
However, for the AVMAC under constraints, dichotomy does not apply and achievability of positive rates is insufficient. Hence, in our problem, proving achievability is more demanding. Hereby, the codebooks construction and the analysis are based on generalization of the techniques by Csisz\'ar and Narayan \cite{CsiszarNarayan:88p}, along with the insights of Ahlswede and Cai \cite{AhlswedeCai:99p}. The converse proof involves observations by Gubner \cite{Gubner:90p} as well.
As a special case, we obtain a full characterization of the capacity region of the AVMAC without constraints, filling the gap left by 
Ahlswede and Cai \cite{AhlswedeCai:99p}.

\section{Definitions and Previous Results}
\label{sec:Mnotation}
 We use the following notation conventions throughout. 
Calligraphic letters $\Xset,\Sset,\Yset,...$ are used for finite sets.
Lowercase letters $x,s,y,\ldots$  stand for constants and values of random variables, and uppercase letters $X,S,Y,\ldots$ stand for random variables.  
 The distribution of a random variable $X$ is specified by a probability mass function (pmf) 
	$P_X(x)=p(x)$ over a finite set $\Xset$. The set of all pmfs over $\Xset$ is denoted by $\pSpace(\Xset)$.
 We use $x^j=(x_1,x_{2},\ldots,x_j)$ to denote  a sequence of letters from $\Xset$. 
 A random sequence $X^n$ and its distribution $P_{X^n}(x^n)$ are defined accordingly. 
The type $\hP_{x^n}$ of a given sequence $x^n$ is defined as the empirical distribution $\hP_{x^n}(a)=N(a|x^n)/n$ for $a\in\Xset$, where $N(a|x^n)$ is the number of occurrences of the symbol $a$ in the sequence $x^n$.
A type class is denoted by $\Tset^n(\hP)=\{ x^n \,:\; \hP_{x^n}=\hP \}$.
For a pair of integers $i$ and $j$, $1\leq i\leq j$, we define the discrete interval $[i:j]=\{i,i+1,\ldots,j \}$. 
In the continuous case, we use the cumulative distribution function 
	$F_Z(z)=\prob{Z\leq z}$ for $z\in\mathbb{R}$, or alternatively, the probability density function (pdf) $f_Z(z)$,  when it exists. 

	\subsection{Channel Description}
	\label{subsec:Mchannels}
 A state-dependent discrete memoryless multiple access channel (MAC) 
$(\Xset_1\times\Xset_2\times\Sset,\mac,\Yset)$ consists of  finite input alphabets $\Xset_1$ and $\Xset_2$, state alphabet $\Sset$, output alphabet $\Yset$,  and a 
conditional pmf 
$\mac$ over $\Yset$. The channel is memoryless without feedback, and therefore   
$W_{Y^n|X_1^n,X_2^n,S^n}(y^n|x_1^n,x_2^n,s^n)= \prod_{i=1}^n \mac(y_i|x_{1,i},x_{2,i},s_i)$. 
The AVMAC is a MAC  with a state sequence of unknown distribution,  not necessarily independent nor stationary. That is, $S^n\sim \qn(s^n)$ with an unknown joint pmf $\qn(s^n)$ over $\Sset^n$. In particular, $\qn(s^n)$ could give mass $1$ to some state sequence $s^n$.
The AVMAC 
is denoted by $\avmac=\{\mac\}$.

The compound MAC is used as a tool in the analysis.  Different models of compound MACs are described in the literature \cite{CsiszarKorner:82b,MYK:05c}. Here,   the compound MAC is  a channel with a discrete memoryless state, where the state distribution $q(s)$ is not known in exact, but rather belongs to a family of distributions $\Qset$, with $\Qset\subseteq \pSpace(\Sset)$. That is, the state sequence $S^n$ is independent and identically distributed (i.i.d.) according to $ q(s)$, for some pmf $q\in\Qset$. 
We note that this differs from the classical definition of the compound channel, as in \cite{CsiszarKorner:82b}, where the state is fixed throughout the transmission.
The compound MAC is denoted by $\Mcompound$.

\subsection{Coding}
\label{subsec:Mcoding}
We introduce some preliminary definitions, starting with the definitions of a deterministic code and a random code for the AVMAC $\avmac$ under input and state constraints. 

\begin{definition}[Code] 
\label{def:Mcapacity}
A $(2^{nR_1},2^{nR_2},n)$ code for the AVMAC $\avmac$ 
consists of the following;   
two message sets $[1:2^{nR_1}]$ and $[1:2^{nR_2}]$, where $2^{nR_1}$ and $2^{nR_2}$ are assumed to be integers, two encoding functions $\enc_1:[1:2^{nR_1}]\rightarrow \Xset_1^n$ and $\enc_2:[1:2^{nR_2}]\rightarrow \Xset_2^n$, and a decoding function
$
\dec: \Yset^n\rightarrow [1:2^{nR_1}]\times [1:2^{nR_2}]  
$. 

Given a pair of messages $m_1\in [1:2^{nR_1}]$ and  $m_2\in [1:2^{nR_2}]$,
 Encoder $k$ transmits the codeword $x^n_k=\enc_k(m_k)$, for $k=1,2$. 
The decoder  receives the channel output $y^n$, and finds an estimate of the message pair $(\hm_1,\hm_2)=g(y^n)$.  We denote the code by $\code=\left(\enc_1(\cdot),\enc_2(\cdot),\dec(\cdot) \right)$.
\end{definition}

We proceed now to coding schemes 
when using stochastic-encoders stochastic-decoder pairs with common randomness.
We distinguish between two classes; random codes \cite{WieseBoche:13p} and divided-randomness codes \cite{Jahn:81p}.

\begin{definition}[Random code]
\label{def:McorrC} 
A $(2^{nR_1},2^{nR_2},n)$ random code for the AVMAC $\avmac$ consists of a collection of 
$(2^{nR_1},2^{nR_2},n)$ codes $\{\code_{\gamma}=(\enc_{1,\gamma},\enc_{2,\gamma},\dec_\gamma)\}_{\gamma\in\Gamma}$, along with a probability distribution $\mu(\gamma)$ over the code collection $\Gamma$. 
We denote such a code by $\gcode=(\mu,\Gamma,\{\code_{\gamma}\}_{\gamma\in\Gamma})$.
\end{definition}

\begin{definition}[Divided-randomness code]
\label{def:DivMcorrC} 
A $(2^{nR_1},2^{nR_2},n)$ divided-randomness code for the AVMAC $\avmac$ is a random code, where the random element consists of two components, \ie $\gamma=(\gamma_1,\gamma_2)$, one at each encoder. 
The components are drawn according to a product distribution $\mu(\gamma_1,\gamma_2)=\mu_1(\gamma_2)\mu_2(\gamma_2)$ over $\Gamma_1\times\Gamma_2$.
Then, User 1  sends $x_1^n=f_{1,\gamma_1}(m_1)$, User 2 sends $x_2^n=f_{2,\gamma_2}(m_2)$, and upon receiving the channel output $y^n$, the receiver applies a decoding mapping $g_{\gamma_1,\gamma_2}$. We denote such a code by $\code^{\Gamma_1\times\Gamma_2}$.
\end{definition}

The general  random code in Definition~\ref{def:McorrC} can thus be thought of as a divided-randomness code where statistical dependence between $\gamma_1$ and $\gamma_2$ is viable. 

\begin{remark}
\label{remark:stochE}
Our underlying assumption is that zero-rate transmission can be arbitrary in the deterministic setting as well, as long as there is no shared randomness.
%
In particular, if User 1 has zero capacity while User 2 transmits at a positive rate, then Encoder 1 may transmit a random sequence at zero rate, \ie $x_1^n=f_1(\sigma)$ where $\sigma\in [1:2^{n\eps}]$ is a random parameter,  which is not known to the other encoder, the decoder, nor the jammer, and the decoder is not required to recover the value of $\sigma$. 
This means that the encoders have access to $n\eps$ random bits, where $\eps>0$ is arbitrarily small.
Since the randomness is local, such an assumption is generally considered to be reasonable (see \eg remark in \cite[p. 282]{CsiszarKorner:82b}). 
On the other hand, in the codes in Definitions \ref{def:McorrC} and \ref{def:DivMcorrC}, there is shared randomness between the encoders and the decoder, which 
is often impractical \cite{Jahn:81p,BBCM:95p}. 

One may also consider the AVMAC with stochastic encoders, \ie when Encoder $k$ transmits $x_k^n=f_k(m_k,\sigma_k)$,
for $m_k\in [1:2^{nR_k}]$, $k=1,2$, where $\sigma_1$ and $\sigma_2$ the random parameters . Then, 
our results apply to the stochastic encoder capacity region. 
\end{remark}

\subsection{Input and  State Constraints} 
\label{subsec:Mconstraints}
Next, we consider input constraints and state constraint, imposed on the encoders and the jammer, respectively.
We note that the constraints specifications are known to both users and the jammer in this model.
Let $\cost_k:\Xset_k\rightarrow [0,\infty)$, $k=1,2$, and $l:\Sset\rightarrow [0,\infty)$ be some given bounded functions, and define
	\begin{align}
	\cost_k^n(x_k^n)=&\frac{1}{n} \sum_{i=1}^n \cost_k(x_{k,i}) \,,\; k=1,2 \,,
	\label{eq:MLInConstraintStrict} \\
	l^n(s^n)=&\frac{1}{n} \sum_{i=1}^n l(s_i) \,.
	\end{align}
Let $\plimit_1>0$, $\plimit_2>0$, and $\Lambda>0$. Below, we specify the input constraints $(\plimit_1,\plimit_2)$ and state constraint $\Lambda$, corresponding to  the functions
$\cost_1^n(x_1^n)$, $\cost_2^n(x_2^n)$, and $l^n(s^n)$, respectively, for the AVMAC and the compound MAC.
%
	
	%
	
	Given input constraints $(\plimit_1,\plimit_2)$, the encoding functions need to satisfy 
	\begin{align}
	\cost^n_k(\enc_k(m_k))\leq\plimit_k \,,\;\text{for all $m_k\in [1:2^{nR_k}]$} \,,\; k=1,2.
	\label{eq:MinputCstrict}
	\end{align}
	That is, the inputs satisfy $\cost^n_1(X_1^n)\leq\plimit_1$ and $\cost^n_2(X_2^n)\leq\plimit_2$ with probability $1$. 
	Moving to the state constraint $\Lambda$, we have different definitions for the AVMAC and for the compound MAC.
	
	The compound MAC has a constraint on average, with a memoryless state such that $\E_q l(S)\leq\Lambda$, while the AVMAC has an almost-surely constraint, with a non-stationary state sequence such that $l^n(S^n)\leq\Lambda$ with probability $1$. 
	Explicitly, we say that a compound MAC $\Mcompound$  is under a state constraint $\Lambda$, if the set $\Qset$ of state distributions is limited to $\Qset \subseteq \pLSpaceS$, where
	\begin{align}
	\pLSpaceS&\triangleq 
\{ q(s)\in\pSpace(\Sset) \,:\; \E_q\, l(S)\leq\Lambda \} \,.
\label{eq:pLSpaceS}
\intertext{ 	
%
As for the AVMAC $\avmac$, 
 it is now assumed that the joint distribution of the state sequence is limited to $q(s^n)\in\pLSpaceSn$, where
}
\pLSpaceSn &\triangleq\{ q(s^n) \in\pSpace(\Sset^n) \,:\; q(s^n)=0 \;\text{ if $l^n(s^n)>\Lambda$}\, \} \,.
\end{align}
This includes the case of a deterministic unknown state sequence, \ie when $q$ gives probablity $1$ to a particular $s^n\in\Sset^n$ with
$l^n(s^n)\leq \Lambda$.

		We may assume without loss of generality that $0\leq \plimit_k\leq \cost_{k,max}$, $k=1,2$, and $0\leq\Lambda\leq l_{max}$, where $\cost_{k,max}=\max_{x_k\in\Xset_k} \cost_k(x_k)$, $k=1,2$, and $l_{max}=\max_{s\in\Sset} l(s)$. It is also assumed that for some $a\in\Xset_1$, $b\in\Xset_2$, and $s_0\in\Sset$, $\cost_1(a)=\cost_2(b)=l(s_0)=0$.

\subsection{Capacity Region Under Constraints}
We move to the definition of  achievable rate pairs and the capacity region of the AVMAC $\avmac$  under input and state constraints.
Deterministic codes and random codes over the AVMAC $\avmac$ are defined as in 
Definition~\ref{def:Mcapacity} and Definition~\ref{def:McorrC},
 respectively, with the additional constraint 
 (\ref{eq:MinputCstrict}) on the codebook.

 Define the conditional probability of error of a code $\code$ given a state sequence $s^n\in\Sset^n$ by  
\begin{subequations}
\begin{align}
\label{eq:Mcerr}
&\cerr(\code)\triangleq 
\frac{1}{2^{ n(R_1+R_2) }}\sum_{m_1=1}^{2^ {nR_1}}\sum_{m_2=1}^{2^ {nR_2}}
\sum_{y^n:\dec(y^n)\neq (m_1,m_2)} W_{Y^n|X_1^n,X_2^n,S^n}(y^n|\enc_1(m_1),\enc_2(m_2),s^n) \,.
\end{align}
Now, define the average probability of error of $\code$ for some distribution $\qn(s^n)\in\pSpace(\Sset^n)$, 
\begin{align}
\err(\qn,\code)\triangleq 
\sum_{s^n\in\Sset^n} \qn(s^n)\cdot\cerr(\code) \,.
\end{align}
\end{subequations}

\begin{definition}[
Achievable rate pair and capacity region under constraints]
\label{def:MLcapacity}
A code $\code=(\enc_1,\enc_2,\dec)$ is a  called a
$(2^{nR_1},2^{nR_2},n,\eps)$ code for the AVMAC $\avmac$ under input constraints $(\plimit_1,\plimit_2)$ and  state constraint $\Lambda$, when (\ref{eq:MinputCstrict})  is satisfied 
 and 
\begin{align}
\label{eq:MLerr}
& \err(q,\code) 
\leq \eps \,,\quad
\text{for all $q\in\pLSpaceSn$} \,,
\end{align}
or, equivalently, $\cerr(\code)\leq\eps$ for all $s^n\in\Sset^n$ with $l^n(s^n)\leq\Lambda$.

  We say that a rate pair $(R_1,R_2)$ is achievable, under	input constraints $(\plimit_1,\plimit_2)$ and state constraint $\Lambda$,
	if for every $\eps>0$ and sufficiently large $n$, there exists a  $(2^{nR_1},2^{nR_2},n,\eps)$ code for the AVMAC	$\avmac$ under input constraints $(\plimit_1,\plimit_2)$ and state constraint $\Lambda$. The operational capacity region is defined as the closure of the set of achievable rate pairs, 
	and it is denoted by $\MCavc$. 
 We use the term `capacity region' referring to this operational meaning, and in some places we call it the deterministic code capacity region in order to emphasize that achievability is measured with respect to  deterministic codes.

Analogously to the deterministic case,  a $(2^{nR_1},2^{nR_2},n,\eps)$ random code $\gcode=$ $(\mu,\Gamma,$ $\{\code_{\gamma}\}_{\gamma\in\Gamma})$
 for the AVMAC $\avmac$, under input constraints $(\plimit_1,\plimit_2)$ and state constraint $\Lambda$, satisfies the requirements
\begin{subequations}
\label{eq:MLrcodeReq}
\begin{align}
&
\sum_{\gamma\in\Gamma}\mu(\gamma)   \cost_k^n(\enc_{k,\gamma}^n(m_k))
  \leq \plimit_k
\,,\; 
\text{for all $m_k\in [1:2^{nR_k}]$}\,,\; k=1,2 \,,  \label{eq:McodeInputCr}
\intertext{and} 
&\err(q,\gcode)\triangleq \sum_{\gamma\in\Gamma} \mu(\gamma) \err(q,\code_\gamma)
\leq \eps \,,\quad\text{for all $q\in\pLSpaceSn$} \,.
\label{eq:MLrerr}				
\end{align}
\end{subequations}
The capacity region achieved by random codes is then denoted by $\MrCav$, and it 
 is referred to as the \emph{random code capacity region}.
%
In addition, a $(2^{nR_1},2^{nR_2},n,\eps)$ divided-randomness code  $\code^{\Gamma_1\times\Gamma_2}$ 
satisfies the requirements
\begin{subequations}
\label{eq:DivMLrcodeReq}
\begin{align}
&
\sum_{\gamma_k}\mu_k(\gamma_k)   \cost_k^n(\enc_{k,\gamma_k}^n(m_k))
  \leq \plimit_k
\,,\; 
\text{for all $m_k\in [1:2^{nR_k}]$}\,,\; k=1,2 \,,  \label{eq:DivMcodeInputCr}
\intertext{and} 
&\err(q,\code^{\Gamma_1\times\Gamma_2})\triangleq \sum_{\gamma_1,\gamma_2} \mu_1(\gamma_1) \mu_2(\gamma_2) 
\err(q,\code_{\gamma_1,\gamma_2})
\leq \eps \,,\quad\text{for all $q\in\pLSpaceSn$} \,.
\label{eq:vMLrerr}				
\end{align}
\end{subequations}
The capacity region achieved by divided-randomness  codes is then denoted by $\MdrCav$, and it 
 is referred to as the \emph{divided-randomness capacity region}.
\end{definition}
 
Note that based on the definitions above,
\begin{align}
\MCavc \subseteq \MdrCav\subseteq \MrCav \,.
\end{align}


The definitions above are naturally extended to the compound MAC, under input constraints 
$(\plimit_1,\plimit_2)$ and state constraint $\Lambda$, by limiting the requirements (\ref{eq:MinputCstrict}), (\ref{eq:MLerr}) and (\ref{eq:MLrcodeReq}) to i.i.d. state distributions $q\in\Qset$. 
The respective deterministic code capacity region, random code capacity region, and divided-randmoness capacity region 
 $\MCcompound$, $\MrCcompound$ and $\MdrCcompound$ are defined accordingly.

\subsection{Related Work}
\subsubsection{Without Constraints} 
\label{sec:MLNOsi}
In this subsection, we briefly review known results for the case where there are no constraints.   
Denote the deterministic code capacity region and the divided-randomness capacity regions of the AVMAC free of constraints by $\MCavcf$ and $\MdrCavf$, respectively. 
We note that this is a special case of the AVMAC under constraints, with $\plimit_1\geq\cost_{1,max}$, $\plimit_2\geq\cost_{2,max}$,  and $\Lambda\geq l_{max}$.

We cite the divided-randomness capacity theorem of the AVMAC free of constraints, due to Jahn 
\cite{Jahn:79t}. Let 
\begin{align}
\label{eq:MICf}
\inC^{\rstarC\rstarC}\hspace{-0.1cm}(\avmac_{\text{free}})=\bigcup_{P_U P_{X_1|U} P_{X_2|U}}
\left\{
\begin{array}{lrl}
(R_1,R_2) \,:\; & R_1 		\leq&   \min_{q(s|u)} I_q(X_1;Y|X_2,U)  \,, \\
								& R_2 		\leq&   \min_{q(s|u)} I_q(X_2;Y|X_1,U)  \,, \\
								& R_1+R_2 \leq&   \min_{q(s|u)} I_q(X_1,X_2;Y|U)  
\end{array}
\right\}
 \,,
\end{align}
with $(U,X_1,X_2,S)\sim P_U(u) P_{X_1|U}(x_1|u) P_{X_2|U}(x_2|u) q(s|u)$. 

\begin{theorem}[see {\cite[Theorem 1a]{Jahn:79t}}] 
\label{theo:MavcC0R}
The divided-randomness capacity region of an AVMAC, free of constraints,  is  given by 
\begin{align}
\MdrCavf = \inC^{\rstarC\rstarC}\hspace{-0.1cm}(\avmac_{\text{free}}) \,.
\end{align}
\end{theorem}

\begin{remark}
Originally, the random code capacity region is expressed in \cite[Theorem 1]{Jahn:79t} as a closed convex hull of a union of regions.
Achievability of (\ref{eq:MICf}) is established through time sharing. 
\end{remark}

Now, we move to the deterministic code capacity region. 
\begin{theorem}[Ahlswede's Dichotomy \cite{Jahn:79t,Jahn:81p}]  
\label{theo:MavcC0}
The capacity region of an AVMAC, free of  constraints, either coincides with the divided-randomness capacity region or else, it has an empty interior.
That is, 
$
\MCavcf = \MdrCavf  
$ or else, 
$
\interior{\MCavcf}=\emptyset 
$. 
\end{theorem}
Necessary and sufficient conditions for the capacity region to have a non-empty interior were 
 established by Gubner \cite{Gubner:90p,Gubner:88z} and Ahlswede and Cai \cite{AhlswedeCai:99p}  in terms of the following definition.
\begin{definition} \cite{Gubner:90p,Gubner:88z,Ericson:85p}
\label{def:Msymmetrizable}
 A state-dependent MAC $\mac$ is said to be 
\begin{enumerate}[1)]
\item
 \emph{symmetrizable}-$\Xset_1\times\Xset_2$ if for some conditional distribution $J(s|x_1,x_2)$,
\begin{multline}
\label{eq:MsymmetrizableJ}
\sum_{s\in\Sset} \mac(y|x_1,x_2,s)J(s|\tx_1,\tx_2)=\sum_{s\in\Sset} \mac(y|\tx_1,\tx_2,s)J(s|x_1,x_2) \,,\; \\
\forall\, x_1,\tx_1\in\Xset_1 \,,\; x_2,\tx_2\in\Xset_2 \,,\; y\in\Yset \,.
\end{multline}
Equivalently,  the channel $\widetilde{W}(y|x_1,x_2,\tx_1,\tx_2)$ $=$ $
\sum_{s\in\Sset} \channel(y|x_1,x_2,s)J(s|\tx_1,\tx_2)$ is symmetric with respect to $(x_1,x_2)$ and 
$(\tx_1,\tx_2)$.

\item
 \emph{symmetrizable}-$\Xset_1|\Xset_2$ if for some conditional distribution $J_1(s|x_1)$,
\begin{align}
\label{eq:Msymmetrizable1}
\sum_{s\in\Sset} \mac(y|x_1,x_2,s)J_1(s|\tx_1)=\sum_{s\in\Sset} \mac(y|\tx_1,x_2,s)J_1(s|x_1) \,,\;  
\forall\, x_1,\tx_1\in\Xset_1 \,,\; x_2\in\Xset_2 \,,\; y\in\Yset \,.
\end{align}

\item
 \emph{symmetrizable}-$\Xset_2|\Xset_1$ if for some conditional distribution $J_2(s|x_2)$,
\begin{align}
\label{eq:Msymmetrizable2}
\sum_{s\in\Sset} \mac(y|x_1,x_2,s)J_2(s|\tx_2)=\sum_{s\in\Sset} \mac(y|x_1,\tx_2,s)J_2(s|x_2) \,,\; 
\forall\, x_1\in\Xset_1 \,,\; x_2,\tx_2\in\Xset_2 \,,\; y\in\Yset \,.
\end{align}
%
	\end{enumerate} 
\end{definition}
We say that the AVMAC $\avmac$ is symmetrizable-$\Xset_1\times\Xset_2$ if the corresponding state-dependent MAC $\mac$ is symmetrizable-$\Xset_1\times\Xset_2$, and similarly for symmetrizability -$\Xset_1|\Xset_2$ and symmetrizability-$\Xset_2|\Xset_1$ .
%
\begin{example} \cite{Gubner:90p}
Consider an adder channel specified by $Y=X_1+X_2+S$, where $\Xset_1=\Xset_2=\{0,1 \}$.
For $\Sset=\{0,1,2\}$, the AVMAC satisfies the conditions in Definition~\ref{def:Msymmetrizable}, 
as (\ref{eq:MsymmetrizableJ})-(\ref{eq:Msymmetrizable2}) hold with $J(s|x_1,x_2)=\delta(s-x_1-x_2)$, 
$J_1(s|x_1)=\delta(s-x_1)$, $J_2(s|x_2)=\delta(s-x_2)$, where $\delta(u)$ is the Kronecker delta function, \ie $\delta(u)=1$ for $u=0$, and $\delta(u)=0$ otherwise. On the other hand, it is shown in \cite{Gubner:90p} that for 
$\Sset=\{0,1\}$, the AVMAC is symmetrizable-$\Xset_1|\Xset_2$ and symmetrizable-$\Xset_2|\Xset_1$, but non-symmetrizable-$\Xset_1\times\Xset_2$.
\end{example}
Ahlswede and Cai \cite{AhlswedeCai:99p} showed by example that
it is also possible that an AVMAC satisfies the first condition in Definition~\ref{def:Msymmetrizable} but does not satisfy the other two. 
\begin{example} \cite{AhlswedeCai:99p}
\label{example:MsymmJn1n2}
Consider a binary MAC, with $\Xset_1=\Xset_2=\Sset=\Yset=\{0, 1\}$,
 specified by the following. For $s=0$,
\begin{align}
&\mac(\cdot|0,0,0)=\mac(\cdot|1,1,0)=(1,0) \,,\nonumber\\
&\mac(\cdot|1,0,0)=\mac(\cdot|0,1,0)=\left(  \frac{1}{2},\frac{1}{2} \right) \,,
\intertext{and for $s=1$,}
&\mac(\cdot|0,0,1)=\mac(\cdot|1,1,1)=\left(  \frac{1}{2},\frac{1}{2} \right) \,,\nonumber\\
&\mac(\cdot|1,0,1)=\mac(\cdot|0,1,1)=  (0,1)\,.
\end{align}
Then it is shown in \cite{AhlswedeCai:99p} that $\mac$ is symmetrizable-$\Xset_1\times\Xset_2$, as 
(\ref{eq:MsymmetrizableJ}) holds for $J(s|x_1,x_2)=1$ for $(x_1=x_2,s=0)$ or $(x_1\neq x_2,s=1)$, and 
$J(s|x_1,x_2)=0$ otherwise. On the other hand, plugging $x_2=0$ in (\ref{eq:Msymmetrizable1}) yields 
$J_1(s|x_1)=\delta(s-x_1)$, while plugging $x_2=1$ in (\ref{eq:Msymmetrizable1}) yields 
$J_1(s|x_1)=\delta(s-(1-x_1))$.
This means that fixing $x_2$, both marginals $\mac(\cdot|\cdot,0,\cdot)$ and $\mac(\cdot|\cdot,1,\cdot)$ of User 1 are symmetrizable in the single-user sense, \ie as in \cite[Definition 2]{CsiszarNarayan:88p}. 
However, the AVMAC is not symmetrizable-$\Xset_1|\Xset_2$, because there is no $J_1(s|x_1)$ which symmetrizes both marginals at the same time. 
 In a similar manner, (\ref{eq:Msymmetrizable2}) implies a contradiction as well. 
 Therefore, the AVMAC is symmetrizable-$\Xset_1\times\Xset_2$, but non-symmetrizable-$\Xset_2|\Xset_1$ and non-symmetrizable-$\Xset_1|\Xset_2$. 
\end{example}

Intuitively, symmetrizability-$\Xset_1\times\Xset_2$ identifies a poor channel, where 
the jammer can impinge the communication scheme by randomizing the state sequence $S^n$ according to $J^n(s^n|\tx_1^n,\tx_2^n)=\prod_{i=1}^n J(s_i|\tx_{1,i},\tx_{2,i})$, for some codewords $\tx_1^n$ and $\tx_2^n$ in the codebooks of User 1 and User 2, respectively. 
 Suppose that the transmitted codewords are $x_1^n$ and $x_2^n$. The codewords $\tx_1^n$ and $\tx_2^n$ can be thought of as impostors transmitted by the jammer.  Now, since  the ``average channel" $\widetilde{W}$ is symmetric with respect to $(x_1^n,x_2^n)$ and  $(\tx_1^n,\tx_2^n)$, the codeword pairs appear to the receiver as equally likely. Similarly, if the AVMAC is symmetrizable-$\Xset_1|\Xset_2$, then the decoder  
confuses between $(x_1^n,x_2^n)$ and  $(\tx_1^n,x_2^n)$, and if the AVMAC is symmetrizable-$\Xset_2|\Xset_1$, then the decoder confuses between $(x_1^n,x_2^n)$ and  $(x_1^n,\tx_2^n)$, where $\tx_1$ and $\tx_2$ are the codewords chosen by the jammer.
 Indeed, by \cite{Gubner:90p}, if one of the conditions in Definition~\ref{def:Msymmetrizable} holds, then the capacity region of the AVMAC free of constraints has an empty interior. This means that the capacity region $\MCavcf$ is either an interval or $\{(0,0)\}$, \ie  one of the users or both have zero capacity.
 
Ahlswede and Cai \cite{AhlswedeCai:99p} proved that the three types of non-symmetrizability are 
not only a necessary condition for a non-empty capacity region, but they are sufficient conditions as well. This yields the following theorem.
\begin{theorem}[see \cite{Gubner:90p,Gubner:88z}{\cite[Theorem 1]{AhlswedeCai:99p}}]
\label{theo:Msymm0}
An AVMAC free of constraints has a capacity region with a non-empty interior, \ie $\interior{\MCavcf}\neq\emptyset$, if and only if it is non-symmetrizable-$\Xset_1\times\Xset_2$, non-symmetrizable-$\Xset_1|\Xset_2$, and non-symmetrizable-$\Xset_2|\Xset_1$. 
\end{theorem}


The following theorem 
combines the results by Gubner \cite{Gubner:90p,Gubner:88z} and Ahlswede and Cai \cite{AhlswedeCai:99p}.
This statement was also given by Boche and Wiese \cite{WieseBoche:13p,Wiese:13z}, who considered the AVMAC with conferencing encoders. 
\begin{theorem}[see {\cite[Theorem 8]{WieseBoche:13p}}]
\label{theo:WieseBoche}
There are four scenarios for the capacity region of the AVMAC free of constraints:
\begin{enumerate}[a)]
\item
If $\mac$ is not symmetrizable-$\Xset_1\times\Xset_2$, -$\Xset_1|\Xset_2$, nor -$\Xset_2|\Xset_1$, then
\begin{align}
\MCavcf=\MdrICavf \,.
\end{align}

\item
If $\mac$ is not symmetrizable-$\Xset_1\times\Xset_2$ nor -$\Xset_2|\Xset_1$, but symmetrizable-$\Xset_1|\Xset_2$,  then 
\begin{align}
\MCavcf\subseteq \left\{(0,R_2) \,:\; R_2\leq \min_{q(s)} \max_{P_{X_1} P_{X_2}} 
I_q(X_2;Y|X_1) \right\} \,.
\label{eq:MCavcf2}
\end{align}

\item
If $\mac$ is not symmetrizable-$\Xset_1\times\Xset_2$ nor -$\Xset_1|\Xset_2$, but symmetrizable-$\Xset_2|\Xset_1$,  then 
\begin{align}
\MCavcf\subseteq \left\{(R_1,0) \,:\; R_1\leq \min_{q(s)} \max_{P_{X_1} P_{X_2}} 
I_q(X_1;Y|X_2) \right\} \,.
\label{eq:MCavcf3}
\end{align}

\item
In all other cases, 
\begin{align}
\MCavcf=\left\{(0,0) \right\} \,.
\end{align}
\end{enumerate}
\end{theorem}

\begin{remark}
\label{remark:WieseBoche}
Observe that in Case b) and Case c) of Theorem~\ref{theo:WieseBoche}, the characterization is incomplete.
As pointed out by Wiese and Boche, this has remained an open problem for nearly 20 years
\cite[Remark 9]{WieseBoche:13p} (see also \cite[Remark 5.6]{Wiese:13z}).
At first glance, it may appear as if achievability in Cases b) and c) immediately follows from the capacity theorem of the single user AVC \cite{CsiszarNarayan:88p}. Consider Case c), and denote the channel from $X_1$ to $Y$, for a fixed $x_2\in\Xset_2$, by 
$W^{(x_2)}_{Y|X_1,S}=\mac(\cdot|\cdot,x_2,\cdot)$. Then, based on the results by Csisz\'{a}r and Narayan 
for the single user AVC \cite{CsiszarNarayan:88p}, if $W^{(x_2)}_{Y|X_1,S}$ is non-symmmetrizable-$\Xset_1$ for  some $x_2\in\Xset_2$, then the capacity of User 1 is positive. Furthermore, if $W^{(x_2)}_{Y|X_1,S}$ is non-symmetrizable for all $x_2\in\Xset_2$, then User 1 can achieve every rate 
$R_1<\min_{q(s)} \max_{P_{X_1} P_{X_2}} I_q(X_1;Y|X_2)$. However, in Case c), knowing that 
the AVMAC is non-symmetrizable-$\Xset_1|\Xset_2$ does not guarantee that $W^{(x_2)}_{Y|X_1,S}$ is non-symmetrizable for all $x_2\in\Xset_2$. Actually, it only guarantees that the channels $W^{(x_2)}_{Y|X_1,S}$,
$x_2\in\Xset_2$, are not all symmetrized by a single  $J_1(s|x_1)$ (see Example~\ref{example:MsymmJn1n2}). 
 Therefore, it is not immediately clear whether the conditions in Cases b) and c) are sufficient for achievability.
 We are going to fill this gap and show that (\ref{eq:MCavcf2})
and (\ref{eq:MCavcf3}) hold with equality.
\end{remark}

\begin{remark}
The dichotomy property in Theorem~\ref{theo:MavcC0} was proved using Ahlswede's Elimination Technique \cite{Ahlswede:78p}, where 
the encoder transmits the random elements $\gamma_1$ and $\gamma_2$ over a negligible portion of the blocklength.
The Elimination Technique only works without state constraints \cite{CsiszarNarayan:88p}, since positive capacity under a state constraint does not guarantee reliable transmission over a fraction of the blocklength.
Moreover, as Csisz{\'{a}}r and Narayan demonstrated in the single user setting,  the dichotomy property does not hold when state constraints are imposed on the jammer. That is, the deterministic code capacity can be lower than the capacity with shared randomness, even if positive rates are achievable with deterministic codes.
This demonstrates the significant effect that constraints have on the behavior of the deterministic code capacity of arbitrarily varying channels.
\end{remark}

\subsubsection{Divided-Randomness Capacity Region}
Gubner and Hughes \cite{GubnerHughes:95p} considered the AVMAC under input and state constraints, and determined the divided-randomness capacity region, \ie assuming  each encoder shares an independent random element with the decoder. Their result is given below. Define
\begin{align}
\label{eq:MdrICav}
\inC^{\rstarC\rstarC}\hspace{-0.1cm}(\avmac)=\bigcup_{ \substack{ P_U P_{X_1|U} P_{X_2|U} \,: \\ \E \cost_k(X_k)\leq \plimit_k \,,\; k=1,2\,. }}
\left\{
\begin{array}{lrl}
(R_1,R_2) \,:\; & R_1 		\leq&   \min\limits_{q(s|u) \,:\; \E_q l(S)\leq\Lambda} I_q(X_1;Y|X_2,U)  \,, \\
								& R_2 		\leq&   \min\limits_{q(s|u) \,:\; \E_q l(S)\leq\Lambda} I_q(X_2;Y|X_1,U)  \,, \\
								& R_1+R_2 \leq&   \min\limits_{q(s|u) \,:\; \E_q l(S)\leq\Lambda} I_q(X_1,X_2;Y|U)  
\end{array}
\right\}
 \,,
\end{align}
with $(U,X_1,X_2,S)\sim P_U(u) P_{X_1|U}(x_1|u) P_{X_2|U}(x_2|u) q(s|u)$. 
It is shown in \cite{GubnerHughes:95p} that the region above is not necessarily convex.
In Remark~\ref{rem:rCavConvex}, we discuss the interpretation of this property and the connection to the statistical independence between
the variables $S$ and $U$ above.

\begin{theorem}[see {\cite{GubnerHughes:95p}}]
\label{theo:GubnerHughes}
The divided-randomness capacity region of the AVMAC under input constraints $(\plimit_1,\plimit_2)$ and state constraint $\Lambda$ is given by 
\begin{align}
\MdrCav=\inC^{\rstarC\rstarC}\hspace{-0.1cm}(\avmac) \,.
\end{align}
\end{theorem}

In the next sections, we determine both the random code capacity region and the deterministic code capacity region.

\section{Main Results -- Random Code Capacity Region}
In this section, we establish the random code capacity region of the AVMAC under input and state constraints.
To this end, we first give an auxiliary result on the compound MAC.

\subsection{The Compound MAC}
We begin with the capacity theorem for the compound MAC $\Mcompound$ under input constraints $(\plimit_1,\plimit_2)$ and state constraint $\Lambda$. 
This is an auxiliary result, obtained by a simple extension of related work (see \cite{MYK:05c}).
Let
\begin{align}
\MICcompound=
\bigcup_{ \substack{ P_U P_{X_1|U} P_{X_2|U} \,: \\ \E \cost_k(X_k)\leq \plimit_k \,,\; k=1,2\,. } }
\left\{
\begin{array}{lrl}
(R_1,R_2) \,:\; & R_1 		\leq&   \inf_{q\in\Qset} I_q(X_1;Y|X_2,U)  \,, \\
								& R_2 		\leq&   \inf_{q\in\Qset} I_q(X_2;Y|X_1,U)  \,, \\
								& R_1+R_2 \leq&   \inf_{q\in\Qset} I_q(X_1,X_2;Y|U)  
\end{array}
\right\}
 \,.
\end{align}
with $(U,X_1,X_2,S)\sim P_U(u) P_{X_1|U}(x_1|u) P_{X_2|U}(x_2|u) q(s)$.

\begin{lemma}
\label{lemm:MCcompound}
The capacity region of the compound MAC $\Mcompound$ is given by
\begin{align}
\MCcompound=\MICcompound \,,
\end{align}
and it is identical to the divided-randomness capacity region and the random code capacity region, \ie $\MrCcompound=\MdrCcompound=\MCcompound$. 
\end{lemma}
The proof of Lemma~\ref{lemm:MCcompound} is given in Appendix~\ref{app:MCcompound}. 

\begin{remark}
\label{rem:TimSharCompound}
Regardless of the statement in Lemma~\ref{lemm:MCcompound}, the capacity region of the compound MAC must be convex, due to the operational time sharing argument.
That is, if  
$(R_{1,u},R_{2,u})$, $u\in\Uset$, are achievable rate pairs, then any convex combination 
\begin{align}
\left( \sum_{u\in\Uset} \theta_u R_{1,u},\sum_{u\in\Uset} \theta_u R_{2,u} \right)
\end{align}
 is achievable, 
for $\theta_u\geq 0$, $\sum_{u\in\Uset}\theta_u=1$. To achieve this rate pair, one can employ a sequence of consecutive codes 
that achieve $(R_{1,u},R_{2,u})$, such that $\theta_u$ is the fraction of the corresponding code length
from the total blocklength (see \cite[Section 15.3.3]{CoverThomas:06b}).
In the classical setting, the random variable $U$ is referred to as the time sharing variable, since $P_U(u)$ can be interpreted as the coefficient $\theta_u$ in the convex combination. We explain below why this interpretation is lacking in the case of the compound MAC.
Furthermore, we will see that operational time sharing is impossible for the AVMAC, yet the convexity of $\MICcompound$ will play a role (see the remarks below Theorem~\ref{theo:MrCav}).

For the classical MAC, achievability of the capacity region can be established by first considering independent inputs $(X_1,X_2)\sim P_{X_1}(x_1)P_{X_2}(x_2)$, and then generalizing to 
$P_U(u)P_{X_1|U}(x_1|u)P_{X_2|U}(x_2|u)$ through the operational time sharing argument. 
However, for the compound MAC, a straightforward application of the operational time sharing argument is insufficient,  because
\begin{align}
\inf_{q\in\Qset} I_q(X_1,X_2;Y|U)\geq \sum_{u\in\Uset} p(u) \cdot \inf_{q\in\Qset} I_q(X_1,X_2;Y|U=u) \,,
\label{eq:McompConv}
\end{align}
in general, and similarly for $I_q(X_1;Y|X_2,U)$ and $I_q(X_2;Y|X_1,U)$. Hence, achievability of the convex combination in the RHS of (\ref{eq:McompConv}) does \emph{not} immediately imply achievability of the LHS.

We deduce that the external variable $U$ may not represent the \emph{operational} time sharing strategy as in the classical sense. Nevertheless, we associate $U$  with \emph{coded} time sharing \cite{HanKobayashi:81p} \cite[Section 4.5.3]{ElGamalKim:11b}. Specifically, to prove Lemma~\ref{lemm:MCcompound}, we use a coded time sharing scheme,  where a time sharing sequence $U^n$ is generated, and then 
a single codebook is
 selected accordingly.
\end{remark}

\subsection{The AVMAC}
We determine the random code capacity region  of the AVMAC under input constraints $(\plimit_1,\plimit_2)$ and state constraint $\Lambda$.
As opposed to Theorem~\ref{theo:GubnerHughes} by \cite{GubnerHughes:95p}, considering divided-randomness coding, we address the  case where  the three parties, two encoders and decoder, share randomness together.
The random code derivation is based on our result on the compound MAC and a simple extension of Ahlswede's RT. 

Define $\MrICav\triangleq \MICcompound \big|_{\Qset=\pLSpaceS} 
$, \ie 
\begin{align}
\MrICav=
\bigcup_{ \substack{ P_U P_{X_1|U} P_{X_2|U} \,: \\ \E \cost_k(X_k)\leq \plimit_k \,,\; k=1,2\,. } }
\left\{
\begin{array}{lrl}
(R_1,R_2) \,:\; & R_1 		\leq&   \min\limits_{q(s) \,:\; \E_q l(S)\leq\Lambda} I_q(X_1;Y|X_2,U)  \,, \\
								& R_2 		\leq&   \min\limits_{q(s) \,:\; \E_q l(S)\leq\Lambda} I_q(X_2;Y|X_1,U)  \,, \\
								& R_1+R_2 \leq&   \min\limits_{q(s) \,:\; \E_q l(S)\leq\Lambda} I_q(X_1,X_2;Y|U)  
\end{array}
\right\}
 \,.
\label{eq:MrICav}
\end{align}
with $(U,X_1,X_2,S)\sim P_U(u) P_{X_1|U}(x_1|u) P_{X_2|U}(x_2|u) q(s)$. 
%
Notice the resemblance between the random code capacity region formula (\ref{eq:MrICav}) and the divided-randomness capacity region formula (\ref{eq:MdrICav}), as the only difference between the formulas is that the state $S$ and the ``time-sharing" variable $U$   are statistically independent in the random code case, while  $S$ and $U$ are dependent in the divided-randomness case.
An intuitive interpretation is given in the remark. 
\begin{theorem}
\label{theo:MrCav}
The random code capacity region of the AVMAC $\avmac$ under input constraints $(\plimit_1,\plimit_2)$ and state constraint $\Lambda$ is given by 
\begin{align}
\MrCav=\MrICav \,.
\end{align}
\end{theorem}
The proof of Theorem~\ref{theo:MrCav} is given in Appendix~\ref{app:MrCav}. The proof is based on the aforementioned result on the compound MAC and an extension of Ahlswede's Robustification Technique \cite{Ahlswede:86p}. Essentially, we use a reliable code for the compound MAC to construct a random code for the AVMAC by  applying random permutations to each codeword symbols. 

\begin{remark}
\label{rem:TimSharAVMAC}
As opposed to the compound channel, the operational time sharing argument is not eligible for the AVC under a state constraint, even in the single user case, as we explain below.
Suppose that two codebooks are used in time sharing, one of rate $R'$ and length $\theta n$, and one of rate $R''$ and length $(1-\theta)n$, for $0<\theta<1$, where $R'$ and $R''$ are both achievable for the AVC under a state constraint $\Lambda$.
Due to the state constraint, it is guaranteed that $\sum_{i=1}^n l(S_i) \leq n\Lambda$ with probability $1$.
However, the jammer is entitled to concentrate the jamming power on the first $\theta n$ symbols,
in which case, $\sum_{i=1}^{\theta n} l(S_i) = n\Lambda$.
If the jammer does so, then the first code, of rate $R'$ and length $\theta n$, needs to be robust against a state sequence $S^{\theta n}$ with the following cost,
\begin{align}
\frac{1}{\theta n} \sum_{i=1}^{\theta n} l(S_i)=\frac{\Lambda}{\theta} >\Lambda \qquad\text{a.s.}
\end{align}
Therefore, in order to achieve the rate $R=\theta R'+(1-\theta)R''$ with operational time sharing, the first coding rate $R'$ needs to be achievable for an AVC under a state constraint $\Lambda/\theta$, and the second coding rate $R''$ needs to be achievable for a state constraint $\Lambda/(1-\theta)$, which is not guaranteed.
Hence, operational time sharing is not eligible for neither the single user AVC, nor the AVMAC (see also \cite{GubnerHughes:95p}).

Nevertheless, we observe that using the proof technique in Appendix~\ref{app:MrCav},   one can devise a reliable coding scheme with ``shuffled time sharing". 
 That is, instead of using the codes  of rates $R'$ and $R''$ above consecutively,  random interleaving of the codes can be realized.
Intuitively, the users are thus able to carry out a time sharing protocol of which the jammer is oblivious,  which explains the statistical independence between the state $S$ and the time sharing variable $U$ in (\ref{eq:MrICav}).
\end{remark}

\begin{remark}
\label{rem:rCavConvex}
As mentioned above, the difference between the formulas given for the random code capacity region and for the divided-randomness capacity region is the statistical independence between the state and the time sharing variable.
Now, we discuss the implications in terms of the convexity of the  regions.

By Theorem~\ref{theo:MrCav} and Lemma~\ref{lemm:MCcompound}, we have that the random code capacity region of the AVMAC under input and state constraints is the same as that of the compound MAC with
$\Qset=\pLSpaceS$. As a consequence, 
we have that the random code capacity region of the AVMAC is convex. It can also be verified directly that the set in the RHS of (\ref{eq:MrICav}) is convex, as $(U,S)\sim P_U(u)q(s)$.

On the other hand, Gubner and Hughes \cite{GubnerHughes:95p} demonstrated that  the divided-randomness capacity region given by  (\ref{eq:MdrICav})
 is \emph{not} necessarily convex, as $(U,S)\sim P_U(u)q(s|u)$  (see Section IV in \cite{GubnerHughes:95p}). In their setting, the encoders have statistically independent random elements 
$\gamma_1$ and $\gamma_2$ (see Definition~\ref{def:DivMcorrC}).
  Gubner and Hughes attribute the non-convexity to the preclusion of operational time sharing \cite{GubnerHughes:95p}.
	It is further mentioned in \cite{GubnerHughes:95p} that there are other instances of non-convex capacity regions in the literature, such as the asynchronous MAC \cite{HuiHumblet:85p,Poltyrev:83p}, where the users' timeframes do not synchronize hence time sharing does not work either.
We observe that the ``shuffled time sharing" mentioned in the previous remark could only work if the shared randomness element is exploited for the coordination between the users. 
Whereas, in a scenario where the random elements are independent, as in \cite{GubnerHughes:95p}, such coordination is impossible.

In our setting, the random elements are not independent, as $\gamma_1=\gamma_2=\gamma$ (\cf Definition~\ref{def:McorrC} and Definition~\ref{def:DivMcorrC}).
In the remark that follows Definition~2 in \cite{GubnerHughes:95p}, it is stated without proof that removing the restriction of independence 
could result in a strictly larger capacity region. Indeed, our result above that $\MrCav$ is convex implies that for the erasure AVMAC in \cite[Section IV]{GubnerHughes:95p}, our random code capacity region $\MrCav$  must be strictly larger than the non-convex divided-randomness capacity region. 
Therefore, we have now validated the assertion by Gubner and Hughes \cite{GubnerHughes:95p}.
%
Comparing (\ref{eq:MdrICav}) and (\ref{eq:MrICav}), we infer that the conditioning of the state distribution on $U$ may lead to a strictly smaller region. Intuitively, knowing the time sharing protocol helps the jammer reduce the coding rates for such a channel.
\end{remark}

\section{Main Results -- Deterministic Code Capacity Region}
The principal result of this paper is the deterministic code capacity theorem, \ie without shared randomness.
The deterministic code derivation is independent of our previous results, and the analysis modifies the techniques of Csisz\'{a}r and Narayan \cite{CsiszarNarayan:88p}, and merges their ideas  with those of Ahlswede and Cai \cite{AhlswedeCai:99p}.

Before we state the capacity theorem, we give the following definitions.
Given an input distribution $P_{X_1,X_2}\in\pSpace(\Xset_1\times\Xset_2)$, consider the average state costs below,
\begin{subequations}
\label{eq:Mtheta}
\begin{align}
\Psi(P_{X_1,X_2})=&
\min_{ \text{symm. $J$}}\sum_{ \substack{x_k\in\Xset_k  \\ k=1,2}   } \sum_{s\in\Sset} 
P_{X_1,X_2}(x_1,x_2) J(s|x_1,x_2) l(s) \,, 
\label{eq:MthetaJ}
\\
\Psi_1(P_{X_1})=&\min_{ \text{symm. $J_1$}}\sum_{x_1\in\Xset_1} \sum_{s\in\Sset} 
P_{X_1}(x_1) J_1(s|x_1) l(s) \,, \label{eq:MthetaJ1}\\
\Psi_2(P_{X_2})=&\min_{ \text{symm. $J_2$}}\sum_{x_2\in\Xset_2} \sum_{s\in\Sset} 
P_{X_2}(x_2) J_2(s|x_2) l(s) \,, \label{eq:MthetaJ2}
\end{align}
\end{subequations}
where the minimizations are over $J(s|x_1,x_2)$, $J_1(s|x_1)$, and $J_2(s|x_2)$, which satisfy the symmetrizing conditions in (\ref{eq:MsymmetrizableJ}), (\ref{eq:Msymmetrizable1}), and (\ref{eq:Msymmetrizable2}), respectively. We use the convention that a minimum over an empty set is $+\infty$.
Then, 
for every $P_{U,X_1,X_2}\in\pSpace(\Uset\times\Xset_1\times\Xset_2)$, 
define
\begin{subequations}
\label{eq:Mtlambda}
\begin{align}
\tLambda(P_{U,X_1,X_2})=& \sum_{u\in\Uset} P_U(u) \Psi(P_{X_1,X_2|U=u})
 \,, 
\label{eq:MtlambdaJ}
\\
\tLambda_1(P_{U,X_1})=& \sum_{u\in\Uset} P_U(u) \Psi_1(P_{X_1|U=u})
\,, \label{eq:MtlambdaJ1}\\
\tLambda_2(P_{U,X_2})=& \sum_{u\in\Uset} P_U(u) \Psi(P_{X_1,X_2|U=u})
\,. \label{eq:MtlambdaJ2}
\end{align}
\end{subequations}
Intuitively, $\min\{\tLambda(P_{U,X_1,X_2})$, $\tLambda_1(P_{U,X_1})$, $\tLambda_2(P_{U,X_2})\}$ is the minimal average state cost which the jammer has to pay to symmetrize the channel, for a given inputs distribution $P_{X_1,X_2|U}$,  where symmetrizing refers to using a conditional distribution that satisfies either one of the symmetrizability conditions in Definition~\ref{def:Msymmetrizable}. If this minimal state cost violates the state constraint $\Lambda$, then the jammer is prohibited from symmetrizing the channel. 

\begin{remark}
\label{rem:MLambdaJeq}
The minimal average state cost can be expressed more explicitly as
\begin{align}
\tLambda(P_{U,X_1,X_2})=&
\min_{ \text{symm. $\{J_u \}$}}\sum_{u,x_1,x_2,s} P_U(u) P_{X_1,X_2|U}(x_1,x_2|u) J_u(s|x_1,x_2) l(s) \,, 
\label{eq:MLambdaJeq}
\end{align}
with minimization over a set of distribution $\{ J_u \}_{u\in\Uset}$, where each distribution $J_u(s|x_1,x_2)$ 
symmetrizes-$\Xset_1\times\Xset_2$ the AVMAC. Notice that the state distributions are indexed by the time sharing variable.
 This has the interpretation of a jamming scheme that varies over time in accordance with the time sharing sequence chosen by the users.
\end{remark}

%
We have defined $\tLambda(P_{U,X_1,X_2})$, $\tLambda_1(P_{U,X_1})$ and $\tLambda_2(P_{U,X_2})$
in (\ref{eq:Mtlambda}) as the minimal average state costs which the jammer has to pay to symmetrize the channel, for a given input distribution $P_{U,X_1,X_2}$.
Intuitively, the users are interested in restricting the jammer by increasing those costs as much as possible, hence the following quantities represent the best thresholds the users can obtain, 
\begin{align}
L^*\triangleq& \max_{P_{U,X_1,X_2} \,:\; \E\cost_k(X_k)\leq\plimit_k ,\, k=1,2} \tLambda(P_{U,X_1,X_2}) \,,
\label{eq:Lstar}
\\
L_1^*\triangleq& \max_{P_{U,X_1} \,:\; \E\cost_1(X_1)\leq\plimit_1 } \tLambda_1(P_{U,X_1}) \,,
\label{eq:Lstar1}
\\
L_2^*\triangleq& \max_{P_{U,X_2} \,:\; \E\cost_2(X_2)\leq\plimit_2 } \tLambda_2(P_{U,X_2}) \,,
\label{eq:Lstar2}
\end{align}
We note that $L^*$, $L_1^*$ and $L_2^*$ depend on the input constraints $(\plimit_1,\plimit_2)$, the MAC
$\mac$, and the state cost function $l:\Sset\rightarrow [0,\infty)$, but they do not depend on the state constraint $\Lambda$. 

\begin{definition}
\label{def:MICavc}
Define the rate region $\MICavc$ as follows. 
\begin{subequations}
\label{eq:MICavc}
\begin{enumerate}[a)]
\item
If $L^*>\Lambda$, $L_1^*>\Lambda$, and $L_2^*>\Lambda$, then
\begin{align}
\MICavc=
\bigcup_{ 
\overline{\pSpace}_{\plimit_1,\plimit_2,\Lambda}(\Uset\times\Xset_1\times\Xset_2) 
}
\left\{
\begin{array}{lrl}
(R_1,R_2) \,:\; & R_1 		\leq&   \min\limits_{q(s|u) \,:\; \E_q l(S)\leq\Lambda} I_q(X_1;Y|X_2,U)  \,, \\
								& R_2 		\leq&   \min\limits_{q(s|u) \,:\; \E_q l(S)\leq\Lambda} I_q(X_2;Y|X_1,U)  \,, \\
								& R_1+R_2 \leq&   \min\limits_{q(s|u) \,:\; \E_q l(S)\leq\Lambda} I_q(X_1,X_2;Y|U)  
\end{array}
\right\}
 \,,
\label{eq:Mcasea}
\end{align}
where 
\begin{align}
\overline{\pSpace}_{\plimit_1,\plimit_2,\Lambda}(\Uset\times\Xset_1\times\Xset_2)=
\{& P_U P_{X_1|U} P_{X_2|U} \,:\; 
\E \cost_k(X_k)\leq \plimit_k \,,\; k=1,2\,,\;\text{and }
\nonumber\\
&\min\{ \tLambda(P_{U,X_1,X_2}),\tLambda_1(P_{U,X_1}),\tLambda_2(P_{U,X_2}) \} \geq \Lambda
\} \,.
\label{eq:MpI}
\end{align}

\item
If $L^*>\Lambda$, $L_2^*>\Lambda$, but $L_1^*\leq\Lambda$, then
\begin{align}
\MICavc=\{ (0,R_2) \,:\; R_2\leq \min_{q(s)\,:\; \E_q l(S)\leq\Lambda} \;\,
\max_{ \substack{ P_{X_1} P_{X_2} \,:\;  \E \cost_k(X_k)\leq \plimit_k \,,\; k=1,2\,,
\\
\min\{ \tLambda(P_{X_1}P_{X_2}),\tLambda_2(P_{X_2})  \} \geq \Lambda
 } }\, I_q(X_2;Y|X_1)
 \} \,.
\label{eq:Mcaseb}
\end{align}

\item
If $L^*>\Lambda$, $L_1^*>\Lambda$, but $L_2^*\leq\Lambda$, then
\begin{align}
\MICavc=\{ (R_1,0) \,:\; R_1\leq \min_{q(s)\,:\; \E_q l(S)\leq\Lambda} \;\,
\max_{ \substack{ P_{X_1}P_{X_2} \,:\;  \E \cost_k(X_k)\leq \plimit_k \,,\; k=1,2\,,
\\
\min\{ \tLambda(P_{X_1}P_{X_2}),\tLambda_1(P_{X_1})  \} \geq \Lambda
 } }\, I_q(X_1;Y|X_2)
 \} \,.
\label{eq:Mcasec}
\end{align}

\item
Otherwise, if $L^*\leq\Lambda$, or if both $L_1^*\leq\Lambda$ and $L_2^*\leq\Lambda$, then
\begin{align}
\MICavc=\{ (0,0) \} \,. 
\end{align}
\label{eq:Mcased}
\end{enumerate}

\end{subequations}
\end{definition}
\begin{remark}
\label{rem:MICdetSymmR}
Observe that the optimization set of the input distribution $P_U P_{X_1|U} P_{X_2|U}$ in (\ref{eq:MICavc}) is a subset of the corresponding set in (\ref{eq:MdrICav}), for the divided-randomness capacity region, hence $\MICavc\subseteq\MdrICav$.
Furthermore, if the AVMAC is non-symmetrizable in the sense of neither 
$\Xset_1\times\Xset_2$, $\Xset_1|\Xset_2$, nor $\Xset_2|\Xset_1$, then
$\tLambda(P_{U,X_1,X_2})=\tLambda_1(P_{U,X_1})=\tLambda_2(P_{U,X_2})=+\infty$, in which case we have that 
$\MICavc=\MdrICav$.
\end{remark}

\begin{theorem}
\label{theo:MCavc}
Assume that $L^*\neq\Lambda$, $L_1^*\neq\Lambda$ and $L_2^*\neq\Lambda$.
Then, the capacity region of the AVMAC $\avmac$ under input constraints $(\plimit_1,\plimit_2)$ and state constraint $\Lambda$ is given by
\begin{align}
\MCavc=\MICavc \,.
\label{eq:MCavcTheo}
\end{align}
Furthermore, if $\avmac$ is non-symmetrizable-$\Xset_1\times\Xset_2$, non-symmetrizable-$\Xset_1|\Xset_2$, and non-symmetrizable-$\Xset_2|\Xset_1$, then the capacity region coincides with the divided-randomness capacity region, \ie $\MCavc=\MdrCav=\MdrICav$.
\end{theorem}
The proof of Theorem~\ref{theo:MCavc} is given in Appendix~\ref{app:MCavc}.
The second part of the theorem follows from Remark~\ref{rem:MICdetSymmR} and (\ref{eq:MCavcTheo}).
The proof does not use our results on the random code capacity region of the AVMAC and on the compound MAC, and it is independent of the divided-randomness analysis by Gubner and Hughes \cite{GubnerHughes:95p}.
The analysis, however, makes use of the properties established for the decoding rule and codebooks specified below,
in Subsections \ref{subsec:MDec} and \ref{subsec:Mcodebooks}.
As mentioned, coded time sharing is an essential replacement for the classical operational time sharing argument, which cannot be applied to the AVMAC under constraints (see Remark~\ref{rem:TimSharAVMAC}).
Hence, our analysis combines our coded time sharing variant of the decoder by Ahlswede and Cai
\cite{AhlswedeCai:99p}, with our generalization of the codebook generated by Csisz\'{a}r and Narayan \cite{CsiszarNarayan:88p}. The converse proof further uses Gubner's observations in \cite{Gubner:90p}.

\begin{remark}
\label{rem:SingleDiffM}
As explained in Remark~\ref{remark:WieseBoche} for the AVMAC without constraints, the case where one of the users has zero capacity does not immediately follow from the results on the single user AVC.
The reason behind this is that Gubner's second and third conditions are stronger than single-user symmetrizability, as defined in \cite[Definition 2]{CsiszarNarayan:88p}. In particular, if the AVMAC is
 symmetrized-$\Xset_1|\Xset_2$ by $J_1(s|x_1)$, then the marginal AVC
$W_{Y|X_1,S}$ is also symmetrized by $J_1(s|x_1)$, but the other direction is not true. 
In the constrained setting, this means that the minimal state cost $\tLambda_1(P_{X_1})$ for symmetrizability-$\Xset_1|\Xset_2$ can be higher than the minimal state cost $\tLambda_0(P_{X_1})$ in 
\cite[Equation (2.13)]{CsiszarNarayan:88p}, for symmetrizability of the marginal AVC $W_{Y|X_1,S}$.
In Example~\ref{example:MsymmJn1n2}, we have seen that the AVMAC is non-symmetrizable-$\Xset_1|\Xset_2$, even though the marginal $W_{Y|X_1,S}$ could be symmetrizable, in which case,
 $\tLambda_1(P_{X_1})=+\infty$ but $\tLambda_0(P_{X_1})\leq l_{max}<\infty$. Therefore, even if $L_1^*>\Lambda$, it is not guaranteed that User 1 can achieve a positive rate.
\end{remark}

\begin{remark}
The boundary case where either $L^*=\Lambda$ or $L_k^*=\Lambda$, $k=1,2$, remains unsolved. Even in the single user setting, say $\Xset_2=\emptyset$, the case of $L_1^*=\Lambda$ is an open problem
 (see \cite{CsiszarNarayan:88p}), although it is conjectured in \cite{CsiszarNarayan:88p} that the capacity is zero in this case. Similarly, we conjecture that the capacity region is $\MCavc=\MICavc$ for all values of $L^*$, $L_1^*$ and $L_2^*$. There are special cases where we can prove that this holds, given in the corollary below. The corollary generalizes the remark following Theorem 3 in \cite{CsiszarNarayan:88p}.
\end{remark}

\begin{coro}
\label{coro:MCavc01}
Let $\avmac$ be an AVMAC under input constraints $(\plimit_1,\plimit_2)$ and state constraint $\Lambda$,
where $\avmac$ is symmetrizable-$\Xset_1\times\Xset_2$, symmetrizable-$\Xset_1|\Xset_2$, and
symmetrizable-$\Xset_2|\Xset_1$. If the symmetrizability equations (\ref{eq:MsymmetrizableJ}), (\ref{eq:Msymmetrizable1}), and (\ref{eq:Msymmetrizable2}) are only satisfied by conditional distributions $J(s|x_1,x_2)$, $J_1(s|x_1)$
$J_2(s|x_2)$ which are $0$-$1$ laws, then
\begin{align}
\MCavc=\MICavc \,.
\label{eq:MCavcTheoSp}
\end{align}
\end{coro}
The proof of Corollary~\ref{coro:MCavc01} is given in Appendix~\ref{app:MCavc01}.
In particular, we note that the condition of $0$-$1$ laws in Corollary~\ref{coro:MCavc01} holds when the
output $Y$ is a deterministic function of $X_1$, $X_2$ and $S$. 

\subsection{
Decoding Rule}
\label{subsec:MDec}
We specify the decoding rule and state the corresponding properties, which are used in the analysis. 
As mentioned in Remark~\ref{rem:TimSharAVMAC} above, we cannot use the operational time sharing argument for the AVMAC under constraints, and therefore, we use coded time sharing \cite{HanKobayashi:81p} \cite[Section 4.5.3]{ElGamalKim:11b}.
Our decoder is similar to that of Ahlswede and Cai \cite{AhlswedeCai:99p}, and the codebooks are generated based on the techniques of Csisz\'ar and Narayan \cite{CsiszarNarayan:88p}, along with the insights of Ahlswede and Cai \cite{AhlswedeCai:99p}. We note that since the code is deterministic, the time sharing sequence is also deterministic, and it is known to the encoders, the decoder, and the jammer as well.

A fundamental difference between our coding scheme and the one in \cite{AhlswedeCai:99p} arises from Ahlswede's dichotomy property. Specifically, Ahlswede and Cai only showed achievability of positive rates 
$R_1=R_2=\eps>0$, proving that the capacity region has a non-empty interior.
According to the dichotomy result in Theorem~\ref{theo:MavcC0} by Jahn \cite{Jahn:79t,Jahn:81p}, this implies that the capacity region of the AVMAC free of constraints is the same as the random code capacity region, which was also determined in \cite{Jahn:79t,Jahn:81p}. 
However, for the AVMAC under constraints, dichotomy does not apply and achievability of positive rates is insufficient. Thereby, our proof is a lot more involved than the one in  \cite{AhlswedeCai:99p}.

%

To specify the decoding rule, we define the decoding sets $\Dset(m_1,m_2)\subseteq \Yset^n$, for $(m_1,m_2)\in [1:2^{nR_1}]\times [1:2^{nR_2}]$, such that $g(y^n)=(m_1,m_2)$ iff $y^n\in\Dset(m_1,m_2)$.
\begin{definition}[Decoder]
\label{def:MLdecoder}
Given the codebooks $\{ \enc_k(m_k) \}_{m_k\in [1:2^{nR_k}]}$, $k=1,2$, and a time sharing sequence $u^n$,
 declare that $y^n\in \Dset(m_1,m_2)$ if there exists $s^n\in\Sset^n$ with $l^n(s^n)\leq\Lambda$ such that the following hold.
\begin{enumerate}[1)]
\item
For $(U,X_1,X_2,S,Y)$ which is distributed according to the joint type 
$\hP_{u^n,f_1(m_1),f_2(m_2),s^n,y^n}$, we have that 
\begin{align}
D(P_{U,X_1,X_2,S,Y}|| P_U\times P_{X_1|U}\times P_{X_2|U}\times P_{S|U} \times \mac )\leq \eta \,.
\end{align}
 
\item
\begin{enumerate}[a)]
\item
For every $\tm_1\neq m_1$ and $\tm_2\neq m_2$ such that for some $\ts^n\in\Sset^n$ with $l^n(\ts^n)\leq\Lambda$, 
\begin{align}
\label{eq:MDcompA} 
D(P_{U,\tX_1,\tX_2,\tS,Y}|| P_U\times P_{\tX_1|U}\times P_{\tX_2|U}\times P_{\tS|U} \times \mac)\leq \eta \,,
\end{align} 
where $(U,\tX_1,\tX_2,\tS,Y)\sim \hP_{u^n,f_1(\tm_1),f_2(\tm_2),\ts^n,y^n}$,
 we have that
\begin{align}
I(X_1,X_2,Y;\tX_1,\tX_2|U,S)\leq \eta  \,.
\end{align}

\item
For every $\tm_1\neq m_1$  such that for some $\ts^n\in\Sset^n$  with $l^n(\ts^n)\leq\Lambda$, 
\begin{align}
D(P_{U,\tX_1,X_2,\tS,Y}|| P_U\times P_{\tX_1|U}\times P_{X_2|U}\times P_{\tS|U} \times \mac )\leq \eta \,,\; 
\end{align}
where $(U,\tX_1,X_2,\tS,Y)\sim \hP_{u^n,f_1(\tm_1),f_2(m_2),\ts^n,y^n}$, we have that
\begin{align}
I(X_1,X_2,Y;\tX_1|U,S)\leq \eta_1  \,.
\end{align}

\item
For every $\tm_2\neq m_2$ such that for some $\ts^n\in\Sset^n$  with $l^n(\ts^n)\leq\Lambda$,
\begin{align}
D(P_{U,X_1,\tX_2,\tS,Y}|| P_U\times P_{X_1|U}\times P_{\tX_2|U}\times P_{\tS|U} \times \mac)\leq \eta \,, 
\end{align}
where $(U,X_1,\tX_2,\tS,Y)\sim \hP_{u^n,f_1(m_1),f_2(\tm_2),\ts^n,y^n}$,
 we have that
\begin{align}
I(X_1,X_2,Y;\tX_2|U,S)\leq \eta_2  \,.
\end{align}
\end{enumerate}

\end{enumerate}
\end{definition}
We note that in Definition~\ref{def:MLdecoder}, the variables $U,X_1,X_2, \tX_1,\tX_2,S,\tS,Y$ are dummy random variables, distributed according to the joint type of 
$(u^n,f_1(m_1),f_2(m_2),f_1(\tm_1),f_2(\tm_2),s^n,\ts^n,y^n)$, where
$u^n$ is a given time sharing sequence, $f_1(m_1),f_2(m_2)$ are ``tested" codewords, 
$f_1(\tm_1),f_2(\tm_2)$ are competing codewords, $s^n$ is a ``tested" state sequence, $\ts^n$ is a competing state sequence, and $y^n$ is the received sequence. None of the sequences are random here.
The Markov relation $U\Cbar (X_1,X_2,S)\Cbar Y$ may not hold for those dummy variables, and we may have that the conditional type $P_{Y|X_1,X_2,S}$ differs from the actual channel $ \mac$. Therefore, the divergences and mutual informations in Definition~\ref{def:MLdecoder} could be positive.

For the definition above to be proper, we need to verify that the decoding sets are disjoint, as stated in the following lemma.
\begin{lemma}[Decoding Disambiguity]
\label{lemm:MdisDec}
Let $u^n$ be a given time sharing sequence, and denote by $P_U$ its type, that is, $P_U=\hP_{u^n}$.
Suppose that in each codebook, all codewords have the same  conditional type, \ie $\hP_{\enc_1(m_1)|u^n}= P_{X_1|U}$ and $\hP_{\enc_2(m_2)|u^n}=P_{X_2|U}$ for all $(m_1,m_2)$.
 Assume that for some $\delta,\delta_k>0$, $P_U(u)\geq\delta$, $P_{X_k|U}(x_k|u)\geq\delta_k$ $\forall x_k\in\Xset_k$, $u\in\Uset$, $k=1,2$, and also
\begin{align}
\min\left\{  \tLambda(P_{U,X_1,X_2}) ,\,   \tLambda_1(P_{U,X_1}) ,\, \tLambda_2(P_{U,X_2})   \right\}  >\Lambda \,.
\label{eq:MdecLambda}
\end{align}
 Then, for sufficiently small $\eta,\eta_1,\eta_2>0$, 
\begin{align}
\Dset(m_1,m_2)\neq \Dset(\tm_1,\tm_2) \,,\;\text{for all $(m_1,m_2)\neq (\tm_1,\tm_2)$} \,.
\end{align}
Specifically,
\begin{enumerate}[1)]
\item
Conditions 1) and 2a) of the decoding rule, with $\tLambda(P_{U,X_1,X_2})>\Lambda$, imply that for sufficiently small $\eta$,
\begin{align}
\Dset(m_1,m_2)\cap \Dset(\tm_1,\tm_2)=\emptyset \,,\;\text{for $m_1\neq \tm_1$ and $m_2\neq \tm_2$} \,.
\label{eq:MdisAmb}
\end{align}
\item
Conditions 1) and 2b) of the decoding rule,  with $\tLambda_1(P_{U,X_1})>\Lambda$, imply that for sufficiently small $\eta$ and $\eta_1$,
\begin{align}
\Dset(m_1,m_2)\cap \Dset(\tm_1,m_2)=\emptyset \,,\;\text{for $m_1\neq \tm_1$} \,.
\label{eq:MdisAmb1}
\end{align}
\item
Conditions 1) and 2c) of the decoding rule, with $\tLambda_2(P_{U,X_2})>\Lambda$, imply that for sufficiently small 
$\eta$ and $\eta_2$,
\begin{align}
\Dset(m_1,m_2)\cap \Dset(m_1,\tm_2)=\emptyset \,,\;\text{for $m_2\neq \tm_2$} \,.
\label{eq:MdisAmb2}
\end{align}
\end{enumerate}
\end{lemma}
The proof of Lemma~\ref{lemm:MdisDec} is given in Appendix~\ref{app:MdisDec}.

\subsection{
Codebooks}
\label{subsec:Mcodebooks}
While the decoding rule above is similar to that of Ahlswede and Cai \cite{AhlswedeCai:99p},
here we prove a generalization of a lemma by Csisz\'{a}r and Narayan \cite{CsiszarNarayan:88p}, in order to generate proper codebooks. 
\begin{lemma}[Codebooks Generation]
\label{lemm:McodeBsets}
For every $\eps>0$, sufficiently large $n$, rates $R_k\geq \eps$ and types $P_U$ and $P_k=P_{X_k|U}$, $k=1,2$, there exist a time sharing sequence $u^n\in\Tset^n(P_U)$, and codebooks, 
$\{(x_1^n(m_1),x_2^n(m_2)):m_k\in [1:2^{nR_k}], k=1,2\}$ of type $P_1\times P_2$, 
such that for every $a_1^n\in\Xset_1^n$, $a_2^n\in\Xset_2^n$, $s^n\in\Sset^n$ with $l^n(s^n)\leq\Lambda$, and every joint type $P_{U,X_1,X_2,\tX_1,\tX_2,S}$ with $P_{X_1,X_2|U}=P_{\tX_1,\tX_2|U}=P_1\times P_2$, the following hold.
\begin{enumerate}[1)]
\item \emph{Joint Typicality}
\begin{align}
|\{ (\tm_1,\tm_2) \,:\; (u^n,a_1^n,a_2^n,x_1^n(\tm_1),x_2^n(\tm_2),s^n)\in\Tset^n(P_{U,X_1,X_2,\tX_1,\tX_2,S})  \}|
\leq
2^{n\left( \left[ R_1+R_2-I(\tX_1,\tX_2;X_1,X_2,S|U) \right]_{+} +\eps \right)} \,, 
\label{eq:M11ebn}
\end{align}
\begin{align}
|\{ (m_1,m_2) \,:\; (u^n,x_1^n(m_1),x_2^n(m_2),s^n)\in\Tset^n(P_{U,X_1,X_2,S})  \}|
\leq
2^{n\left(  R_1+R_2-\frac{\eps}{2}   \right)}
\,,\;\text{if $I(X_1,X_2;S|U)>\eps$} \,,
\label{eq:M12ebn}
\end{align}
and
\begin{multline}
|\{ (m_1,m_2) \,:\; (u^n,x_1^n(m_1),x_2^n(m_2),x_1^n(\tm_1),x_2^n(\tm_2),s^n)\in\Tset^n(P_{U,X_1,X_2,\tX_1,\tX_2,S})   \,,\;\text{for some $\tm_1\neq m_1$, $\tm_2\neq m_2$}
\}|\\
\leq
2^{n\left(  R_1+R_2-\frac{\eps}{2}   \right)}
\,,\; 
\text{if $I(X_1,X_2;\tX_1,\tX_2,S|U)-\left[ R_1+R_2-I(\tX_1,\tX_2;S|U) \right]_{+}>\eps$} \,.
\label{eq:M13ebn}
\end{multline}

\item \emph{Conditional Typicality Given $m_2$}
\begin{align}
&|\{ \tm_1 \,:\; (u^n,a_1^n,a_2^n,x_1^n(\tm_1),s^n)\in\Tset^n(P_{U,X_1,X_2,\tX_1,S})  \}|
\leq
2^{n\left( \left[ R_1-I(\tX_1;X_1,X_2,S|U) \right]_{+} +\eps \right)} \,, 
\label{eq:M2b1ebn}
\intertext{and}
&|\{ m_1 \,:\; (u^n,x_1^n(m_1),a_2^n,x_1^n(\tm_1),s^n)\in\Tset^n(P_{U,X_1,X_2,\tX_1,S})   \;\text{for some $\tm_1\neq m_1$}
\}| 
\leq
2^{n\left(  R_1-\frac{\eps}{2}   \right)}
\,,\; \nonumber\\&
\text{if $I(X_1,X_2;\tX_1,S|U)-\left[ R_1-I(\tX_1;S|U) \right]_{+}>\eps$} \,.
\label{eq:M2b3ebn}
\end{align}

\item \emph{Conditional Typicality Given $m_1$}
\begin{align}
&|\{ \tm_2 \,:\; (u^n,a_1^n,a_2^n,x_2^n(\tm_2),s^n)\in\Tset^n(P_{U,X_1,X_2,\tX_2,S})  \}|
\leq
2^{n\left( \left[ R_2-I(\tX_2;X_1,X_2,S|U) \right]_{+} +\eps \right)} \,, 
\label{eq:M2c1ebn}
\intertext{
and
}
&|\{ m_2 \,:\; (u^n,a_1^n,x_2^n(m_2),x_2^n(\tm_2),s^n)\in\Tset^n(P_{U,X_1,X_2,\tX_2,S})   \;\text{for some $\tm_2\neq m_2$}
\}|
\leq 
2^{n\left(  R_2-\frac{\eps}{2}   \right)} \nonumber\\&
\,,\;\text{if $I(X_1,X_2;\tX_2,S|U)-\left[ R_2-I(\tX_2;S|U) \right]_{+}>\eps$} \,.
\label{eq:M2c3ebn}
\end{align}
\end{enumerate}
\end{lemma}
The proof of Lemma~\ref{lemm:McodeBsets} is given in Appendix~\ref{app:McodeBsets}.

\subsection{Examples}
To illustrate our results, we give the following examples.

\begin{example}
\label{example:GB95}
(see \cite{GubnerHughes:95p})
In the first example, we use Corollary~\ref{coro:MCavc01} and previous results by Gubner and Hughes \cite{GubnerHughes:95p} to show that the deterministic code capacity region can be non convex. Consider the state dependent erasure MAC, specified
\begin{align}
Y=\begin{cases}
X_1+X_2 & X_1\cdot X_2=S=0 \,, \\
r       & \text{otherwise}
\end{cases}
\,,
\end{align}
where $\Sset=\{0,1 \}$, $\Xset_1=\Xset_2=\{0,1,\ldots, r-1 \}$, and $\Yset=\{0,1,\ldots,r\}$, with $r\geq 2$. Consider the AVMAC under a state constraint $\frac{1}{n}\sum_{i=1}^n s_i\leq \Lambda$, for $0<\Lambda\leq 1$, and with inactive input constraints, \ie
$\plimit_k=\cost_{k,\max}$ for $k=1,2$. 

It can be readily verified that the symmetrizability conditions in Definition~\ref{def:Msymmetrizable} hold when the distributions
$J(s|x_1,x_2)$ and $J_k(s|x_k)$ 
assign unit probability to $S=1$. Then, $\tLambda(P_{U,X_1,X_2})=\tLambda_k(P_{U,X_k})=1$ for all $P_{U,X_1,X_2}$, hence
$L^*=L_k^*=1$ for $k=1,2$. By Corollary~\ref{coro:MCavc01}, we have that without a state constraint, \ie for $\Lambda=1$, the deterministic code capacity region is
\begin{align}
\MCavc=\{(0,0)\} \,.
\end{align}
Whereas, if $\Lambda<1$, then $\Lambda$ is strictly less than $L^*=L_k^*=\tLambda(P_{U,X_1,X_2})=\tLambda_k(P_{U,X_k})=1$, for all $P_{U,X_1,X_2}$, $k=1,2$. Hence, by Corollary~\ref{coro:MCavc01},
\begin{align}
\MCavc=
\bigcup_{ P_U P_{X_1|U} P_{X_2|U} 
}
\left\{
\begin{array}{lrl}
(R_1,R_2) \,:\; & R_1 		\leq&   \min\limits_{q(s|u) \,:\; \E S\leq\Lambda} I_q(X_1;Y|X_2,U)  \,, \\
								& R_2 		\leq&   \min\limits_{q(s|u) \,:\; \E S\leq\Lambda} I_q(X_2;Y|X_1,U)  \,, \\
								& R_1+R_2 \leq&   \min\limits_{q(s|u) \,:\; \E S\leq\Lambda} I_q(X_1,X_2;Y|U)  
\end{array}
\right\} 
 \,.
\label{eq:GB95eMcasea}
\end{align}
Now, based on Theorem~\ref{theo:GubnerHughes} (see \cite{GubnerHughes:95p}), we deduce that the deterministic code capacity region is the same as the divided-randomness capacity region, \ie $\MCavc=\MdrCav$.
For the latter, Gubner and Hughes derive inner and outer bounds that are close enough in order to establish that 
$\MdrCav$ is non convex for high values of $r\geq 2$. This, in turn, implies that the deterministic code capacity region is also non convex in general.
\end{example}

\begin{example} 
\label{example:BSMAC}
Let $\avmac$ be an AVMAC which consists of independent binary symmetric channels. 
Specifically, let the state and the output be pairs as well, \ie $S=(S_1,S_2)$ and $Y=(Y_1,Y_2)$, such that
\begin{align}
Y_1=&X_1+S_1 \,\mod 2 \,,\nonumber\\
Y_2=&X_2+S_2 \,\mod 2 \,,
\end{align}
where $X_1$, $X_2$, $S_1$, $S_2$, $Y_1$ and $Y_2$ are binary. Suppose that 
the input and state cost functions are Hamming weights, \ie
\begin{align}
\cost_1(x_1)=x_1 \,,\; \cost_2(x_2)=x_2 \,,\; l(s)=s_1+s_2 \,,
\end{align}
while the constraints $\plimit_1$, $\plimit_2$ and $\Lambda$ are in the interval $(0,1]$. 

First, we use Theorem~\ref{theo:MrCav} to show that the random code capacity is given by 
\begin{align}
\MrCav=
\left\{
\begin{array}{lrl}
(R_1,R_2) \,:\; & R_1 		\leq&   h(\omega_1*\lambda)-h(\lambda)  \,, \\
								& R_2 		\leq&   h(\omega_2*\lambda)-h(\lambda)  
\end{array}
\right\} \,,
\label{eq:MrCavBSMAC}
\end{align}
where
\begin{align}
\omega_1=\min\left(\plimit_1,\frac{1}{2} \right) \,,\;
\omega_2=\min\left(\plimit_2,\frac{1}{2} \right) \,,\;
\lambda=\min\left(\Lambda,\frac{1}{2} \right) \,.
\end{align}
In particular, if $\Lambda\geq \frac{1}{2}$, then the random code capacity region is $\MrCav=\{ (0,0) \}$.

 It can further be seen that the binary AVMAC is symmetrizable-$\Xset_1\times\Xset_2$, symmetrizable- $\Xset_1|\Xset_2$, and symmetrizable-$\Xset_2|\Xset_1$. 
In particular, the symmetrizability equations (\ref{eq:MsymmetrizableJ}), (\ref{eq:Msymmetrizable1}), and (\ref{eq:Msymmetrizable2}) only hold with the $0$-$1$ laws  $J(s|x_1,x_2)=\delta(s_1-x_1)\delta(s_2-x_2)$, 
$J_1(s|x_1)=\delta(s_1-x_1)\delta(s_2-x_2')$,
$J_2(s|x_2)=\delta(s_1-x_1')\delta(s_2-x_2)$, for arbitrary $x_1',x_2'\in \{0,1\}$, where
$\delta(u)=1$ for $u=0$, and $\delta(u)=0$ otherwise.

 Then, we use Corollary~\ref{coro:MCavc01}  
to show that the capacity region is given by the following. If 
$\plimit_1>\Lambda$ and $\plimit_2>\Lambda$, then
\begin{align}
\MCavc=\MrCav=
\left\{
\begin{array}{lrl}
(R_1,R_2) \,:\; & R_1 		\leq&   h(\omega_1*\lambda)-h(\lambda)  \,, \\
								& R_2 		\leq&   h(\omega_2*\lambda)-h(\lambda)  
\end{array}
\right\} \,,
\label{eq:MCavcBSMACa}
\end{align}
where $h(t)=-t\log t-(1-t)\log(1-t)$ for $0< t< 1$, and $\alpha*\beta=(1-\alpha)\beta+\alpha(1-\beta)$. 
If $\plimit_1\leq\Lambda$ and $\plimit_2>\Lambda$, then
\begin{align}
\MCavc=
\{(0,R_2) \,:\;  R_2 		\leq&\;   h(\omega_2*\lambda)-h(\lambda) \} \,. 
\label{eq:MCavcBSMACb}
\end{align}
If $\plimit_1>\Lambda$ and $\plimit_2\leq\Lambda$, then
\begin{align}
\MCavc=
\{(R_1,0) \,:\;  R_1 		\leq&\;   h(\omega_1*\lambda)-h(\lambda) \} \,. 
\label{eq:MCavcBSMACc}
\end{align}
Otherwise,  if $\plimit_1\leq\Lambda$ and $\plimit_2\leq \Lambda$, then
\begin{align}
\MCavc=\{ (0,0) \} \,.
\label{eq:MCavcBSMACd}
\end{align}
The analysis is given in Appendix~\ref{app:BSMAC}.

We observe that the deterministic code capacity region and the random code capacity region are the same, only if $\Lambda\geq\frac{1}{2}$ or both input constraints are higher than the state constraints. In all other cases, the deterministic code capacity is strictly included within the random code capacity region, \ie
$\MCavc\subset\MrCav$.
\end{example}

\begin{example}
\label{example:GaussMAC}
Consider the Gaussian AVMAC, specified by
\begin{align}
Y=X_1+X_2+S+Z \,,
\end{align}
with $Z\sim\mathcal{N}(0,\sigma^2)$, where the transmitters and the jammer have the power constraints
 $\frac{1}{n}\sum_{i=1}^n X_{k,i}^2\leq \plimit_k$, for $k=1,2$, and 
$\frac{1}{n}\sum_{i=1}^n S_i^2\leq \Lambda$.  This channel was treated independently by 
Hosseinigoki and Kosut \cite{HosseinigokiKosut:19c}, using the packing lemmas from \cite{HosseinigokiKosut:17a}. Here, we use our results on the general AVMAC.

Although we previously assumed that the input, state and output alphabets are finite, our results can be extended to the continuous case as well, using standard discretization techniques \cite{BBT:59p,Ahlswede:78p,Csiszar:92p} \cite[Section 3.4.1]{ElGamalKim:11b}.
First, we use Theorem~\ref{theo:MrCav} to show that the random code capacity region is 
\begin{align}
\MrCav=
\left\{
\begin{array}{lrl}
(R_1,R_2) \,:\; & R_1 		\leq&  \frac{1}{2}\log\left(1+\frac{\plimit_1}{\Lambda+\sigma^2}\right)  \,, \\
								& R_2 		\leq&  \frac{1}{2}\log\left(1+\frac{\plimit_2}{\Lambda+\sigma^2}\right)  \,, \\
								& R_1+R_2 		\leq&  \frac{1}{2}\log\left(1+\frac{\plimit_1+\plimit_2}{\Lambda+\sigma^2}\right)  
\end{array}
\right\} \,.
\label{eq:MrCavGauss}
\end{align}

Then, we use Theorem~\ref{theo:MCavc} to show that the capacity region is given by the following. If 
$\plimit_1>\Lambda$ and $\plimit_2>\Lambda$, then
\begin{align}
\MCavc=\MrCav=
\left\{
\begin{array}{lrl}
(R_1,R_2) \,:\; & R_1 		\leq&  \frac{1}{2}\log\left(1+\frac{\plimit_1}{\Lambda+\sigma^2}\right)  \,, \\
								& R_2 		\leq&  \frac{1}{2}\log\left(1+\frac{\plimit_2}{\Lambda+\sigma^2}\right)  \,, \\
								& R_1+R_2 		\leq&  \frac{1}{2}\log\left(1+\frac{\plimit_1+\plimit_2}{\Lambda+\sigma^2}\right)  
\end{array}
\right\} \,.
\label{eq:MCavcGaussa}
\end{align}
If $\plimit_1\leq\Lambda$ and $\plimit_2>\Lambda$, then
\begin{align}
\MCavc=
\left\{ (0,R_2) \,:\;  R_2 		\leq\;   \frac{1}{2}\log\left(1+\frac{\plimit_2}{\Lambda+\sigma^2}\right) \right\} \,. 
\label{eq:MCavcGaussb}
\end{align}
If $\plimit_1>\Lambda$ and $\plimit_2\leq\Lambda$, then
\begin{align}
\MCavc=
\left\{ (R_1,0) \,:\;  R_1 		\leq\;   \frac{1}{2}\log\left(1+\frac{\plimit_1}{\Lambda+\sigma^2}\right) \right\} \,. 
\label{eq:MCavcGaussc}
\end{align}
Otherwise,  if $\plimit_1\leq\Lambda$ and $\plimit_2\leq \Lambda$, then
\begin{align}
\MCavc=\{ (0,0) \} \,.
\label{eq:MCavcGaussd}
\end{align}
The analysis is given in Appendix~\ref{app:GaussMAC}.
We observe that the deterministic code capacity region and the random code capacity region are the same, only if  both input constraints are higher than the state constraints. In all other cases, the deterministic code capacity is strictly included within the random code capacity region. 

\end{example}

\subsubsection{Without Constraints}
We have seen in Subsection~\ref{sec:MLNOsi} that Gubner \cite{Gubner:90p} and Ahlswede and Cai \cite{AhlswedeCai:99p} determined the capacity region in all but two cases, where User 1 has zero capacity and User 2 has positive capacity, and vice versa (see Remark~\ref{remark:WieseBoche}).
In this subsection, we give full characterization of the capacity region of the AVMAC without constraints, closing the gap in the results by Ahlswede and Cai \cite{AhlswedeCai:99p}.

\begin{theorem}
\label{theo:NoConstraintsFull}
The capacity region of the AVMAC free of constraints is given by the following.
\begin{enumerate}[a)]
\item
If $\mac$ is not symmetrizable-$\Xset_1\times\Xset_2$, -$\Xset_1|\Xset_2$, nor -$\Xset_2|\Xset_1$, then
\begin{align}
\MCavcf=\MdrICavf \,.
\end{align}

\item
If $\mac$ is not symmetrizable-$\Xset_1\times\Xset_2$ nor -$\Xset_2|\Xset_1$, but symmetrizable-$\Xset_1|\Xset_2$,  then 
\begin{align}
\MCavcf= \left\{(0,R_2) \,:\; R_2\leq \min_{q(s)} \max_{p(x_1)p(x_2)} 
I_q(X_2;Y|X_1) \right\} \,.
\label{eq:MCavcf2Full}
\end{align}

\item
If $\mac$ is not symmetrizable-$\Xset_1\times\Xset_2$ nor -$\Xset_1|\Xset_2$, but symmetrizable-$\Xset_2|\Xset_1$,  then 
\begin{align}
\MCavcf= \left\{(R_1,0) \,:\; R_1\leq \min_{q(s)} \max_{p(x_1)p(x_2)} 
I_q(X_1;Y|X_2) \right\} \,.
\label{eq:MCavcf3Full}
\end{align}

\item
In all other cases, 
\begin{align}
\MCavcf=\left\{(0,0) \right\} \,.
\end{align}
\end{enumerate}
\end{theorem}
Theorem~\ref{theo:NoConstraintsFull} is a direct consequence of our previous results in the presence of constraints, since plugging $\plimit_k>\cost_{k,max}$, $k=1,2$, and $\Lambda>l_{max}$ yields the capacity region without constraints. In particular, if the AVMAC free of constraints is non-symmetrizable-$\Xset_1\times\Xset_2$, then $L^*=+\infty>\Lambda$.  Otherwise, if the AVMAC free of constraints is symmetrizable-$\Xset_1\times\Xset_2$, then $L^*\leq l_{max}<\Lambda$.
A similar argument holds for $L^*_1$, $L_2^*$ and non-symmetrizability-$\Xset_1|\Xset_2$, -$\Xset_2|\Xset_1$, respectively. Thus,  Theorem~\ref{theo:NoConstraintsFull} follows from Theorem~\ref{theo:MCavc} and 
Definition~\ref{def:MICavc}.
The theorem above completes the partial characterization in Cases b) and c) in Theorem~\ref{theo:WieseBoche}.

\begin{appendices}

\section{Proof of Lemma~\ref{lemm:MCcompound}}
\label{app:MCcompound}
Consider the compound MAC $\Mcompound$ under input constraints $(\plimit_1,\plimit_2)$ and state constraint $\Lambda$. 
To prove the direct part, we construct a code based on simultaneous decoding with respect to a state type which is ``close" to some  $q\in\Qset$.
The converse part follows 
by standard arguements.

\subsection{Achievability Proof}
 Let $\eps,\delta>0$ be arbitrarily small.
We use the following notation. Basic method of types concepts are defined as in \cite[Chapter 2]{CsiszarKorner:82b}; including the definition of a type $\hP_{x^n}$ of a sequence $x^n$; a joint type $\hP_{x^n,y^n}$ and a conditional type $\hP_{x^n|y^n}$ of a pair of sequences $(x^n,y^n)$; and 
 a $\delta$-typical set $\tset(P_{X,Y})$ with respect to a distribution $P_{X,Y}(x,y)$. 
 We also define a set of state types $\tQ$ by  
\begin{align}
\label{eq:tQ}
\tQ=\{ \hP_{s^n} \,:\; s^n\in\Aset^{ \delta_1 
}(q) \;\text{ for some  $q\in\pLSpaceS$}\, \} \,,
\end{align}
where 
\begin{align}
\label{eq:delta1}
\delta_1 \triangleq
\frac{\delta}{2\cdot |\Sset|} \,.
\end{align} 
Namely, $\tQ$ is the set of types that are $\delta_1$-close 
 to some state distribution $q(s)$ in $\pLSpaceS$.
Then, fix $P_U P_{X_1|U} P_{X_2|U}$ such that $\E\cost_k(X_k)\leq\plimit_k-\eps$, for $k=1,2$.
 
\emph{Codebook Generation}: 
Generate a random time sharing sequence $u^n\sim \prod_{i=1}^n P_U(u_i)$.
Then,  generate $2^{nR_k}$ conditionally independent sequences 
$x_k^n(m_k)$, $m_k\in[1:2^{nR_k}]$, at random, each according to 
$\prod_{i=1}^n P_{X_k|U}(x_{k,i}|u_i)$, for $k=1,2$. 
 Reveal the  sequence $u^n$ and the codebooks $\{x_1^n(m_1)\}$ and $\{x_2^n(m_2)\}$
 to the encoders and the decoder.

\emph{Encoding}: To send $(m_1,m_2)$, Encoder $k$ transmits $x_k^n(m_k)$, provided that 
\begin{align}
\cost_k(x_k^n(m_k))\leq\plimit_k \,,\; \text{for $k=1,2$}\,. 
\label{eq:McompoundConstI}
\end{align}
Otherwise, repeatedly send the symbol $a_k$ 
with $\cost_k(a_k)=0$.

\emph{Decoding}:
For every state distribution $q(s)$, define
 \begin{align}
\label{eq:MUchannelYL}
P^{q}_{Y|X_1,X_2}(y|x_1,x_2)=\sum_{s\in\Sset} q(s)\mac(y| x_1,x_2,s) \,.
 \end{align}
 As $y^n$ is received, the decoder finds a unique pair $(\hm_1,\hm_2)\in[1:2^{nR_1}]\times[1:2^{nR_2}]$ such that $(u^n,x_1^n(\hm_1),x_2^n(\hm_2),y^n)\in\tset(P_U P_{X_1|U} P_{X_2|U}  P^q_{Y|X_1,X_2})$
for some type $q\in\tQ$. If there is none, or more than one such pair, 
declare an error.

\emph{Analysis of Probability of Error}:
Assume without loss of generality that the users sent the messages 
$m_1=m_2=1$. Let $q(s)\in\Qset$ denote the \emph{actual} state distribution chosen by the jammer.
The error event is within the union of the following events, 
\begin{align}
\Eset_1=&\{ (U^n,X_1^n(1),X_2^n(1))\notin \Aset^{\nicefrac{\delta}{3}}(P_U P_{X_1|U} P_{X_2|U}) \}\,,  \\
\Eset_2=&\{ (U^n,X_1^n(1),X_2^n(1),Y^n)\notin \tset(P_U P_{X_1|U} P_{X_2|U} P^{q'}_{Y|X_1,X_2})   
\,,\;
\text{for all $q'\in\tQ$} \} \,,
\label{eq:MSlemmCompoundCerrEv1cap}\\
\Eset_3=&\{ (U^n,X_1^n(m_1),X_2^n(m_2),Y^n)\in \tset(P_U P_{X_1|U} P_{X_2|U} P^{q'}_{Y|X_1,X_2}) 
\,,\;
\text{for some $m_1\neq 1,\,m_2\neq 1,\, q'\in\tQ$} \}\,, \\
\Eset_4=&\{ (U^n,X_1^n(m_1),X_2^n(1),Y^n)\in \tset(P_U P_{X_1|U} P_{X_2|U} P^{q'}_{Y|X_1,X_2}) 
\,,\;
\text{for some $m_1\neq 1,\, q'\in\tQ$} \}\,, \\
\Eset_5=&\{ (U^n,X_1^n(1),X_2^n(m_2),Y^n)\in \tset(P_U P_{X_1|U} P_{X_2|U} P^{q'}_{Y|X_1,X_2}) 
\,,\;
\text{for some $m_2\neq 1,\, q'\in\tQ$} \}\,.
\,.
\label{eq:MSlemmCompoundCerrEv2cap}
\end{align}
 The probability of error is then bounded by
\begin{align}
 \err(q,\code) 
\leq& \prob{\Eset_1}+\cprob{\Eset_2}{\Eset_1^c}+ \cprob{\Eset_3}{\Eset_2^c}
+\cprob{\Eset_4}{\Eset_2^c}
+ \cprob{\Eset_5}{\Eset_2^c} \,,
\label{eq:MSlemmCompoundCerr}
\end{align}
 where the conditioning on $(M_1,M_2)=(1,1)$ is omitted for convenience of notation. 
The first term in the RHS of (\ref{eq:MSlemmCompoundCerr}) tends to zero exponentially as $n\rightarrow\infty$, by the law of large numbers and Chernoff's bound (see \eg \cite[Theorem 1.1]{Kramer:08n}). Now, 
given that the event $\Eset_1^c$ occurs, we have that   $X_1^n(1)$ and $X_2^n(1)$ satisfy the input constraints (\ref{eq:McompoundConstI}), 
for sufficiently small $\delta>0$, and are thus the channel inputs.

Moving to the second term, suppose that
\begin{align}
(U^n,X_1^n(1),X_2^n(1),Y^n)\in \Aset^{\nicefrac{\delta}{2}}(P_U P_{X_1|U} P_{X_2|U} P^{q}_{Y|X_1,X_2}) 
\,. 
\label{eq:Mux1x2y}
\end{align}
%
Then, for sufficiently large $n$, there is a type $q'(s)$ such that 
$
|q'(s)-q(s)|\leq \delta_1 
$, 
 for all $s\in\Sset$, hence, $q'\in\tQ$ (see definition in (\ref{eq:tQ})), and 
\begin{align}
|P_{Y|X_1,X_2}^{q'}(y|x_1,x_2)-P_{Y|X_1,X_2}^{q}(y|x_1,x_2)|\leq |\Sset|\cdot \delta_1=\frac{\delta}{2} \,,
x_k\in\Xset_k ,\, y\in\Yset \,,
\end{align}
for $\delta_1=\delta/2|\Sset|$ (see (\ref{eq:MUchannelYL})), 
hence $\Eset_2$ does not hold. 
It follows by contradiction that 
	\begin{align}
	\cprob{\Eset_2}{\Eset_1^c} 
	\leq& \cprob{(U^n,X_1^n(1),X_2^n(1),Y^n)\notin \Aset^{\nicefrac{\delta}{2}}(P_U P^{q}_{Y|U})  }{\Eset_1^c} \,,
		\label{eq:MSllnRL}
	\end{align}
which 
  tends to zero exponentially as $n\rightarrow\infty$ by the law of large numbers and  Chernoff's bound. 
	
	As for the third term in the RHS of (\ref{eq:MSlemmCompoundCerr}),  	by the union of events bound 
and the fact that the number of type classes in $\Sset^n$ is bounded by $(n+1)^{|\Sset|}$,  
 we have that  
\begin{align}
&\cprob{\Eset_3}{\Eset_2^c}
\leq (n+1)^{|\Sset|}\cdot \sup_{q'\in\tQ}
\nonumber\\&
 \Pr \bigg(
(U^n,X_1^n(m_1),X_2^n(m_2),Y^n)\in \tset(P_U P_{X_1|U} P_{X_2|U} P^{q'}_{Y|X_1,X_2}) \,,\;
\text{for some $m_1\neq 1 ,\, m_2\neq 1$} 
| \Eset_1^c \bigg) \nonumber\\
\leq& (n+1)^{|\Sset|}\cdot 2^{n(R_1+R_2)} \cdot \sup_{q'\in\tQ}\bigg[  
\sum_{u^n,x_1^n,x_2^n} P_{U^n}(u^n) P_{X_1^n|U^n}(x_1^n|u^n) P_{X_2^n|U^n}(x_2^n|u^n) 
\nonumber\\&
\cdot \sum_{y^n \,:\; (u^n,x_1^n,x_2^n,y^n)\in \tset(P_U P_{X_1|U} P_{X_2|U} P^{q'}_{Y|X_1,X_2})} 
P_{Y^n}^q(y^n|u^n)
\bigg] \,,
\label{eq:MSE2poly}
\end{align}
where we have defined
	$
	P_{Y|U}^q(y|u)
	=\sum\limits_{x_1,x_2,s\in\Sset} P_{X_1|U}(x_1|u) P_{X_2|U}(x_2|u)q(s) \mac(y|$ $x_1,x_2,s) 
	$. 
This follows since 
 $X_1^n(m_1)$ and $X_2^n(m_2)$ are independent of $Y^n$ for every $m_1\neq 1$ and $m_2\neq 1$. 
 Let $y^n$ satisfy $(u^n,x_1^n,x_2^n,y^n)\in \tset(P_U P_{X_1|U} P_{X_2|U}$ $ P^{q'}_{Y|X_1,X_2})$. Then, 
$\,(u^n,y^n)\in\Aset^{\delta_2}(P_{U,Y}^{q'})$ with $\delta_2\triangleq |\Xset_1||\Xset_2|\cdot\delta$. By Lemmas 2.6 and 2.7 in
 \cite{CsiszarKorner:82b},
\begin{align}
P_{Y^n|U^n}^q(y^n|u^n)=2^{-n\left(  H(\hP_{y^n|u^n})+D(\hP_{y^n|u^n}||P_{Y|U}^q)
\right)}
\leq& 2^{-n H(\hP_{y^n|u^n})}
\leq 2^{-n\left( H_{q'}(Y|U) -\eps_1(\delta) \right)} \,,
\label{eq:MSpYbound}
\end{align}
where $\eps_1(\delta)\rightarrow 0$ as $\delta\rightarrow 0$. Therefore, by (\ref{eq:MSE2poly})$-$(\ref{eq:MSpYbound}), along with 
 \cite[Lemma 2.13]{CsiszarKorner:82b},
\begin{align}
& \prob{\Eset_3|\Eset_2^c}           																									
\leq
 \;(n+1)^{|\Sset|}\cdot \sup_{q'\in\Qset} 
2^{-n[ I_{q'}(X_1,X_2;Y|U) 
-R_1-R_2-\eps_2(\delta) ]} \label{eq:MSLexpCR} \,,
\end{align}
with $\eps_2(\delta)\rightarrow 0$ as $\delta\rightarrow 0$, 
 The RHS of (\ref{eq:MSLexpCR})
  tends to zero exponentially as $n\rightarrow\infty$, provided that 
	\begin{align}
	R_1+R_2<\inf_{q'\in\Qset} I_{q'}(X_1,X_2;Y|U)-\eps_2(\delta) \,.
	\end{align}

By similar considerations, the fourth term is bounded by 
$
 \prob{\Eset_4|\Eset_2^c}           																									
\leq
 \;(n+1)^{|\Sset|}\cdot \sup_{q'\in\Qset} 
2^{-n[ I_{q'}(X_1;Y|X_2,U) 
-R_1-\eps_3(\delta) ]} 
$, 
with $\eps_3(\delta)\rightarrow 0$ as $\delta\rightarrow 0$.
 This bound 
tends to zero exponentially as $n\rightarrow\infty$, provided that 
	$R_1<\inf_{q'\in\Qset} I_{q'}(X_1;Y|X_2,U)-\eps_3(\delta)$. By symmetry, we have that 
$ \prob{\Eset_5|\Eset_2^c}$ 
tends to zero as well, provided that 
$R_2<\inf_{q'\in\Qset} I_{q'}(X_2;Y|X_1,U)-\eps_3(\delta)$.

We conclude that the probability of error, averaged over the class of the codebooks, exponentially decays to zero  as $n\rightarrow\infty$. Therefore, there must exist a $(2^{nR_1},2^{nR_2},n,\eps)$ deterministic code, for a sufficiently large $n$.
\qed

\subsection{Converse Proof} 
The converse part follows from the same arguments as in the converse proof of the classical MAC 
\cite{Ahlswede:74p} (see also \cite[Section 15.3.4]{CoverThomas:06b}). 
Since the deterministic code capacity region is always bounded by the random code capacity region,
we consider a sequence of $(2^{nR_1},2^{nR_2},n,\alpha_n)$ random codes, where $\alpha_n\rightarrow 0$  as $n\rightarrow\infty$.
Then, let $X_1^n=f_{1,\gamma}^n(M_1)$ and $X_2^n=f_{2,\gamma}^n(M_2)$ be the channel input sequences, and 
$Y^n$ be the corresponding output sequence, where $\gamma\in\Gamma$ is the random element shared between the encoders and the decoder.
 For every $q\in\Qset$, we have by 
Fano's inequality that $H_q(M_1,M_2|Y^n,\gamma)\leq n\eps_n $, hence 
$H_q(M_1|M_2,Y^n,\gamma)
\leq n\eps_n $ and 
$H_q(M_2|M_1,Y^n,\gamma)
\leq n\eps_n $, where
$\eps_n\rightarrow 0$ as $n\rightarrow\infty$. Since
\begin{align}
n(R_1+R_2)=&H(M_1,M_2|\gamma)=I_q(M_1,M_2;Y^n|\gamma)+H(M_1,M_2|Y^n,\gamma) \,,\\
nR_1=&H(M_1|M_2,\gamma)=I_q(M_1;Y^n|M_2,\gamma)+H(M_1|M_2,Y^n,\gamma) \,,\\
nR_2=&H(M_2|M_1,\gamma)=I_q(M_2;Y^n|M_1,\gamma)+H(M_2|M_1,Y^n,\gamma) \,,
\end{align}
it follows that
\begin{align}
nR_1\leq& I_q(M_1;Y^n|M_2,\gamma)+n\eps_n=\sum_{i=1}^n I_q(M_1;Y_i|M_2,Y^{i-1},\gamma)+n\eps_n \,,\\
nR_2\leq& I_q(M_2;Y^n|M_1,\gamma)+n\eps_n=\sum_{i=1}^n I_q(M_2;Y_i|M_1,Y^{i-1},\gamma)+n\eps_n \,, \\
n(R_1+R_2)\leq& I_q(M_1,M_2;Y^n|\gamma)+n\eps_n=\sum_{i=1}^n I_q(M_1,M_2;Y_i|Y^{i-1},\gamma)+n\eps_n \,.
\end{align}
As $X_k^n=f_{k,\gamma}(M_k)$, this yields
\begin{align}
R_1\leq&  \frac{1}{n}\sum_{i=1}^n I_q(X_{1,i},M_1;Y^n|X_{2,i},M_2,Y^{i-1},\gamma)+\eps_n\\
R_2\leq&  \frac{1}{n} \sum_{i=1}^n I_q(X_{2,i},M_2;Y^n|X_{1,i},M_1,Y^{i-1},\gamma)+\eps_n \\
R_1+R_2 \leq& \frac{1}{n}\sum_{i=1}^n I_q(X_{1,i},X_{2,i},M_1,M_2;Y_i|Y^{i-1},\gamma)+\eps_n \,.
\end{align}
Then, since $(\gamma,M_1,M_2,Y^{i-1})\Cbar (X_{1,i},X_{2,i})\Cbar Y_i$ form a Markov chain, we have that for every $q\in\Qset$,
\begin{align}
R_1\leq& I_q(X_{1,T};Y_T|X_{2,T},T,\gamma)+\eps_n \,, 			\label{eq:compoundConvR2} \\
R_2\leq& I_q(X_{2,T};Y_T|X_{1,T},T,\gamma)+\eps_n \,, 			\label{eq:compoundConvR1}	\\
R_1+R_2 \leq& I_q(X_{1,T},X_{2,T};Y_T|T,\gamma)+\eps_n \,,	\label{eq:compoundConvSum}
\end{align}
where $T$ is a random variable which is uniformly distributed over $[1:n]$, and independent of $(\gamma,X_1^n,X_2^n,S^n,Y^n)$. Defining $X_1=X_{1,T}$, $X_2=X_{2,T}$, $Y=Y_T$, and $U=(T,\gamma)$, it follows that
\begin{align}
R_1\leq& \inf_{q\in\Qset} I_q(X_1;Y|X_2,U)+\eps_n \,, \label{eq:compoundConvR1u} \\
R_2\leq& \inf_{q\in\Qset} I_q(X_2;Y|X_1,U)+\eps_n \,, \label{eq:compoundConvR2u}\\
R_1+R_2 \leq& \inf_{q\in\Qset} I_q(X_{1},X_2;Y|U)+\eps_n \label{eq:compoundConvSumu}\,.
\end{align}
As $X_1$ and $X_2$ are conditionally independent given $U$, this completes  the proof of the converse part.  
%
\qed

\section{Proof of Theorem~\ref{theo:MrCav}}
\label{app:MrCav}
Consider the AVMAC $\avmac$ under input constraints $(\plimit_1,\plimit_2)$ and state constraint $\Lambda$.

\subsection{Achievability Proof}
To prove the random code capacity theorem for the AVMAC, we use our result on the compound MAC along with a simple extension of Ahlswede's Robustification Technique (RT). 
We begin with a lemma from \cite{PeregSteinberg:19p1}, based on 
 Ahlswede's RT \cite{Ahlswede:86p}. 
\begin{lemma}[Ahlswede's RT \cite{Ahlswede:86p} {\cite[Lemma 9]{PeregSteinberg:19p1}}] 
\label{lemm:LRT}
Let $h:\Sset^n\rightarrow [0,1]$ be a given function. If, for some fixed $\alpha_n\in(0,1)$, and for all $ q^n(s^n)=\prod_{i=1}^n q(s_i)$, with 
$q\in\apLSpaceS$, 
\begin{align}
\label{eq:RTcondCs}
\sum_{s^n\in\Sset^n} q^n(s^n)h(s^n)\leq \alpha_n \,,
\end{align}
then,
\begin{align}
\label{eq:RTresCs}
\frac{1}{n!} \sum_{\pi\in\Pi_n} h(\pi s^n)\leq \beta_n \,,\quad\text{for all $s^n\in\Sset^n$ such that $l^n(s^n)\leq\Lambda$} \,,
\end{align}
where $\Pi_n$ is the set of all $n$-tuple permutations $\pi:\Sset^n\rightarrow\Sset^n$, and 
$\beta_n=(n+1)^{|\Sset|}\cdot\alpha_n$. 
\end{lemma}

Let $(R_1,R_2)\in\MrICav$.
At first, we consider the compound MAC under input constraints $(\plimit_1,\plimit_2)$, with $\Qset=\pLSpaceS$. 
According to Lemma~\ref{lemm:MCcompound},  
for some $\theta>0$ and sufficiently large $n$,   there exists a  $(2^{nR_1},2^{nR_2},n)$  code 
$\code=(\enc_1^n(m_1),\enc_2^n(m_2),$ $\dec(y^n))$ for the compound MAC $\avmac^{\pLSpaceS}$ such that 
\begin{align}
\label{eq:MLrAVcosti}
&
\cost_k^n(\enc_k(m_k))  \leq\plimit_k \,,\; 
\text{for all $m_k\in [1:2^{nR_k}]$, $k=1,2$}\,,
\end{align}
and
\begin{align}
\label{eq:MLrAVerrDirect}
&  \err(q,\code)=\sum_{s^n\in\Sset^n} q(s^n) \cdot \cerr(\code)  \leq e^{-2\theta n} \,,
\end{align}
for all i.i.d. state distributions $q(s^n)=\prod_{i=1}^n q(s_i)$, with $q\in\apLSpaceS$.

Therefore, by Lemma~\ref{lemm:LRT}, taking $h_0(s^n)=\cerr(\code)$ and $\alpha_n=e^{-2\theta n}$, we have that for a sufficiently large $n$,
\begin{align}
\label{eq:MALdetErrC}
\frac{1}{n!} \sum_{\pi\in\Pi_n} \E\, P_{e|\pi s^n}^{(n)}(\code)\leq (n+1)^{|\Sset|}e^{-2\theta n} 
\leq e^{-\theta n}  \,,
\end{align}
for all $s^n\in\Sset^n$ with $l^n(s^n)\leq\Lambda$, where the sum is over the set of all $n$-tuple permutations. 

On the other hand, for every  $\pi\in\Pi_n$,
\begin{align}
P_{e|\pi s^n}^{(n)}(\code) 
  \stackrel{(a)}{=}&
\frac{1}{2^{ n(R_1+R_2) }}\sum_{m_1,m_2}
\sum_{y^n:\dec(y^n)\neq (m_1,m_2)}  W_{Y^n|X_1^n,X_2^n,S^n}(y^n|\enc_1(m_1),\enc_2(m_2),\pi s^n) \nonumber\\
\stackrel{(b)}{=}& \frac{1}{2^{ n(R_1+R_2) }}\sum_{m_1,m_2}
\sum_{y^n:\dec(\pi y^n)\neq (m_1,m_2)}  W_{Y^n|X_1^n,X_2^n,S^n}(\pi y^n|\enc_1(m_1),\enc_2(m_2),\pi s^n) \nonumber\\
\stackrel{(c)}{=}& \frac{1}{2^{ n(R_1+R_2) }}\sum_{m_1,m_2}
\sum_{y^n:\dec(\pi y^n)\neq (m_1,m_2)}  W_{Y^n|X_1^n,X_2^n,S^n}( y^n|\pi^{-1}\enc_1(m_1),\pi^{-1}\enc_2(m_2), s^n) \,,
\label{eq:MLcerrpi}
\end{align}
where $(a)$ is obtained by plugging 
 $\pi s^n$  in (\ref{eq:Mcerr});
in $(b)$ we simply change the order of summation over $y^n$; and $(c)$ holds because the channel is memoryless. 

Then, consider the $(2^{nR_1},2^{nR_2},n)$ random code $\code^\Pi$, specified by 
\begin{align}
\label{eq:MLCpi}
&f_{1,\pi}(m_1)= \pi^{-1} \enc_1(m_1) \,,\; f_{2,\pi}(m_2)= \pi^{-1} \enc_2(m_2) \,,\;
 g_\pi(y^n)=\dec(\pi y^n)   \,,
\end{align}
with a uniform distribution $\mu(\pi)=\frac{1}{|\Pi_n|}=\frac{1}{n!}$ for $\pi\in\Pi_n$. 
As the inputs cost is additive (see (\ref{eq:MLInConstraintStrict})), the permutation does not affect the costs of the codewords, hence the random code satisfies the input constraints $(\plimit_1,\plimit_2)$.
 From (\ref{eq:MLcerrpi}), we see that 
 $
\cerr(\code^\Pi)=\sum_{\pi\in\Pi_n} \mu(\pi) \cdot\E\, P_{e|\pi s^n}^{(n)}(\code) 
$,
for all $s^n\in\Sset^n$ with $l^n(s^n)\leq \Lambda$. Therefore, together with (\ref{eq:MALdetErrC}), we have that the probability of error of the random code $\code^\Pi$ is bounded by 
$
\err(\qn,\code^{\Pi})\leq e^{-\theta n} 
$, 
for every $\qn(s^n)\in\pLSpaceSn$. 
It follows that $\code^\Pi$ is a $(2^{nR},n,e^{-\theta n})$ random 
 code for the AVMAC $\avmac$  under input constraints $(\plimit_1,\plimit_2)$ and state constraint $\Lambda$. 
\qed

\subsection{Converse Proof}
Assume to the contrary that there exists an achievable rate pair 
\begin{align}
(R_1,R_2)\notin\inC(\Mcompound) \big|_{\Qset=\pSpace_{\Lambda-\delta}} \,,
\label{eq:MrCconverseRate}
\end{align}
using random codes over the AVMAC $\avmac$  under input constraints $(\plimit_1,\plimit_2)$ and state constraint $\Lambda$, where $\delta>0$ is arbitrarily small. 
That is, for every $\eps>0$ and sufficiently large $n$,
there exists a $(2^{nR_1},2^{nR_2},n)$ random code $\code^\Gamma=(\mu,\Gamma,\{\code_\gamma\}_{\gamma\in\Gamma})$ for the AVMAC $\avmac$, such that 
 $\sum_{\gamma\in\Gamma} \mu(\gamma)  \cost^n_k(\enc_{k,\gamma}(m_k))  \leq\plimit_k$, for
$k=1,2$, and 
\begin{align}
& \err(q,\code^\Gamma)\leq\eps \,,
\label{eq:MStateConverse1b}
\end{align}
 for all $m_1\in [1:2^{nR_1}]$, $m_2\in [1:2^{nR_2}]$, and $q(s^n)\in\pLSpaceSn$. 
	In particular, for  distributions $q(\cdot)$ which give mass $1$ to some sequence $s^n\in\Sset^n$ with $l^n(s^n)\leq\Lambda$, we have that
	$
	\cerr(\code^\Gamma)\leq\eps 
	$. 
	
	Consider using the random code $\code^\Gamma$ over the compound MAC $\avc^{\overline{\pSpace}_{\Lambda-\delta}(\Sset)}$ under input constraints $(\plimit_1,\plimit_2)$. Let $\oq(s)\in\overline{\pSpace}_{\Lambda-\delta}(\Sset)$ be a given state distribution. Then, 
	 define a sequence of i.i.d. random variables $\oS_1,\ldots,\oS_n\sim \oq(s)$.  
	Letting 
	$\oq^n(s^n)\triangleq\prod_{i=1}^n \oq(s_i)$, the probability of error is bounded by
	\begin{align}
	\err(\oq,\code^\Gamma)
	\leq 
	\sum_{s^n\,:\; l^n(s^n)\leq\Lambda} \oq^{n}(s^n) \cerr(\code^\Gamma)
	+\prob{l^n(\oS^{n})>\Lambda} . 
	\end{align}
	The first sum is bounded by (\ref{eq:MStateConverse1b}), and the second term vanishes 
	by the law of large numbers, since
	$\oq(s)\in\overline{\pSpace}_{\Lambda-\delta}(\Sset)$. 
	It follows that the random code $\code^\Gamma$ achieves a rate pair 
	$(R_1,R_2)$ as in (\ref{eq:MrCconverseRate}) 
	over the compound MAC $\avmac^{\overline{\pSpace}_{\Lambda-\delta}(\Sset)}$ under  input constraints 
	$(\plimit_1,\plimit_2)$, for an arbitrarily small $\delta>0$, in contradiction to Lemma~\ref{lemm:MCcompound}. 
	We deduce that the assumption is false, and $\MrCav\subseteq  \inC(\Mcompound) \big|_{\Qset=\pSpace_{\Lambda}}=\MrICav$.
\qed

\section{Proof of Lemma~\ref{lemm:MdisDec}}
\label{app:MdisDec}
Consider the AVMAC $\avmac$ under input constraints $(\plimit_1,\plimit_2)$ and state constraint $\Lambda$.
We accommodate the proof of Lemma 1 in \cite{AhlswedeCai:99p} to the case where there are state constraints. 
We prove part 2 of the lemma, and the rest follows by similar considerations.

Our first step is to  extend an auxiliary lemma by Csis\`ar and Narayan \cite[Lemma A2]{CsiszarNarayan:88p}.
Fix $P_U$,  $P_{X_k|U}$, for $k=1,2$, as in Lemma~\ref{lemm:MdisDec}.
\begin{lemma}
\label{lemm:MA2}
For every pair of conditional state distributions $Q(s|x_1,u)$ and $Q'(s|x_1,u)$  such that 
\begin{align}
\max\left\{
\sum_{u,x_1,s} P_U(u) P_{X_1|U}(x_1|u)Q(s|x_1,u)l(s)  \,,\; 
\sum_{u,x_1,s} P_U(u) P_{X_1|U}(x_1|u)Q'(s|x_1,u)l(s) 
\right\}
<&\tLambda_1(P_{U,X_1}) \,, 
\label{eq:MA2sump} 
\end{align}
 there exists $\xi>0$ such that
\begin{align}
\max_{u,x_1,\tx_1,x_2,y} \Big|\sum_{s} Q(s|\tx_1,u) \mac(y|x_1,x_2,s) -\sum_{s} Q'(s|x_1,u) \mac(y|\tx_1,x_2,s) \Big|
\geq \xi \,.
\label{eq:MdecCNresA2}
\end{align}
\end{lemma}

\begin{proof}[Proof of Lemma~\ref{lemm:MA2}]
Assume to the contrary that the LHS in (\ref{eq:MdecCNresA2}) is zero, and
define
\begin{align}
Q_A(s|x_1,u)=\frac{1}{2}\left(Q(s|x_1,u)+Q'(s|x_1,u)\right) \,.
\end{align}
By symmetry,
\begin{align}
0=& \max_{u,x_1,\tx_1,x_2,y} \Big|\sum_{s}  Q(s|\tx_1,u)\mac(y|x_1,x_2,s) -\sum_{s}  Q'(s|x_1,u)\mac(y|\tx_1,x_2,s) \Big|
\nonumber\\
=&  \frac{1}{2} \max_{u,x_1,\tx_1,x_2,y} \Big|\sum_{s}  Q(s|\tx_1,u)\mac(y|x_1,x_2,s) -\sum_{s}  Q'(s|x_1,u)\mac(y|\tx_1,x_2,s)\Big|
\nonumber\\
&+\frac{1}{2}\max_{u,x_1,\tx_1,x_2,y} \Big|\sum_{s}  Q'(s|\tx_1,u)\mac(y|x_1,x_2,s) -\sum_{s}  Q(s|x_1,u)\mac(y|\tx_1,x_2,s)\Big| 
\nonumber\\
\geq&    \max_{u,x_1,\tx_1,x_2,y} \Big|\sum_{s} Q_A(s|x_1,u)\mac(y|\tx_1,x_2,s) -\sum_{s}  Q_A(s|\tx_1,u) \mac(y|x_1,x_2,s) \Big| \,,
\end{align}
 where the last line follows from the triangle inequality. Then,
it follows that
\begin{align}
\sum_{s\in\Sset} Q_A(s|x_1,u)\mac(y|\tx_1,x_2,s) =\sum_{s\in\Sset}  Q_A(s|\tx_1,u) \mac(y|x_1,x_2,s) \,,
\end{align}
 for all $u\in\Uset$, $x_1,\tx_1\in\Xset_1$, $x_2\in\Xset_2$, and $y\in\Yset$. In other words, $J_{1,u}\equiv Q_A(\cdot|\cdot,u)$ 
symmetrizes-$\Xset_1|\Xset_2$ the AVMAC, for all $u\in\Uset$.

Next, recall from Remark~\ref{rem:MLambdaJeq} that the minimal state cost in (\ref{eq:MtlambdaJ1}) can be written as 
\begin{align}
\tLambda_1(P_{U,X_1})= \min_{ \text{symm. $\{ J_{1,u} \}$}} \sum_{u\in\Uset} \sum_{x_1\in\Xset_1} \sum_{s\in\Sset} P_U(u) 
P_{X_1|U}(x_1|u) J_{1,u}(s|x_1) l(s) \,,
\end{align}
where the minimization is over the set of 
distributions $\{J_{1,u} \}_{u\in\Uset}$, such that each
$J_{1,u}(s|x_1)$ satisfies (\ref{eq:Msymmetrizable1}), for $u\in\Uset$. 
Nevertheless, by (\ref{eq:MA2sump}), 
\begin{align}
\sum_{u,x_1,s} P_U(u) P_{X_1|U}(x_1|u)Q_A(s|x_1,u)l(s)  <& \tLambda_1(P_{U,X_1}) \,.
\end{align}
This is a contradiction, since we have seen above that the distributions $J_{1,u}\equiv Q_A(\cdot|\cdot,u)$ symmetrize-$\Xset_1|\Xset_2$ the AVMAC, for all $u\in\Uset$. It follows that 
the LHS of  (\ref{eq:MdecCNresA2}) must be positive.
This completes the proof of the auxiliary Lemma.
\end{proof}

We move to the main part of the proof. Notice that while the parameters $\eta$, $\eta_1$ and $\eta_2$ can be chosen freely, the rest of the parameters, $\xi$, $\delta_0$, $\delta_1$, and $\delta_2$, depend on $P_U$ and $P_{X_k|U}$.
To show that (\ref{eq:MdisAmb1}) holds for sufficiently small $\eta$ and $\eta_1$, assume to the contrary that there exists $y^n$ in $\Dset(m_1,m_2)\cap\Dset(\breve{m}_1,m_2)\neq\emptyset$ for some $\breve{m}_1\neq m_1$.
By the assumption in the lemma, the codewords $\{f_k(m_k)\}_{m_k\in [1:2^{nR_k}]}$ in Codebook $k$ have the same conditional type, given
the time sharing sequence $u^n$, for $k=1,2$. In particular, $P_{\tX_1|U}=P_{X_1|U}=P_1$ and $P_{\tX_2|U}=P_{X_2|U}=P_2$. 

By Condition 1) of the decoding rule for $\Dset(m_1,m_2)$,
\begin{align}
&D(P_{U,X_1,X_2,S,Y}|| P_U\times P_{X_1|U}\times P_{X_2|U}\times P_{S|U} \times \mac) \nonumber \\
=&\sum_{u,x_1,x_2,s,y} P_{U,X_1,X_2,S,Y}(u,x_1,x_2,s,y) 
\cdot  \log \frac{P_{U,X_1,X_2,S,Y}(u,x_1,x_2,s,y)}{ P_U(u) P_1(x_1|u) P_2(x_2|u) P_{S|U}(s|u)  \mac(y|x_1,x_2,s)}  \leq \eta \,,
\label{eq:MDrule1q}
\end{align}
and by Condition 2b) of the decoding rule for $\Dset(m_1,m_2)$,
\begin{align}
I(X_1,X_2,Y;\tX_1|U,S) 
=& \sum_{u,x_1,\tx_1,x_2,s,y} P_{U,X_1,\tX_1,X_2,S,Y}(u,x_1,\tx_1,x_2,s,y)
\cdot
\log 
\frac{P_{\tX_1|U,X_1,X_2,S,Y}(\tx_1|u,x_1,x_2,s,y)}{P_{\tX_1|U,S}(\tx_1|u,s)} 
\leq \eta_1 \,,
\label{eq:MDrule2bq}
\end{align}
where $U,X_1,\tX_1,X_2,S,Y$ are distributed according to the joint type of 
$u^n$, $f_1(m_1)$, $f_1(\breve{m}_1)$, $f_2(m_2)$, $s^n$, and $y^n$, for some $s^n\in\Sset^n$ with $l^n(s^n)\leq \Lambda$. 
 Adding (\ref{eq:MDrule1q}) and (\ref{eq:MDrule2bq}) yields
\begin{multline}
\sum_{u,x_1,\tx_1,x_2,s,y} P_{U,X_1,\tX_1,X_2,S,Y}(u,x_1,\tx_1,x_2,s,y) 
\cdot\log 
\frac{P_{U,X_1,\tX_1,X_2,S,Y}(u,x_1,\tx_1,x_2,s,y)}{P_U(u) P_1(x_1|u) P_{\tX_1,S|U}(\tx_1,s|u) P_2(x_2|u)\mac(y|x_1,x_2,s)} \\ \leq \eta+\eta_1 \,.
\end{multline}
That is, $D(P_{U,X_1,\tX_1,X_2,S,Y}||  P_U\times P_1\times P_1\times P_{S|U,\tX_1}\times P_2\times \mac )\leq \eta+\eta_1$. Therefore, by the log-sum inequality (see \eg \cite[Theorem 2.7.1]{CoverThomas:06b}),  
\begin{align}
&D(P_{U,X_1,\tX_1,X_2,Y}|| P_U\times P_1\times P_1 \times P_2 \times  V_{Y|U,X_1,\tX_1,X_2}  ) \nonumber\\
\leq& D(P_{U,X_1,\tX_1,X_2,S,Y}|| P_U\times P_1\times P_1\times P_{S|U,\tX_1}\times P_2\times \mac )\leq \eta+\eta_1
\,,
\end{align}
where $V_{Y|U,X_1,\tX_1,X_2}(y|u,x_1,\tx_1,x_2)=\sum_{s\in\Sset} \mac(y|x_1,x_2,s)P_{S|U,\tX_1}(s|u,\tx_1)$.
Then, by Pinsker's inequality (see \eg \cite[Problem 3.18]{CsiszarKorner:82b}),
\begin{multline}
\sum_{u,x_1,\tx_1,x_2,y} |P_{U,X_1,\tX_1,X_2,Y}(u,x_1,\tx_1,x_2,y)  -
 P_U(u) P_1(x_1|u)
 P_1(\tx_1|u)  P_2(x_2|u)  V_{Y|U,X_1,\tX_1,X_2}(y|u,x_1,\tx_1,x_2)|\\   \leq c\sqrt{\eta+\eta_1} \,,
\label{eq:MDrule2bqm1}
\end{multline}
where $c>0$ is a constant. 

Similarly, following our assumption that the decoding rules  for $\Dset(\breve{m}_1,m_2)$ hold as well, Conditions 1) and 2b) claim that
\begin{align}
&D(P_{U,\tX_1,X_2,S,Y}|| P_U\times P_{\tX_1|U}\times P_{X_2|U}\times P_{\tS|U} \times \mac)\leq \eta \,, 
\intertext{and}
&I(\tX_1,X_2,Y;X_1|U,\tS)\leq \eta_1 \,,
\end{align}
where $U,X_1,\tX_1,X_2,\tS,Y$ are distributed according to the joint type of 
$u^n$, $f_1(m_1)$, $f_1(\breve{m}_1)$, $f_2(m_2)$, $\ts^n$, and $y^n$, for some $\ts^n\in\Sset^n$ with $l^n(\ts^n)\leq \Lambda$. 
Here,  $\tX_1$ and $\tS$ have switched places with $X$ and $S$, respectively, since $f_1(\breve{m}_1)$ and $\ts^n$ are the tested codeword and state sequence, while $f_1(m_1)$ and $s^n$ are the competing ones. 
By the same arguments that led to (\ref{eq:MDrule2bqm1}), it follows that 
\begin{multline}
\sum_{u,x_1,\tx_1,x_2,y} |P_{U,X_1,\tX_1,X_2,Y}(u,x_1,\tx_1,x_2,s)-
 P_U(u)P_1(x_1|u)
 P_1(\tx_1|u)  P_2(x_2|u)  V_{Y|U,X_1,\tX_1,X_2}'(y|u,x_1,\tx_1,x_2)| \\ \leq c\sqrt{\eta+\eta_1} \,,
\label{eq:MDrule2bqbm1}
\end{multline}
where $ V_{Y|U,X_1,\tX_1,X_2}'(y|u,x_1,\tx_1,x_2)=\sum_{s\in\Sset} \mac(y|\tx_1,x_2,s)P_{\tS|U,X_1}(s|u,x_1)$.
Since $P_U(u)\geq\delta_0$ and $P_{X_k|U}(x_k|u)\geq \delta_k$, for all $u\in\Uset$ and $x_k\in\Xset_k$, for $k=1,2$, we have by
(\ref{eq:MDrule2bqm1}) and (\ref{eq:MDrule2bqbm1}) that
\begin{align}
\max_{u,x_1,\tx_1,x_2,y} 
\Big| V_{Y|U,X_1,\tX_1,X_2}(y|u,x_1,\tx_1,x_2) - V_{Y|U,X_1,\tX_1,X_2}'(y|u,x_1,\tx_1,x_2) \Big| 
\leq \frac{2c\sqrt{\eta+\eta_1}}{\delta_0\delta_1^2\delta_2} 
\,,
\label{eq:MVlowd}
\end{align}
Equivalently, the above can be expressed as
\begin{multline}
\max_{u,x_1,\tx_1,x_2,y} 
\Big|\sum_{s} P_{S|U,\tX_1}(s|u,\tx_1) \mac(y|x_1,x_2) 
-\sum_{s} P_{\tS|U,X_1}(s|u,x_1) \mac(y|\tx_1,x_2,s) \Big| 
\\
\leq
 \frac{2c\sqrt{\eta+\eta_1}}{\delta_0 \delta_1^2\delta_2 
} \,,
\label{eq:MVlowd1}
\end{multline}

Now, we show that the state distributions $Q=P_{S|U,\tX_1}$ and $Q'=P_{\tS|U,X_1}$ satisfy the conditions of 
Lemma~\ref{lemm:MA2}. Indeed, 
\begin{align}
&\max\left\{ \sum_{u,\tx_1,s}P_U(u) P_{1}(\tx_1|u)Q(s|u,\tx_1)l(s) ,\, \sum_{u,x_1,s} P_U(u) P_{1}(x_1|u)Q'(s|u,x_1)l(s) \right\} 
\nonumber\\
=& \max\left\{ \sum_{u,\tx_1,s} P_{U}(u) P_{1}(\tx_1|u)P_{S|U,\tX_1}(s|u,\tx_1)l(s) ,\, \sum_{u,x_1,s} P_U(u) P_{1}(x_1|u)P_{\tS|U,X_1}(s|u,x_1)l(s) \right\} 
\nonumber\\
=&\max\left\{ \sum_{s} P_S(s)l(s) ,\, \sum_{s} P_{\tS}(s)l(s)  \right\}
\nonumber\\
=&\max\left\{ l^n(s^n),\, l^n(\ts^n) \right\}\leq  \Lambda< \tLambda_1(P_{U,X_1}) \,,
\end{align}
where the last inequality is due to (\ref{eq:MdecLambda}). 
Thus, there exists $\xi>0$ such that (\ref{eq:MdecCNresA2}) holds with $Q=P_{S|U,\tX_1}$ and $Q'=P_{\tS|U,X_1}$, 
which contradicts (\ref{eq:MVlowd1}), if we choose $\eta$ and $\eta_1$ to be sufficiently small such that 
$\frac{2c\sqrt{\eta+\eta_1}}{\delta_0\delta_1^2\delta_2 }<\xi$.
\qed

\section{Proof of Lemma~\ref{lemm:McodeBsets}}
\label{app:McodeBsets}
Fix a sequence $u^n\in\Uset^n$ of type $P_U$. 
Let $\bar{ Z}^n(m)=(Z_1^n(m_1),Z_2^n(m_2))$,  $m\in [1:2^{n(R_1+R_2)}]$, be independent sequence pairs, uniformly distributed over the conditional type classes $\Tset^n(P_1)$ and $\Tset^n( P_2)$, 
where we have assigned an index $m\in [1:2^{n(R_1+R_2)}]$  to each message pair $(m_1,m_2)\in [1:2^{nR_1}]\times [1:2^{nR_2}]$.
Fix $a_1^n\in\Xset_1^n$, $a_2^n\in\Xset_2^n$, and $s^n\in\Sset^n$, and consider a joint type 
$P_{U,X_1,X_2,\tX_1,\tX_2,S}$, such that 
$P_{X_1,X_2|U}=P_{\tX_1,\tX_2|U}=P_1\times P_2$, \ie
\begin{align}
P_{X_1,X_2|U}(x_1,x_2|u)=P_{\tX_1,\tX_2|U}(x_1,x_2|u)=P_1(x_1|u)\cdot P_2(x_2|u) \,.
\end{align}
We intend to show that $\{Z_1^n(m_1)\}$ and $\{Z_2^n(m_2)\}$ satisfy each of the desired properties with double exponential high probability $(1-e^{-2^{\dB n}})$, $\dB>0$, implying that there exist deterministic codebooks that satisfy (\ref{eq:M11ebn})-(\ref{eq:M2c3ebn}) simultaneously.
This will only be shown for the properties in parts 1 and 2, since part 3 is symmetric with part 2.  

 We will use the following large deviations result by Csisz\'{a}r and Narayan 
\cite{CsiszarNarayan:88p}. 
\begin{lemma}[see {\cite[Lemma A1]{CsiszarNarayan:88p}}]
\label{lemm:MbookLD}
Let $\alpha,\beta\in [0,1]$, and consider a sequence of random vectors $ V^n(m)$, and functions $\varphi_m: \Xset^{nm}\rightarrow [0,1]$, for $m\in [1:\dM]$. If
\begin{align}
\E \left( \varphi_m( V^n(1)\,\ldots, V^n(m) ) \big|  V^n(1)\,\ldots,  V^n(m-1) \right) \leq \alpha \;\text{a.s., for
$m\in [1:\dM]$ } \,,
\end{align}
then
\begin{align}
\prob{\sum_{m=1}^{\dM} \varphi_m(  V^n(1)\,\ldots,  V^n(m) )>\dM \beta  }\leq 
 \exp\{ -\dM (\beta-\alpha\log e) \} \,.
\end{align}
\end{lemma}

\subsection*{Part 1}

To show that (\ref{eq:M11ebn}) holds, consider the indicator function
\begin{align}
\varphi_{m}(\bar{ Z}^n(1),\ldots,\bar{ Z}^n(m))=
\begin{cases}
1 &\text{if $(u^n, a_1^n, a_2^n,  Z_1^n(m_1), Z_2^n(m_2),s^n)\in\Tset^n(P_{U,X_1,X_2,\tX_1,\tX_2,S})$}
\\
0 &\text{otherwise}
\end{cases}
\label{eq:MindJ1}
\end{align}
where $\bar{ Z}^n(m)=(Z_1^n(m_1),Z_2^n(m_2))$ as defined above.
By standard type class considerations (see \eg \cite[Theorem 1.3]{Kramer:08n}), we have that
\begin{align}
\E \left[ \varphi_{m}(\bar{ Z}^n(1),\ldots,\bar{ Z}^n(m)  \big|
\bar{ Z}^n(1),\ldots,\bar{ Z}^n(m-1) \right] \leq&  
2^{-n\left(I(\tX_1,\tX_2;U,X_1,X_2,S)-\frac{\eps}{4} \right)} 
\label{eq:MphiZb11}
\\
\leq& 2^{-n\left(I(\tX_1,\tX_2;X_1,X_2,S|U)-\frac{\eps}{4} \right)} 
 \,,
\label{eq:MphiZb1}
\end{align}
where the last inequality holds since $I(\tX_1,\tX_2;U,X_1,X_2,S)\geq I(\tX_1,\tX_2;X_1,X_2,S|U)$.

Next, we use Lemma~\ref{lemm:MbookLD}, and plug
\begin{align}
&( V(1),\ldots, V(\dM)) \leftarrow (\bar{ Z}^n(1),\ldots,\bar{ Z}^n(2^{n(R_1+R_2)})) 
\,,\; \dM=2^{n(R_1+R_2)} \,,
\nonumber\\
& \alpha= 2^{-n\left(I(\tX_1,\tX_2;X_1,X_2,S|U)-\frac{\eps}{4} \right)}  \,,\;  \nonumber\\
& \beta=2^{n\left( \left[ R_1+R_2-I(\tX_1,\tX_2;X_1,X_2,S|U) \right]_{+} -(R_1+R_2)+\eps \right)} \,.
\end{align}
For sufficiently large $n$, we have that $\dM(\beta-\alpha\log e)\geq 2^{n\eps/2}$. Hence, by 
Lemma~\ref{lemm:MbookLD},
\begin{align}
&\prob{ \sum_{m=1}^{2^{n(R_1+R_2)}} \varphi_{m}(\bar{ Z}^n(1),\ldots,\bar{ Z}^n(2^{n(R_1+R_2)}))
>2^{n\left( \left[ R_1+R_2-I(\tX_1,\tX_2;X_1,X_2,S|U) \right]_{+} +\eps \right)  }} 
\leq 
 e^{- 2^{n\eps/2}}  \,.
\label{eq:MdLm1n}
\end{align}
The double exponential decay of the probability above 
implies that there exist codebooks which satisfy (\ref{eq:M11ebn}).

Similarly, to show  (\ref{eq:M12ebn}), we replace the indicator of the type $P_{U,X_1,X_2,\tX_1,\tX_2,S}$ in (\ref{eq:MindJ1}) by an indicator of the type $P_{U,\tX_1,\tX_2,S}$, and rewrite (\ref{eq:MphiZb11})
with 
$I(\tX_1,\tX_2;S|U)$,
 to obtain
\begin{align}
\Pr\Big( |\{ (\tm_1,\tm_2) \,:\; (u^n,Z_1^n(\tm_1),Z_2^n(\tm_2),s^n)\in\Tset^n(P_{U,\tX_1,\tX_2,S}) \}|>
2^{n\left( \left[ R_1+R_2-I(\tX_1,\tX_2;S|U) \right]_{+} +\eps_1 \right)}
 \Big) \;<  e^{- 2^{n\eps_1/2}}  \,,
\label{eq:Txs1}
\end{align}
where $\eps_1>0$ is arbitrarily small.
If $I(\tX_1,\tX_2;S|U)>\eps$ and $R_1+R_2\geq\eps$, then choosing $\eps_1=\frac{\eps}{2}$, we have that
\begin{align}
\left[ R_1+R_2-I(\tX_1,\tX_2;S|U) \right]_{+}+\eps_1 \leq R_1+R_2-\frac{\eps}{2} \,,
\end{align}
 hence,
\begin{align}
\Pr\Big( |\{ (\tm_1,\tm_2) \,:\; (u^n,Z_1^n(\tm_1),Z_2^n(\tm_2),s^n)\in\Tset^n(P_{U,\tX_1,\tX_2,S}) \}|>
2^{n\left(  R_1+R_2- \frac{\eps}{2} \right)}
 \Big) \;
<
e^{- 2^{n\eps/4}}  \,.
\end{align}

It remains to show that (\ref{eq:M13ebn}) holds. Assume that
\begin{align}
I(X_1,X_2;\tX_1,\tX_2,S|U)-\left[ R_1+R_2-I(\tX_1,\tX_2;S|U) \right]_{+}>\eps \,.
\label{eq:MencJassump1}
\end{align}
 Let $\Aset_m$ denote the set of indices 
$\tm<m$ such that $(u^n,\bar{ Z}^n(\tm),s^n)\in\Tset^n(P_{U,\tX_1,\tX_2,S})$, provided that their number does not exceed $2^{n\left(\left[ R_1+R_2-I(\tX_1,\tX_2;S|U) \right]_{+} +\frac{\eps}{8} \right)}$; else, let $\Aset_m=\emptyset$. Also, let
\begin{align}
\psi_m(\bar{ Z}^n(1),\ldots,\bar{ Z}^n(m))=\begin{cases}
1 &\text{if $(u^n,\bar{ Z}^n(m),\bar{ Z}^n(\tm),s^n)\in\Tset^n(P_{U,X_1,X_2,\tX_1,\tX_2,S})$}\\
  &\text{for some $\tm\in\Aset_m$}\,, \\
0 &\text{otherwise.}
\end{cases}
\end{align}
Then, choosing $\eps_1=\frac{\eps}{8}$ in (\ref{eq:Txs1}) yields
\begin{align}
&\Pr \Big( \sum_{m=1}^{2^{n(R_1+R_2)}} \psi_m(\bar{ Z}^n(1),\ldots,\bar{ Z}^n(m))\neq\; \nonumber\\&
|\{ m \,:\; (u^n,\bar{ Z}^n(m),\bar{ Z}^n(\tm),s^n)\in\Tset^n(P_{U,X_1,X_2,\tX_1,\tX_2,S}) \;\text{for some $\tm<m$} \}|    
\Big) \;<  e^{- 2^{n\eps/16}} \,.
\label{eq:MsetsEquiv1}
\end{align}
Therefore, instead of bounding the set of message pairs, it is sufficient to consider the sum $\sum
\psi_m(\bar{ Z}^n(1),\ldots,\bar{ Z}^n(m))$.
Furthermore, by standard type class considerations (see \eg \cite[Theorem 1.3]{Kramer:08n}), we have that
\begin{align}
&\E \left( \psi_m(\bar{ Z}^n(1),\ldots,\bar{ Z}^n(m)) \big| \bar{ Z}^n(1),\ldots,\bar{ Z}^n(m-1) \right) \leq
|\Aset_m|\cdot 2^{-n\left(I(X_1,X_2;\tX_1,\tX_2,S|U)-\frac{\eps}{8} \right)}
\nonumber\\
\leq& 2^{n\left(\left[ R_1+R_2-I(\tX_1,\tX_2;S|U) \right]_{+}-I(X_1,X_2;\tX_1,\tX_2,S|U)+\frac{\eps}{4} \right)}
< 2^{-3n\eps/4} \,,
\end{align}
where the last inequality is due to (\ref{eq:MencJassump1}). Thus, by Lemma~\ref{lemm:MbookLD},
\begin{align}
\prob{ \sum_{m=1}^{2^{n(R_1+R_2)}} \psi_m(\bar{ Z}^n(1),\ldots,\bar{ Z}^n(m))>
2^{n\left( R_1+R_2-\frac{\eps}{2} \right)} }<
e^{-2^{n\left(R_1+R_2-\frac{3\eps}{4}  \right)}}\leq e^{-2^{n\eps/4}} \,, 
\label{eq:Mb1j}
\end{align}
as we have assumed that $R_1+R_2\geq \eps$.
Equations (\ref{eq:MsetsEquiv1}) and (\ref{eq:Mb1j}) imply that the property in (\ref{eq:M13ebn}) holds with double exponential probability $1-e^{-2^{\dE_1\cdot n}}$, where $\dE_1>0$.

\subsection*{Part 2}
Fix $m_2\in [1:2^{nR_2}]$ and $ Z_2^n(m_2)=z_2^n\in \Tset^n(P_2)$. 
To show that (\ref{eq:M2b1ebn}) holds, consider the indicator
\begin{align}
\varphi_{ m_1}( Z_1^n(1),\ldots, Z_1^n(m_1))=
\begin{cases}
1 &\text{if $( u^n,a_1^n,z_2^n, Z_1^n(m_1),s^n)\in\Tset^n(P_{U,X_1,X_2,\tX_1,S})$
} 
\\
0 &\text{otherwise}
\end{cases}
\label{eq:MindJ2b}
\end{align}
By standard type class considerations (see \eg \cite[Theorem 1.3]{Kramer:08n}), we have that
\begin{align}
\E \left[ \varphi_{ m_1}( Z_1^n(1),\ldots, Z_1^n(m_1))  \big|
 Z_1^n(1),\ldots, Z_1^n(m_1-1) \right] \leq 
2^{-n\left(I(\tX_1;X_1,X_2,S|U)-\frac{\eps}{4} \right)}  \,.
\end{align}

Next, we use Lemma~\ref{lemm:MbookLD}, and plug
\begin{align}
&( V(1),\ldots, V(\dM)) \leftarrow ( Z_1^n(1),\ldots, Z_1^n(2^{nR_1})) 
\,,\; \dM=2^{nR_1} \,,
\nonumber\\
& \alpha= 2^{-n\left(I(\tX_1;X_1,X_2,S|U)-\frac{\eps}{4} \right)}  \,,\;  \nonumber\\
& \beta=2^{n\left( \left[ R_1-I(\tX_1;X_1,X_2,S|U) \right]_{+} -R_1+\eps \right)} \,.
\end{align}
For sufficiently large $n$, we have that $\dM(\beta-\alpha\log e)\geq 2^{n\eps/2}$. Hence, by 
Lemma~\ref{lemm:MbookLD},
\begin{align}
&\prob{ \sum_{m_1=1}^{2^{nR_1}} \varphi_{ m_1}( Z_1^n(1),\ldots, Z_1^n(2^{nR_1}))
>2^{n\left( \left[ R_1-I(\tX_1;X_1,X_2,S|U) \right]_{+} +\eps \right)  }} \leq 
 e^{- 2^{n\eps/2}}  \,.
\label{eq:MdLm2bn}
\end{align}
By the symmetry between $m_1$ and $\tm_1$ in the derivation above, the double exponential decay of the probability in (\ref{eq:MdLm2bn}) implies that there exist codebooks which satisfy (\ref{eq:M2b1ebn}).

Next, we show that (\ref{eq:M2b3ebn}) holds. 
Replacing the indicator of the type $P_{X_1,X_2,\tX_1,S}$ in (\ref{eq:MindJ2b}) with an indicator of the type $P_{\tX_1,S}$ yields
\begin{align}
\Pr\Big( |\{ \tm_1 \,:\; (u^n,x_1^n(\tm_1),s^n)\in\Tset^n(P_{U,\tX_1,S}) \}|>
2^{n\left( \left[ R_1-I(\tX_1;S|U) \right]_{+} +\eps_2 \right)}
 \Big) \;<  e^{- 2^{n\eps_2/2}}  \,,
\label{eq:Txs2b}
\end{align}
where $\eps_2>0$ is arbitrarily small.
 Assume that
\begin{align}
I(X_1,X_2;\tX_1,S|U)-\left[ R_1-I(\tX_1;S|U) \right]_{+}>\eps \,.
\label{eq:M2bencJassump}
\end{align}
 Let $\Aset_{ m_1}$ denote the set of indices 
$\tm_1<m_1$ such that $( Z_1^n(\tm_1),s^n)\in\Tset^n(P_{\tX_1,S})$, provided that their number does not exceed $2^{n\left(\left[ R_1-I(\tX_1;S) \right]_{+} +\frac{\eps}{8} \right)}$; else, let $\Aset_{ m_1}=\emptyset$. Also, let
\begin{align}
\psi_{ m_1}( Z_1^n(1),\ldots, Z_1^n(m_1))=\begin{cases}
1 &\text{if $(u^n, Z_1^n(m_1),z_2^n, Z_1^n(\tm_1),s^n)\in\Tset^n(P_{U,X_1,X_2,\tX_1,S})$, 
for some $\tm_1\in\Aset_{ m_1}$}\,, \\
0 &\text{otherwise.}
\end{cases}
\end{align}
Then, choosing $\eps_2=\frac{\eps}{8}$ in (\ref{eq:Txs2b}) yields
\begin{multline}
\Pr \Big( \sum_{m_1=1}^{2^{nR_1}} \psi_{ m_1}( Z_1^n(1),\ldots, Z_1^n(m_1))\neq\; 
|\{ m_1 \,:\;
 (u^n, Z_1^n(m_1),z_2^n, Z_1^n(\tm_1),s^n)\in\\ 
\Tset^n(P_{U,X_1,X_2,\tX_1,S}) \;\text{for some $\tm_1<m_1$} \}|    
\Big) \;<  e^{- 2^{n\eps/16}} \,.
\label{eq:MsetsEquiv2b}
\end{multline}
Therefore, instead of bounding the set of messages, it is sufficient to consider the sum $\sum
\psi_{ m_1}( Z_1^n(1),\ldots, Z_1^n(m_1))$.
Furthermore, by standard type class considerations (see \eg \cite[Theorem 1.3]{Kramer:08n}), we have that
\begin{align}
\E \left( \psi_{ m_1}( Z_1^n(1),\ldots, Z_1^n(m_1)) \big|  Z_1^n(1),\ldots, Z_1^n(m_1-1) \right) 
 \leq&
|\Aset_{ m_1}|\cdot 2^{-n\left(I(X_1,X_2;\tX_1,S|U)+\frac{\eps}{8} \right)}
\nonumber\\
\leq& 2^{n\left(\left[ R_1-I(\tX_1;S|U) \right]_{+}-I(X_1,X_2;\tX_1,S|U)+\frac{\eps}{4} \right)}
< 2^{-3n\eps/4} \,,
\end{align}
where the last inequality is due to (\ref{eq:M2bencJassump}). Thus, by Lemma~\ref{lemm:MbookLD},
\begin{align}
\prob{ \sum_{m_1=1}^{2^{nR_1}} \psi_{ m_1}( Z_1^n(1),\ldots, Z_1^n(m_1))>
2^{n\left(R_1-\frac{\eps}{2} \right)} }<
e^{-2^{n\left(R_1-\frac{3\eps}{4}  \right)}}\leq e^{-2^{n\eps/4}} \,, 
\label{eq:Mb2bj}
\end{align}
as we have assumed that $R_1+R_2\geq \eps$.
The double exponential bounds in (\ref{eq:MsetsEquiv2b}) and (\ref{eq:Mb2bj}) imply that there exists codebooks that satisfy (\ref{eq:M2b3ebn}) as well. 
\qed

\section{Proof of  Theorem~\ref{theo:MCavc}}
\label{app:MCavc}
To prove the theorem, we consider each case in the definition of the region $\MICavc$ separately (see Definition~\ref{def:MICavc}). Case A requires most of the effort, as the other cases follow from similar, yet simpler, considerations.

\subsection*{Case A}
Suppose that $L^*>\Lambda$, $L_1^*>\Lambda$ and $L_2^*>\Lambda$.

\subsubsection*{Achievability Proof}
Let $\eps>0$ be chosen later, and $u^n\in\Uset^n$ be a sequence in the type class of $P_U$, such that $P_U(u)>0$ $\forall$ $u\in\Uset$.
For $k=1,2$, let $P_{X_k|U}$ be a conditional type over $\Xset_k$, for which $P_{X_k|U}(x_k|u)>0$ $\forall x_k\in\Xset_k$, $u\in\Uset$, 
$\E \cost_k(X_k)\leq \plimit_k$,  with
\begin{align}
\tLambda_k(P_{U,X_k})>&\Lambda \,,\;
\intertext{and}
\tLambda(P_{U,X_1,X_2})>&\Lambda \,.
\end{align}
Furthermore, choose $\eta,\eta_1,\eta_2>0$ accordingly to be sufficiently small, such that Lemma~\ref{lemm:MdisDec} guarantees that the decoder in Definition~\ref{def:MLdecoder} is well defined.
%
Now, Lemma~\ref{lemm:McodeBsets} assures that there are codebooks, $\{x_1^n(m_1):m_1\in [1:2^{nR_1}]\}$ of type $P_{X_1|U}$, and $\{x_2^n(m_2):m_2\in [1:2^{nR_2}]\}$ of type $P_{X_2|U}$, which satisfy (\ref{eq:M11ebn})-(\ref{eq:M2c3ebn}).
Consider the following coding scheme.

\emph{Encoding}: To send $m_k\in [1:2^{nR_k}]$, Encoder $k$ transmits $x_k^n(m_k)$, for $k=1,2$.

\emph{Decoding}: Find a unique message pair $(\hm_1,\hm_2)$ such that the received sequence $y^n$ belongs to $\Dset(\hm_1,\hm_2)$, as in Definition~\ref{def:MLdecoder}. If there is none, declare an error.
Lemma~\ref{lemm:MdisDec} guarantees that there cannot be two message pairs for which this holds.

\emph{Analysis of Probability of Error}: Fix $s^n\in\Sset^n$ with $l^n(s^n)\leq\Lambda$, let $q=P_{S|U}$ denote the conditional type of $s^n$ given $u^n$, and let $(M_1,M_2)$ denote the transmitted message pair.
Consider the error events
\begin{align}
\Eset_{1}=&\{ D(P_{U,X_1,X_2,S,Y}||P_U\times P_{X_1|U}\times P_{X_2|U}\times P_{S|U} \times \mac)> \eta \}
\\
\Eset_{2\text{a}}=&\{ \text{Condition 2a) of the decoding rule is violated} \}
\\
\Eset_{2\text{b}}=&\{ \text{Condition 2b) of the decoding rule is violated} \}
\\
\Eset_{2\text{c}}=&\{ \text{Condition 2c) of the decoding rule is violated} \}
\end{align}
and 
\begin{align}
\Fset_1=&\{ I_q(U,X_1,X_2;S)>\eps \} \,, \\
\Fset_2=&\{ I_q(X_1,X_2;\tX_1,\tX_2,S|U)>\left[ R_1+R_2-I(\tX_1,\tX_2;S|U) \right]_{+}+\eps \,,
\; 
\text{for some $\tm_1\neq M_1$ and $\tm_2\neq M_2$}
\} \,,
\\
\Fset_3=&\{ I_q(X_1,X_2;\tX_1,S|U)>\left[ R_1-I(\tX_1;S|U) \right]_{+}+\eps \,,\;
\text{for some $\tm_1\neq M_1$}
\} \,,
\\
\Fset_4=&\{ I_q(X_1,X_2;\tX_2,S|U)>\left[ R_2-I(\tX_2;S|U) \right]_{+}+\eps \,,\;
\text{for some $\tm_2\neq M_2$}
\} \,,
\end{align}
where $(U,X_1,X_2,\tX_1,\tX_2,S)$ are dummy random variables, which are distributed as the joint type 
of $(x_1^n(M_1),x_2^n(M_2),x_1^n(\tm_1),$ $x_2^n(\tm_2),s^n)$. By the union of events bound,
\begin{align}
\cerr(\code)\leq& \prob{\Fset_1}+\prob{\Fset_2}+\prob{\Fset_3}+\prob{\Fset_4}
\nonumber\\&
+\prob{\Eset_{1}\cap\Fset_1^c}+
\prob{\Eset_{2\text{a}}\cap\Eset_1^c\cap\Fset_2^c}+
\prob{\Eset_{2\text{b}}\cap\Eset_1^c\cap\Fset_3^c}+
\prob{\Eset_{2\text{c}}\cap\Eset_1^c\cap\Fset_4^c} \,,
\end{align}
where the conditioning on $s^n$ is omitted for convenience of notation.
Based on Lemma~\ref{lemm:McodeBsets}, 
the probabilities of the events $\Fset_1$, $\Fset_2$, $\Fset_3$, and $\Fset_4$, tend to zero as $n\rightarrow\infty$,
by (\ref{eq:M12ebn}), (\ref{eq:M13ebn}), (\ref{eq:M2b3ebn}), and (\ref{eq:M2c3ebn}), respectively.

Now, suppose that Condition 1) of the decoding rule is violated. 
Observe that the event $\Eset_{1}\cap\Fset_1^c$ implies that 
\begin{align}
&D(P_{U,X_1,X_2,S,Y}||P_{U,X_1,X_2,S}\times \mac)
\nonumber\\
=& D(P_{U,X_1,X_2,S,Y}||P_U\times P_{X_1|U}\times P_{X_2|U}\times P_{S|U} \times \mac)-I(X_1,X_2;S|U)
>\eta-\eps \,.
\end{align}
Then,  by standard large deviations considerations (see \eg \cite[pp. 362--364]{CoverThomas:06b}),
\begin{align}
\prob{\Eset_{1}\cap\Fset_1^c}  
\leq& \max_{P_{U,X_1,X_2,S,Y}\,:\; \Eset_{1}\cap\Fset_1^c\;\text{holds}
} 2^{-n (D(P_{U,X_1,X_2,S,Y} || P_{U,X_1,X_2,S}\times \mac)-\eps)} 
< 2^{-n(\eta-2\eps)} \,,
\end{align}
which tends to zero as $n \rightarrow \infty$, for sufficiently small $\eps>0$, with $\eps<\frac{1}{2}\eta$.

Moving to Condition 2a) of the decoding rule, let $\Dset_{2\text{a}}$ denote the set of joint types 
$P_{U,X_1,X_2,\tX_1,\tX_2,S}$ such that
\begin{align}
&D(P_{U,X_1,X_2,S,Y}||P_U\times P_{X_1|U}\times P_{X_2|U}\times P_{S|U} \times \mac)\leq \eta
\label{eq:D2a11}
 \,, \\
&D(P_{U,\tX_1,\tX_2,\tS,Y}||P_U\times P_{\tX_1|U}\times P_{\tX_2|U}\times P_{\tS|U} \times \mac)\leq \eta \,,\;
\text{for some  $\tS\sim \tq(s|u)$} \,,
\label{eq:D2a2}
\\
&I_q(X_1,X_2,Y;\tX_1,\tX_2|U,S)>\eta \,.
\label{eq:D2a3}
\end{align}
Observe that the event $\Eset_1^c$ implies that (\ref{eq:D2a11}) holds.
Also, when the event $\Eset_{\text{2a}}$ occurs, \ie Condition 2a) of the decoding rule is violated, then there exist $\tm_1\neq m_1$ and $\tm_2\neq m_2$ such that
(\ref{eq:D2a2}) and (\ref{eq:D2a3}) hold for some $\ts^n\in\Sset^n$ with $l^n(\ts^n)\leq \Lambda$.
Then,  
by standard type class considerations (see \eg \cite[Theorem 1.3]{Kramer:08n}), 
\begin{align}
&\cprob{\Eset_{2\text{a}}\cap\Eset_1^c\cap\Fset_2^c}{M_1=m_1,M_2=m_2} \nonumber\\ \leq&
\sum_{  \substack{P_{U,X_1,X_2,\tX_1,\tX_2,S}\in \Dset_{2\text{a}} \,:\; \\   \Fset_2^c \;\text{holds}  }   }
|\{(\tm_1,\tm_2)\,:\;
 (u^n,x_1^n(m_1),x_2^n(m_2),x_1^n(\tm_1),x_2^n(\tm_2),s^n)\in \Tset^n(P_{U,X_1,X_2,\tX_1,\tX_2,S}) \}|
\nonumber\\
&\times 2^{-n\left( I_q(\tX_1,\tX_2;Y|U,X_1,X_2,S)-\eps  \right)} \,,
\end{align}
for every given $m_1\in [1:2^{nR_1}]$ and $m_2\in [1:2^{nR_2}]$. Hence, by (\ref{eq:M11ebn}),
\begin{align}
&\prob{\Eset_{2\text{a}}\cap\Eset_1^c\cap\Fset_2^c}\leq 
\sum_{  \substack{P_{U,X_1,X_2,\tX_1,\tX_2,S}\in \Dset_{2\text{a}} \,:\;  \\   \Fset_2^c \;\text{holds}  }   } 
 2^{-n\left( I_q(\tX_1,\tX_2;Y|U,X_1,X_2,S)
-\left[ R_1+R_2-I_q(\tX_1,\tX_2;X_1,X_2,S|U) \right]_{+}-2\eps  \right)} \,.
\label{eq:E2E02cB}
\end{align}

To further bound $\prob{\Eset_{2\text{a}}\cap\Eset_1^c\cap\Fset_2^c}$, consider the following cases.
Suppose that $R_1+R_2\leq I_q(\tX_1,\tX_2;S|U)$. Then, given $\Fset_2^c$, we have that
\begin{align}
I_q(X_1,X_2;\tX_1,\tX_2|U,S)\leq I_q(X_1,X_2;\tX_1,\tX_2,S|U)\leq \eps \,.
\end{align}
By (\ref{eq:D2a3}), it then follows that
\begin{align}
I_q(\tX_1,\tX_2;Y|U,X_1,X_2,S)=&I_q(\tX_1,\tX_2;X_1,X_2,Y|U,S)-I_q(\tX_1,\tX_2;X_1,X_2|U,S) 
\nonumber\\
\geq& \eta-\eps \,.
\label{eq:Mrule2aIcase1}
\end{align}
Returning to (\ref{eq:E2E02cB}), we note that since the number of types is polynomial in $n$, the cardinality of the set of types $\Dset_{2\text{a}}$ can be bounded by $2^{n\eps}$, for sufficiently large $n$. Hence, by (\ref{eq:E2E02cB}) and (\ref{eq:Mrule2aIcase1}), we have that $\prob{\Eset_{2\text{a}}\cap\Eset_1^c\cap\Fset_2^c}\leq 2^{-n(\eta-4\eps)}$, which tends to zero as $n\rightarrow\infty$, for $\eps<\frac{1}{4}\eta$.

Otherwise, if $R_1+R_2> I_q(\tX_1,\tX_2;S|U)$, then given $\Fset_2^c$,
\begin{align}
R_1+R_2>&I_q(X_1,X_2;\tX_1,\tX_2,S|U)+I(\tX_1,\tX_2;S|U)-\eps 
\nonumber\\
=& I_q(\tX_1,\tX_2;X_1,X_2,S|U)+I(X_1,X_2;S|U)-\eps 
\nonumber\\
\geq& I_q(\tX_1,\tX_2;X_1,X_2,S|U)-\eps \,.
\end{align}
Thus,
\begin{align}
\left[ R_1+R_2-I_q(\tX_1,\tX_2;X_1,X_2,S|U) \right]_{+}\leq 
R_1+R_2-I_q(\tX_1,\tX_2;X_1,X_2,S|U)+\eps \,.
\end{align}
Hence, by (\ref{eq:E2E02cB}) we have that
\begin{align}
\prob{\Eset_{2\text{a}}\cap\Fset_2^c}\leq& \sum_{  \substack{P_{U,X_1,X_2,\tX_1,\tX_2,S}\in \Dset_{2\text{a}}  \\   \Fset_2^c \;\text{holds}  }   } 2^{-n(I(\tX_1,\tX_2;X_1,X_2,S,Y|U)-R_1-R_2-3\eps )}
\nonumber\\
\leq& \sum_{  \substack{P_{U,X_1,X_2,\tX_1,\tX_2,S}\in \Dset_{2\text{a}} \,:\; \\   \Fset_2^c \;\text{holds}  }   } 2^{-n(I_q(\tX_1,\tX_2;Y|U)-R_1-R_2-3\eps )} \,.
\end{align}
For $P_{U,X_1,X_2,\tX_1,\tX_2,S}\in \Dset_{2\text{a}}$, we have by (\ref{eq:D2a2}) that
$P_{\tX_1,\tX_2,\tS,Y|U}$ is arbitrarily close to some 
$P_{X_1,X_2,\tS,\tY|U}$, where
\begin{align}
P_{X_1,X_2,\tS,\tY|U}(x_1,x_2,s,y|u)=P_{X_1|U}(x_1|u)P_{X_2|U}(x_2|u)\tq(s|u)\mac(y|x_1,x_2,s) \,,
\end{align}
if $\eta>0$ is sufficiently small. In which case, 
\begin{align}
I_q(\tX_1,\tX_2;Y|U)\geq I_{\tq}(X_1,X_2;Y|U)-\delta \,,
\end{align}
where $\delta>0$ is arbitrarily small.
Therefore, provided that
\begin{align}
R_1+R_2 <& \min_{q(s|u) \,:\; \E_q l(S)\leq\Lambda} I_q(X_1,X_2;Y|U)-\delta-5\eps
\,,
\end{align}
we have that $\prob{\Eset_{2\text{a}}\cap\Fset_2^c}\leq 2^{-n(I_q(\tX_1,\tX_2;Y|U)-R_1-R_2-4\eps )}$ tends to zero as $n\rightarrow\infty$.

Next, consider Condition 2b) of the decoding rule, and let $\Dset_{2\text{b}}$ denote the set of joint types 
$P_{U,X_1,X_2,\tX_1,S}$ such that
\begin{align}
&D(P_{U,X_1,X_2,S,Y}||P_U\times P_{X_1|U}\times P_{X_2|U}\times P_{S|U} \times \mac)\leq \eta \,, \\
&D(P_{U,\tX_1,X_2,\tS,Y}||P_U\times P_{\tX_1|U}\times P_{X_2|U}\times P_{\tS|U} \times \mac)\leq \eta \,,\;
\text{for some $\tS\sim \tq(s|u)$} 
\label{eq:D2b2}
\\
&I_q(X_1,X_2,Y;\tX_1|U,S)>\eta \,.
\label{eq:D2b3}
\end{align}
Observe that when 
the event $\Eset_{\text{2b}}$ occurs, \ie Condition 2b) of the decoding rule is violated, then there exists $\tm_1\neq m_1$  such that
(\ref{eq:D2b2}) and (\ref{eq:D2b3}) hold for some $\ts^n\in\Sset^n$ with $l^n(\ts^n)\leq \Lambda$.
Then,  by standard type class considerations (see \eg \cite[Theorem 1.3]{Kramer:08n}), 
\begin{align}
&\cprob{\Eset_{2\text{b}}\cap\Eset_1^c\cap\Fset_3^c}{M_1=m_1,M_2=m_2} \nonumber\\ \leq&
\sum_{  \substack{P_{U,X_1,X_2,\tX_1,S}\in \Dset_{2\text{b}} \,:\;  \\   \Fset_3^c \;\text{holds}  }   }
|\{ \tm_1 \,:\;
 (u^n,x_1^n(m_1),x_2^n(m_2),x_1^n(\tm_1),s^n)\in \Tset^n(P_{U,X_1,X_2,\tX_1,S}) \}|
\cdot
 2^{-n\left( I_q(\tX_1;Y|U,X_1,X_2,S)-\eps  \right)} \,,
\end{align}
for every given $m_1\in [1:2^{nR_1}]$ and $m_2\in [1:2^{nR_2}]$. Hence, by (\ref{eq:M2b1ebn}),
\begin{align}
&\prob{\Eset_{2\text{b}}\cap\Eset_1^c\cap\Fset_3^c}\leq
\sum_{  \substack{P_{U,X_1,X_2,\tX_1,S}\in \Dset_{2\text{b}}  \\   \Fset_{3}^c \;\text{holds}  }   } 
 2^{-n\left( I_q(\tX_1;Y|U,X_1,X_2,S)
-\left[ R_1-I_q(\tX_1;X_1,X_2,S|U) \right]_{+}-2\eps  \right)} \,.
\label{eq:E2bE02cB}
\end{align}

If $R_1\leq I_q(\tX_1;S|U)$, then given $\Fset_3^c$, we have that
\begin{align}
I_q(X_1,X_2;\tX_1|U,S)\leq I_q(X_1,X_2;\tX_1,S|U)\leq \eps \,.
\end{align}
Hence, by (\ref{eq:D2b3}), 
\begin{align}
I_q(\tX_1;Y|U,X_1,X_2,S)=I_q(\tX_1;X_1,X_2,Y|U,S)-I_q(\tX_1;X_1,X_2|U,S) \geq \eta-\eps \,.
\label{eq:Mrule2bIcase1}
\end{align}
Returning to (\ref{eq:E2bE02cB}), we note that since the number of types is polynomial in the sequence length, $|\Dset_{2\text{b}}|\leq 2^{n\eps}$
for sufficiently large $n$. Hence, by (\ref{eq:E2bE02cB}) and (\ref{eq:Mrule2bIcase1}), we have that 
$\prob{\Eset_{2\text{b}}\cap\Fset_3^c}\leq 2^{-n(\eta-4\eps)}$, which tends to zero as $n\rightarrow\infty$, for $\eps<\frac{1}{4}\eta$.

Otherwise, if $R_1> I_q(\tX_1;S|U)$, then given $\Fset_3^c$,
\begin{align}
R_1>&I_q(X_1,X_2;\tX_1,S|U)+I(\tX_1;S|U)-\eps 
\nonumber\\
=& I_q(\tX_1;X_1,X_2,S|U)+I(X_1,X_2;S|U)-\eps 
\nonumber\\
\geq& I_q(\tX_1;X_1,X_2,S|U)-\eps \,.
\end{align}
Thus,
\begin{align}
\left[ R_1-I_q(\tX_1;X_1,X_2,S|U) \right]_{+}\leq 
R_1-I_q(\tX_1;X_1,X_2,S|U)+\eps \,.
\end{align}
Hence, by (\ref{eq:E2bE02cB}) we have that
\begin{align}
\prob{\Eset_{2\text{b}}\cap\Eset_1^c\cap\Fset_3^c}\leq& \sum_{  \substack{P_{U,X_1,X_2,\tX_1,S}\in \Dset_{2\text{b}}  \\   \Fset_3^c \;\text{holds}  }   } 2^{-n(I_q(\tX_1;X_1,X_2,S,Y|U)-R_1-3\eps )}
\nonumber\\
\leq& \sum_{  \substack{P_{U,X_1,X_2,\tX_1,S}\in \Dset_{2\text{b}} \,:\; \\   \Fset_3^c \;\text{holds}  }   } 2^{-n(I_q(\tX_1;Y|X_2,U)-R_1-3\eps )} \,,
\end{align}
where the last inequality holds since
\begin{align}
I_q(\tX_1;X_1,X_2,S,Y|U)=&
I_q(\tX_1;Y|X_2,U)+I_q(\tX_1;X_2|U)+I_q(\tX_1;X_1,S|X_2,Y,U) \nonumber\\
\geq&  I_q(\tX_1;Y|X_2,U)\,.
\end{align}
For $P_{U,X_1,X_2,\tX_1,S}\in \Dset_{2\text{b}}$, we have by (\ref{eq:D2b2}) that
$P_{\tX_1,X_2,\tS,Y|U}$ is arbitrarily close to some 
$P_{X_1,X_2,\tS,\tY|U}$, where
\begin{align}
P_{X_1,X_2,\tS,\tY|U}(x_1,x_2,s,y|u)=P_{X_1|U}(x_1|u)P_{X_2|U}(x_2|u)\tq(s|u)\mac(y|x_1,x_2,s) \,,
\end{align}
if $\eta>0$ is sufficiently small. In which case, 
\begin{align}
I_q(\tX_1;Y|X_2,U)\geq I_{\tq}(X_1;Y|X_2,U)-\delta_1 \,,
\end{align}
where $\delta_1>0$ is arbitrarily small.
Therefore, provided that
\begin{align}
R_1 < \min_{q(s|u) \,:\; \E_q l(S)\leq\Lambda} I_q(X_1;Y|X_2,U)-\delta_1-5\eps
\end{align}
we have that $\prob{\Eset_{2\text{b}}\cap\Fset_3^c}\leq 2^{-n(I_q(\tX_1;Y|X_2,U)-R_1-4\eps )}$ tends to zero as $n\rightarrow\infty$.

In same manner, it can be shown that $\prob{\Eset_{2\text{c}}\cap\Fset_4^c}$ tends to zero as $n\rightarrow\infty$, provided that 
\begin{align}
R_2 < \min_{q(s|u) \,:\; \E_q l(S)\leq\Lambda} I_q(X_2;Y|X_1,U)-\delta_2-5\eps \,,
\end{align}
where $\delta_2>0$ is arbitrarily small.
\qed

\subsubsection*{Converse Proof}
We will use the following lemma, based on the observations of Gubner \cite{Gubner:90p}.
\begin{lemma} 
\label{lemm:Gubner}
Consider the AVMAC free of state constraints, and let  $\code=(f_1,f_2,g)$ be a
$(2^{nR_1},2^{nR_2},n)$ deterministic code.
\begin{enumerate}[1)]
\item
Suppose that $\mac$ is symmetrizable-$\Xset_1\times\Xset_2$, and let $J_i(s|x_1,x_2)$, $i\in [1:n]$, be a set of conditional state distributions that satisfy (\ref{eq:MsymmetrizableJ}).
If $R_1+R_2>0$, then
\begin{align}
&\err(\tq,\code) \geq \frac{1}{4} \,,\;
\intertext{for} 
&\tq(s^n)= \frac{1}{2^{n(R_1+R_2)}} \sum_{m_1=1}^{2^{nR_1}} \sum_{m_2=1}^{2^{nR_2}} J^n(s^n|f_1(m_1),f_2(m_2)) 
\label{eq:MconvFtq}
\,,
\end{align}
where $J^n(s^n|x_1^n,x_2^n)=\prod_{i=1}^n J_i(s_i|x_{1,i},x_{2,i})$.
\item
Suppose that $\mac$ is symmetrizable-$\Xset_1|\Xset_2$, and let $J_{1,i}(s|x_1)$, $i\in [1:n]$, be a set of conditional state 
distributions that satisfy (\ref{eq:Msymmetrizable1}).
If $R_1>0$, then
\begin{align}
&\err(\tq_1,\code) \geq \frac{1}{4} \,,\;
\intertext{for} 
&\tq_1(s^n)= \frac{1}{2^{nR_1}} \sum_{m_1=1}^{2^{nR_1}}  J^n_1(s^n|f_1(m_1)) 
\,,
\label{eq:MconvFtq1}
\end{align}
where $J_1^n(s^n|x_1^n)=\prod_{i=1}^n J_{1,i}(s_i|x_{1,i})$.
\item
Suppose that $\mac$ is symmetrizable-$\Xset_2|\Xset_1$, and let $J_{2,i}(s|x_2)$, $i\in [1:n]$, be a set of conditional state distributions that satisfy (\ref{eq:Msymmetrizable2}).
If $R_2>0$, then
\begin{align}
&\err(\tq_2,\code) \geq \frac{1}{4} \,,\;
\intertext{for} 
&\tq_2(s^n)= \frac{1}{2^{nR_2}} \sum_{m_2=1}^{2^{nR_2}}  J^n_2(s^n|f_2(m_2)) 
\,,
\label{eq:MconvFtq2}
\end{align}
where $J_2^n(s^n|x_2^n)=\prod_{i=1}^n J_{2,i}(s_i|x_{2,i})$.
\end{enumerate}
\end{lemma}
For completeness, we give the proof below.
\begin{proof}[Proof of Lemma~\ref{lemm:Gubner}]
Denote the codebooks size by  $\dM_k=2^{nR_k}$, $k=1,2$, $\dM=2^{n(R_1+R_2)}$, and the codewords by $x_k^n(m_k)=f_k(m_k)$, $k=1,2$.

Under the conditions of part 1,
\begin{align}
& \err(\tq,\code)=\sum_{s^n\in\Sset^n} q(s^n) \frac{1}{\dM} \sum_{m_1,m_2} 
\sum_{y^n \,:\; g(y^n)\neq (m_1,m_2)} W^n(y^n|x_1^n(m_1),x_2^n(m_2),s^n)
\nonumber\\
=& \frac{1}{\dM^2} \sum_{\tm_1,\tm_2} \sum_{s^n\in\Sset^n} J^n(s^n|x_1^n(\tm_1),x_2^n(\tm_2)) \sum_{m_1,m_2} 
\sum_{y^n \,:\; g(y^n)\neq (m_1,m_2)} W^n(y^n|x_1^n(m_1),x_2^n(m_2),s^n)
\end{align}
where have defined $W^n\equiv W_{Y^n|X_1^n,X_2^n,S^n}$ for short notation.
By switching between the summation indices $(m_1,m_2)$ and $(\tm_1,\tm_2)$, we obtain
\begin{align}
 \err(\tq,\code)
=& \frac{1}{2\dM^2} \sum_{m_1,m_2,\tm_1,\tm_2}\; \sum_{y^n \,:\; g(y^n)\neq (m_1,m_2)} 
\sum_{s^n\in\Sset^n} W^n(y^n|x_1^n(m_1),x_2^n(m_2),s^n) J^n(s^n|x_1^n(\tm_1),x_2^n(\tm_2))  
\nonumber\\
+& \frac{1}{2\dM^2} \sum_{m_1,m_2,\tm_1,\tm_2}\; \sum_{y^n \,:\; g(y^n)\neq (\tm_1,\tm_2)} 
\sum_{s^n\in\Sset^n} W^n(y^n|x_1^n(\tm_1),x_2^n(\tm_2),s^n) J^n(s^n|x_1^n(m_1),x_2^n(m_2))  \,.
\end{align}
Now, as the channel is memoryless, 
\begin{align}
&\sum_{s^n\in\Sset^n} W^n(y^n|x_1^n(\tm_1),x_2^n(\tm_2),s^n) J^n(s^n|x_1^n(m_1),x_2^n(m_2)) 
\nonumber\\ =&
\prod_{i=1}^n \sum_{s_i\in\Sset} \mac(y_i|x_{1,i}(\tm_1),x_{2,i}(\tm_2),s_i)J_i(s_i|x_{1,i}(m_1),x_{2,i}(m_2))
\nonumber\\ =&
\prod_{i=1}^n \sum_{s_i\in\Sset} \mac(y_i|x_{1,i}(m_1),x_{2,i}(m_2),s_i)J_i(s_i|x_{1,i}(\tm_1),x_{2,i}(\tm_2))
\nonumber\\ =&
\sum_{s^n\in\Sset^n} W^n(y^n|x_1^n(m_1),x_2^n(m_2),s^n) J^n(s^n|x_1^n(\tm_1),x_2^n(\tm_2)) \,,
\end{align}
where the second equality is due to (\ref{eq:MsymmetrizableJ}). Therefore, 
\begin{align}
 \err(\tq,\code)
\geq & \frac{1}{2\dM^2} \sum_{(\tm_1,\tm_2)\neq (m_1,m_2)}\; \sum_{s^n\in\Sset^n}
\Big[ \sum_{y^n \,:\; g(y^n)\neq (m_1,m_2)} 
 W^n(y^n|x_1^n(m_1),x_2^n(m_2),s^n) J^n(s^n|x_1^n(\tm_1),x_2^n(\tm_2))  
\nonumber\\&
+\sum_{y^n \,:\; g(y^n)\neq (\tm_1,\tm_2)} 
 W^n(y^n|x_1^n(m_1),x_2^n(m_2),s^n) J^n(s^n|x_1^n(\tm_1),x_2^n(\tm_2))
\Big]
\nonumber\\
\geq& \frac{1}{2\dM^2} \sum_{(\tm_1,\tm_2)\neq (m_1,m_2)}\; \sum_{s^n\in\Sset^n}
\sum_{y^n\in\Yset^n} 
 W^n(y^n|x_1^n(m_1),x_2^n(m_2),s^n) J^n(s^n|x_1^n(\tm_1),x_2^n(\tm_2))  
\nonumber\\
=& \frac{\dM(\dM-1)}{2\dM^2}=\frac{1}{2}\left( 1-\frac{1}{\dM} \right) \,.
\end{align}
Assuming  the sum rate is positive, we have that $\dM\geq 2$, hence $\err(\tq,\code)\geq
\frac{1}{4}$.

Under the conditions of part 2,
\begin{align}
 \err(\tq_1,\code)
=& \frac{1}{\dM_1^2\dM_2} \sum_{\tm_1} \sum_{s^n\in\Sset^n} J^n_1(s^n|x_1^n(\tm_1)) \sum_{m_1,m_2} 
\sum_{y^n \,:\; g(y^n)\neq (m_1,m_2)} W^n(y^n|x_1^n(m_1),x_2^n(m_2),s^n) \,.
\end{align}
Then, switching between $m_1$ and $\tm_1$ yields
\begin{align}
 \err(\tq_1,\code)
=& \frac{1}{2\dM_1^2\dM_2} \sum_{m_1,m_2,\tm_1}\; \sum_{y^n \,:\; g(y^n)\neq (m_1,m_2)} 
\sum_{s^n\in\Sset^n} W^n(y^n|x_1^n(m_1),x_2^n(m_2),s^n) J_1^n(s^n|x_1^n(\tm_1))  
\nonumber\\
&+ \frac{1}{2\dM_1^2\dM_2} \sum_{m_1,m_2,\tm_1}\; \sum_{y^n \,:\; g(y^n)\neq (\tm_1,m_2)} 
\sum_{s^n\in\Sset^n} W^n(y^n|x_1^n(\tm_1),x_2^n(m_2),s^n) J_1^n(s^n|x_1^n(m_1))  \,.
\end{align}
Now, as the channel is memoryless, we have by (\ref{eq:Msymmetrizable1}) that
\begin{align}
&\sum_{s^n\in\Sset^n} W^n(y^n|x_1^n(\tm_1),x_2^n(m_2),s^n) J_1^n(s^n|x_1^n(m_1)) 
\nonumber\\ =&
\sum_{s^n\in\Sset^n} W^n(y^n|x_1^n(m_1),x_2^n(m_2),s^n) J_1^n(s^n|x_1^n(\tm_1)) \,.
\end{align}
 Therefore, 
\begin{align}
 \err(\tq_1,\code)
\geq & \frac{1}{2\dM_1^2\dM_2} \sum_{\tm_1\neq m_1,m_2}\; \sum_{s^n\in\Sset^n}
\Big[ \sum_{y^n \,:\; g(y^n)\neq (m_1,m_2)} 
 W^n(y^n|x_1^n(m_1),x_2^n(m_2),s^n) J_1^n(s^n|x_1^n(\tm_1))  
\nonumber\\&
+\sum_{y^n \,:\; g(y^n)\neq (\tm_1,m_2)} 
 W^n(y^n|x_1^n(m_1),x_2^n(m_2),s^n) J_1^n(s^n|x_1^n(\tm_1))
\Big]
\nonumber\\
\geq& \frac{1}{2\dM_1^2\dM_2} \sum_{\tm_1\neq m_1,m_2}\; \sum_{s^n\in\Sset^n}
\sum_{y^n\in\Yset^n} 
 W^n(y^n|x_1^n(m_1),x_2^n(m_2),s^n) J_1^n(s^n|x_1^n(\tm_1))  
\nonumber\\
=& \frac{\dM_1(\dM_1-1)\dM_2}{2\dM_1^2}\geq \frac{1}{2}\left( 1-\frac{1}{\dM_1} \right) \,.
\end{align}
Since User 1 has a positive rate, $R_1>0$, the corresponding codebook has size $\dM_1\geq 2$, hence $\err(q,\code)\geq
\frac{1}{4}$.

The proof of part 3 is similar, and thus omitted.
\end{proof}

Now, we are in position to prove the converse part of Theorem~\ref{theo:MCavc}.
Consider a sequence of $(2^{nR_1},2^{nR_2},n,\alpha_n)$ deterministic codes $\code_n$ over the AVMAC under input constraints $(\plimit_1,\plimit_2)$ and state constraint $\Lambda$, where $\alpha_n\rightarrow 0$ as
$n\rightarrow \infty$. 
In particular, we have that the conditional probability of error given a state sequence $s^n$ is bounded by
\begin{align}
\cerr(\code_n)\leq \alpha_n \,,\;\text{for $s^n\in\Sset^n$ with $l^n(s^n)\leq\Lambda$} \,.
\label{eq:MStateConverse1bDet}
\end{align}
For simplicity, we assume that both $R_1>0$ and $R_2>0$, but the proof can be easily modified elsewhere.

First, we show that 
\begin{align}
R_1\leq& \min_{q(s|u) \,:\; E_q l(S)\leq \Lambda} I_q(X_1;Y|X_2,U)+\eps_n \,, \label{eq:MavcConvR1u} \\
R_2\leq& \min_{q(s|u) \,:\; E_q l(S)\leq \Lambda} I_q(X_2;Y|X_1,U)+\eps_n \,, \label{eq:MavcConvR2u}\\
R_1+R_2 \leq& \min_{q(s|u) \,:\; E_q l(S)\leq \Lambda} I_q(X_{1},X_2;Y|U)+\eps_n \label{eq:MavcConvSumu}\,,
\end{align}
where $\eps_n>0$ tends to zero as $n\rightarrow\infty$. To this end,	consider using the same code in the following setting. 
Consider a different channel model, with an average state constraint. 
Specifically, consider a MAC where the jammer selects an independent state sequence at random, $\oS^n\sim \prod_{i=1}^n \oq_i(z_i)$, under the \emph{average} state constraint $\frac{1}{n}\sum_{i=1}^n \E l(S_i)\leq\Lambda-\delta$. 
Here, there is no state constraint with probability $1$, as the jammer may select a sequence $\oS^n$ with $l^n(\oS^n)>\Lambda$. 
We claim that the code sequence of the constrained AVMAC achieves the same rate pair $(R_1,R_2)$ over the ``new"  MAC $W_{Y|X_1,X_2,\oS}$, which is governed by the state sequence $\oS^n$, under an average constraint. Indeed, using the code $\code_n$ over the MAC  $W_{Y|X_1,X_2,\oS}$, the probability of error is given by
\begin{align}
\err(\oq,\code_n)=&\sum_{s^n\in\Sset^n} \oq^n(s^n)\cerr(\code_n)
\nonumber\\
\leq& 
\sum_{s^n\in\Sset^n \,:\; l^n(s^n)\leq\Lambda} \oq^n(s^n)\cerr(\code_n)
+\prob{l^n(\oS^n)\geq\Lambda} \,.
\label{eq:MConv2}
\end{align} 
By (\ref{eq:MStateConverse1bDet}), we have that the sum in the RHS 
is bounded by $\alpha_n$, hence tends to zero as $n\rightarrow\infty$. As for the second term, 
\begin{align}
\prob{l^n(\oS^n)\geq\Lambda}\leq \prob{ \frac{1}{n}\sum_{i=1}^n (l(\oS_i)-\E l(\oS_i))\geq \delta }\leq 
\frac{\sum_{i=1}^n \var\left(  l(\oS_i) \right)} {n^2\delta^2}\leq \frac{l_{max}^2}{n\delta^2}
\label{eq:Hconv1}
\end{align}
where the first inequality holds since $\frac{1}{n}\sum_{i=1}^n \E l(\oS_i)\leq\Lambda-\delta$, 
and the second is due to Chebyshev's inequality. 
Thus, we have by (\ref{eq:MConv2}) that the probability of error tends to zero as $n\rightarrow\infty$, when using the code $\code_n$ over the MAC governed by $\oS^n$.

Therefore, it suffices to prove the converse part for the MAC $W_{Y|X_1,X_2,\oS}$ governed by the state sequence
$\oS^n\sim \oq^n(s^n)= \prod_{i=1}^n \oq_i(s_i)$.
Then, let $X_1^n=f_{1}^n(M_1)$ and $X_2^n=f_{2}^n(M_2)$ be the channel input sequences, and 
$Y^n$ be the corresponding output sequence.
Fano's inequality implies that
for every jamming strategy $\oq^n(s^n)$,
\begin{align}
R_1\leq&  \frac{1}{n}\sum_{i=1}^n I_{\oq_i}(X_{1,i};Y_i|X_{2,i})+\eps_n\\
R_2\leq&  \frac{1}{n} \sum_{i=1}^n I_{\oq_i}(X_{2,i};Y_i|X_{1,i})+\eps_n \\
R_1+R_2 \leq& \frac{1}{n}\sum_{i=1}^n I_{\oq_i}(X_{1,i},X_{2,i};Y_i)+\eps_n \,.
\end{align}
(see \cite[Section 15.3.4]{CoverThomas:06b}).
Let $U$ be a random variable which is uniformly distributed over $[1:n]$, and independent of $(X_1^n,X_2^n,S^n,Y^n)$. Then, the bounds can be expressed as
\begin{align}
R_1\leq&  I_{q}(X_{1,U};Y_U|X_{2,U},U)+\eps_n \,, 			\label{eq:MConvR2} \\
R_2\leq& I_{q}(X_{2,U};Y_U|X_{1,U},U)+\eps_n \,, 			\label{eq:MConvR1}	\\
R_1+R_2 \leq& I_{q}(X_{1,U},X_{2,U};Y_U|U)+\eps_n \,,	\label{eq:MConvSum}
\end{align}
where we have defined  $q(s|u)=\oq_u$ for $u\in [1:n]$.
Then, the bounds (\ref{eq:MConvR2})-(\ref{eq:MConvSum}) hold for every conditional state distribution $q(s|u)$ such that 
$\E l(S_U)\leq \Lambda$. 
Thus, the bounds in (\ref{eq:MavcConvR1u})-(\ref{eq:MavcConvSumu}) follow by 
defining 
 \begin{align}
X_1=X_{1,U} \,,\; X_2=X_{2,U} \,,\;\text{and}\; Y=Y_U \,,
\label{eq:XYUconvAVMAC}
\end{align} 
Note that $X_1$ and $X_2$ are conditionally independent given $U$, as required. 
%

Returning to the original AVMAC, we now show that $\tLambda(P_{U,X_1,X_2})\geq \Lambda$.
If the AVMAC is non-symmetrizable-$\Xset_1\times\Xset_2$, then  $\tLambda(P_{U,X_1,X_2})=+\infty$, and there is nothing to show. Hence, consider the case where the AVMAC is symmetrizable-$\Xset_1\times\Xset_2$.
Assume to the contrary that $\tLambda(P_{U,X_1,X_2})<\Lambda$. 
Based on Remark~\ref{rem:MLambdaJeq}, and our definition of the external variable $U$,
this means that there exist conditional state distributions $J_i(s|x_1,x_2)$, $i\in [1:n]$, which symmetrize-$\Xset_1\times\Xset_2$ the AVMAC, such that
\begin{align}
\tLambda(P_{U,X_1,X_2})=
\frac{1}{n}\sum_{i=1}^n \sum_{x_{1,i},x_{2,i},s_i} P_{X_{1,i},X_{2,i}}(x_{1,i},x_{2,i}) J_i(s_i|x_{1,i},x_{2,i})l(s_i) \leq \Lambda \,.
\label{eq:JnConvH10}
\end{align}
Now,  consider the following jamming strategy. First, the jammer selects from the codebooks a pair of codewords $(\tX_1^n,\tX_2^n)$ uniformly at random. Then, 
the jammer selects a sequence $\tS^n$ at random, according to the conditional distribution
\begin{align}
\cprob{\tS^n=s^n}{\tX_1=x_1^n,\tX_2=x_2^n}= J^n(s^n|x_1^n,x_2^n)\triangleq \prod_{i=1}^n J_i(s_i|x_{1,i},x_{2,i})\,.
\label{eq:HtSn}
\end{align} 
At last, if $l^n(\tS^n)\leq \Lambda$, the jammer chooses the state sequence to be $S^n=\tS^n$. Otherwise, the jammer chooses $S^n$ to be some sequence of zero cost. Such jamming strategy satisfies the state constraint $\Lambda$ with probability $1$.

To contradict our assumption that $\tLambda(P_{U,X_1,X_2})<\Lambda$, we first show that 
$\E l^n(\tS^n)=\tLambda(P_{U,X_1,X_2})$.
Observe that for every $(x_1^n,x_2^n)\in\Xset_1^n\times\Xset_2^n$, 
\begin{align}
\E\, \left( l^n(\tS^n) | \tX_1^n=x_1^n ,\, \tX_2^n=x_2^n \right)  =& 
\frac{1}{n} \sum_{i=1}^n \sum_{s\in\Sset} l (s) J_i(s|x_{1,i},x_{2,i}) \,.
\end{align}
Since $(\tX_1^n,\tX_2^n)$ are distributed as $(X_1^n,X_2^n)$, we obtain
\begin{align}
\E\, l^n(\tS^n)=&  \sum_{s\in\Sset} l (s) \cdot \frac{1}{n}  \sum_{i=1}^n \E J_i(s|X_{1,i},X_{2,i}) \nonumber\\
=& \sum_{u,x_1,x_2,s} P_U(u) P_{X_1,X_2|U}(x_1,x_2|u)J_u(s|x_1,x_2) l (s) \nonumber\\
=  & \tLambda(P_{X_1,X_2})<\Lambda \,.
\end{align}
Thus, by Chebyshev's inequality we have that for sufficiently large $n$, 
\begin{align}
\prob{ l^n(\tS^n)>\Lambda} \leq \delta_0 \,,
\end{align}
where $\delta_0>0$ is arbitrarily small.
Now, on the one hand, the probability of error is bounded by
\begin{align}
\err(q,\code_n)
\geq& \prob{g(Y^n)\neq (M_1,M_2), l^n(\tS^n)\leq \Lambda}
\nonumber\\
=&   \sum_{s^n \,:\; l^n(s^n)\leq\Lambda} \tq(s^n) \cerr(\code_n) \,,
\label{eq:MconvEb1}
\end{align}
where $\tq(s^n)$ is as defined in (\ref{eq:MconvFtq}).
On the other hand, the sequence $\tS^n$ can be thought of as the state sequence of an AVMAC without a state constraint, hence, by 
part 1 of Lemma~\ref{lemm:Gubner}, 
\begin{align}
\frac{1}{4}\leq &\err(\tq,\code_n)
\leq  \sum_{s^n \,:\; l^n(s^n)\leq\Lambda} \tq(s^n) \cerr(\code_n) +\prob{ l^n(\tS^n)> \Lambda}
\nonumber\\
\leq&  \sum_{s^n \,:\; l^n(s^n)\leq\Lambda} \tq(s^n) \cerr(\code_n)  +\delta_0 \,.
\label{eq:MconvEbf}
\end{align}
Thus, by (\ref{eq:MconvEb1})-(\ref{eq:MconvEbf}), the probability of error is bounded by 
$\err(q,\code_n)\geq \frac{1}{4}-\delta_0$. As this cannot be the case for a code with vanishing probability of error, we deduce that the assumption is false, \ie $\tLambda(P_{U,X_1,X_2})\geq \Lambda$.

It remains to show that $\tLambda_1(P_{U,X_1})\geq \Lambda$ and $\tLambda_2(P_{U,X_2})\geq \Lambda$.
Due to the symmetry, it suffices to show this for User 1.
We only need to consider an AVMAC which is symmetrizable-$\Xset_1|\Xset_2$, as otherwise, $\tLambda_1(P_{X_1})=+\infty$. Then, assume to the contrary that $\tLambda_1(P_{X_1})<\Lambda$,
 and let $J_{1,i}(s|x_1)$, $i\in [1:n]$, be the symmetrizing distributions that satisfy (\ref{eq:Msymmetrizable1}) and achieves the minimum in (\ref{eq:MtlambdaJ1}), \ie
\begin{align}
\tLambda_1(P_{U,X_1})=&
\sum_{ u,x_1,s} P_U(u) P_{X_1|U}(x_1|u) J_{1,u}(s|x_1) l(s) <\Lambda \,.
\label{eq:MtlambdaJConv1b}
\end{align}
%
Consider a jamming strategy, where the jammer first selects a codeword $\tX_1^n$ from the codebook of User 1, uniformly at random. Then, 
the jammer selects a sequence $\tS^n_1$ at random, according to the conditional distribution
\begin{align}
\cprob{\tS_1^n=s^n}{\tX_1=x_1^n}= J^n_1(s^n|x_1^n)\triangleq \prod_{i=1}^n J_{1,i}(s_i|x_{1,i})\,.
\end{align} 
At last, if $l^n(\tS^n_1)\leq \Lambda$, the jammer chooses the state sequence to be $S^n=\tS^n_1$. Otherwise, the jammer chooses $S^n$ to be some sequence of zero cost. 

To contradict our assumption that $\tLambda_1(P_{U,X_1})<\Lambda$, we first show that 
$\E l^n(\tS^n_1)=\tLambda_1(P_{U,X_1})$.
Observe that for every $x_1^n\in\Xset_1^n$, 
\begin{align}
\E\, \left( l^n(\tS^n_1) | \tX_1^n=x_1^n  \right)  =& 
\frac{1}{n} \sum_{i=1}^n \sum_{s\in\Sset} l (s) J_{1,i}(s|x_{1,i}) \,.
\end{align}
Since $\tX_1^n$ is distributed as $X_1^n$, we obtain
\begin{align}
\E\, l^n(\tS^n_1)=&  
\sum_{u,s\in\Sset} P_U(u) P_{X_1|U}(x_1|u) J_{1,u}(s|x_1) l (s) 
							 =   
 \tLambda_1(P_{U,X_1})<\Lambda \,.
\label{eq:MElSn1}
\end{align}
where the last equality is due to 
(\ref{eq:MtlambdaJConv1b}). 
Next, the probability of error is bounded by
\begin{align}
\err(q,\code_n)
\geq& \prob{g(Y^n)\neq (M_1,M_2), l^n(\tS^n_1)\leq \Lambda}
\nonumber\\
=&   \sum_{s^n \,:\; l^n(s^n)\leq\Lambda} \tq_1(s^n) \cerr(\code_n) \,,
\label{eq:MconvEb11}
\end{align}
where $\tq_1(s^n)$ is as defined in (\ref{eq:MconvFtq1}).
On the other hand, the sequence $\tS^n_1$ can be thought of as the state sequence of an AVMAC without a state constraint, hence, by 
part 2 of Lemma~\ref{lemm:Gubner}, 
\begin{align}
\frac{1}{4}\leq &\err(\tq_1,\code_n)
\leq  \sum_{s^n \,:\; l^n(s^n)\leq\Lambda} \tq_1(s^n) \cerr(\code_n) +\prob{ l^n(\tS^n_1)> \Lambda}
\nonumber\\
\leq&  \sum_{s^n \,:\; l^n(s^n)\leq\Lambda} \tq_1(s^n) \cerr(\code_n)  +\delta_1 \,,
\label{eq:MconvEbf1}
\end{align}
where the last line is due to (\ref{eq:MElSn1}) and Chebyshev's inequality, with arbitrarily small $\delta_1>0$.
Thus,  (\ref{eq:MconvEb11})-(\ref{eq:MconvEbf1}) imply that 
$\err(q,\code_n)\geq \frac{1}{4}-\delta_1$,  which cannot hold for a
 code with vanishing probability of error. We deduce that the assumption is false, \ie 
$\tLambda_1(P_{U,X_1})\geq \Lambda$.
This completes the converse proof.
\qed

\subsection*{Case B and Case C}
Before we begin, we note that Cases B-D can also be proved by directly adjusting the techniques of Csisz\'{a}r and Narayan for the single user AVC \cite{CsiszarNarayan:88p}.
Although, as explained in Remark~\ref{remark:WieseBoche}, it is not an immediate consequence.
Thereby, it is easier for us to use our previous derivations instead.


The proof follows similar arguments as in Case A, and thus we only give the outline.
Since Case B and Case C in Definition~\ref{def:MICavc} are symmetric, we only treat the former.
Suppose that $L^*>\Lambda$ and $L_2^*>\Lambda$, but $L_1^*<\Lambda$.
For the direct part, we can use the same coding scheme as in Case A with the following changes.
First, coded time sharing is no longer necessary, hence we take $U=\emptyset$.
Then, let $ P_{X_1}$ and $P_{X_2}$ be types, such that $\E \cost_k(X_k)\leq\plimit_k$, for $k=1,2$, 
$\tLambda_2(P_{X_2})>\Lambda$ and $\tLambda(P_{X_1,X_2})>\Lambda$. 
As User 1 transmits at zero rate, we can discard of Condition 2b) of the decoding rule (see Definition~\ref{def:MLdecoder}).
Nevertheless, given our assumption in Remark~\ref{remark:stochE}, Encoder 1 may use ``local randomness" and send a sequence $x_1^n=f_1(\sigma)$, where $\sigma\in[1:2^{nR_1}]$ is drawn 
 uniformly at random, with $R_1=\eps$. 
Upon receiving $y^n\in\Yset^n$, the decoder declares its estimation $g(y^n)=m_2$ iff there exists $\sigma$ 
such that $y^n\in\Dset(\sigma,m_2)$, where the decoding sets $\Dset(\sigma,m_2)\subseteq \Yset^n$ are as in Definition~\ref{def:MLdecoder}.
The message of User 2 is still  decoded uniquely,
since the only part of Lemma~\ref{lemm:MdisDec} that depends on $\tLambda_1(P_{X_1})$ is part 2, which is no longer necessary. The analysis of the probability of error remains exactly the same, except that the error event $\Eset_{2\text{b}}$ can be ignored. It follows that the probability of error tends to zero, provided that 
\begin{align}
R_1+R_2 <& \min_{q(s) \,:\; \E_q l(S)\leq\Lambda} I_q(X_1,X_2;Y)-\delta-5\eps \,,
\nonumber\\ 
R_2 <& \min_{q(s) \,:\; \E_q l(S)\leq\Lambda} I_q(X_2;Y|X_1)-\delta_2-5\eps \,.
\end{align}
Since $R_1=\eps$, and $I_q(X_1,X_2;Y)\geq I_q(X_2;Y|X_1)$, the first inequality is inactive, and the direct part follows.

The converse part also follows from the converse proof for Case A. It was shown that if  the jammer selects the state sequence to be
\begin{align}
S^n=\begin{cases}
\tS^n_1 &\text{if $l^n(\tS_1^n)\leq \Lambda$},\\
(s_0,\ldots,s_0) &\text{otherwise}
\end{cases} \,,
\end{align} 
for $\tS^n_1\sim\tq_1(s^n)$ as in (\ref{eq:MconvFtq1}), and $s_0\in\Sset$ with $l(s_0)=0$,
then the probability of error is lower bounded by $\err(q,\code_n)\geq \frac{1}{4}-\delta_1$ for $R_1>0$, hence User 1 cannot achieve positive rates. As for User 2, we have by (\ref{eq:MavcConvR2u}) that
\begin{align}
R_2\leq& \min_{q(s|u) \,:\; E_q l(S)\leq \Lambda-\delta} I_q(X_2;Y|X_1,U)+\eps_n \,. \label{eq:MavcConvR2uB}
\end{align}
Then, observe that $I_q(X_2;Y|X_1,U)\leq I_q(X_2;Y|X_1)$, since $U\Cbar (X_1,X_2)\Cbar Y$ form a Markov chain, and conditioning reduces entropy (see \eg \cite[Theorem 2.6.5]{CoverThomas:06b}). By the same considerations as in Case A, $\tLambda_2(P_{X_2})\geq \Lambda$, and the converse part follows.
\qed

\subsection*{Case D}
Suppose that $L^*<\Lambda$.
 It was shown in the converse proof for Case A, that if  the jammer selects the state sequence 
\begin{align}
S^n=\begin{cases}
\tS^n &\text{if $l^n(\tS^n)\leq \Lambda$},\\
(s_0,\ldots,s_0) &\text{otherwise}
\end{cases} \,,
\end{align} 
for $\tS^n\sim\tq(s^n)$ as in (\ref{eq:MconvFtq}), and $s_0\in\Sset$ with $l(s_0)=0$,
then the probability of error is lower bounded by $\err(q,\code_n)\geq \frac{1}{4}-\delta_0$ for 
$R_1+R_2>0$. Thus, positive rates cannot be achieved. 

Now, suppose that both $L^*_1<\Lambda$ and $L^*_2<\Lambda$. We have already seen in the proof of Case B, that $L^*_1<\Lambda$ implies that User 1 cannot achieve $R_1>0$, and by symmetry,
$L^*_2<\Lambda$ implies that User 2 cannot achieve $R_2>0$.
Therefore, if $L^*<\Lambda$, or both $L^*_1<\Lambda$ and $L^*_2<\Lambda$, then the deterministic code capacity region is 
$\{(0,0)\}$, as we were set to prove.
This concludes the proof of Theorem~\ref{theo:MCavc}.
\qed

\section{Proof of Corollary~\ref{coro:MCavc01}}
\label{app:MCavc01}
Assume that the AVMAC $\avmac$ satisfies the conditions of Corollary~\ref{coro:MCavc01}. 
%
Looking into  the converse proof of Theorem~\ref{theo:MCavc} in Appendix~\ref{app:MCavc}
above, the following addition suffices.
We show that for every code $\code_n$ as in Appendix~\ref{app:MCavc},
$\tLambda(P_{U,X_1,X_2})=\Lambda$ implies that $R_1+R_2=0$, 
$\tLambda_1(P_{U,X_1})=\Lambda$ implies that $R_1=0$, and $\tLambda_2(P_{U,X_2})=\Lambda$ implies that $R_2=0$.
Since there is only a polynomial number of types, we may consider $P_{U,X_1,X_2}$ to be the joint type of $(u^n,f_1(m_1),f_2(m_2))$, for all $m_1$ and $m_2$ (see \cite[Problem 6.19]{CsiszarKorner:82b}).

Suppose that  $\tLambda(P_{U,X_1,X_2})=\Lambda$,  assume to the contrary that $R_1+R_2>0$,
 and let $J_u(s|x_1,x_2)$ be distributions that achieve the minimum in (\ref{eq:MtlambdaJ}), \ie
\begin{align}
\tLambda(P_{U,X_1,X_2})=&
 \sum_{ u,x_1,x_2,s} P_U(u) P_{X_{1},X_{2}|U}(x_1,x_2|u) J_u(s|x_1,x_2) l(s) =\Lambda \,.
\label{eq:MtlambdaJConvc}
\end{align}
Based on the condition of the corollary, we may assume that  $J_u(s|x_1,x_2)$  is a $0$-$1$ law, \ie 
\begin{align}
J_u(s|x_1,x_2)=\begin{cases}
1 &\text{if $s=G_u(x_1,x_2)$},\\
0 &\text{otherwise}
\end{cases} \,,
\end{align}
for some deterministic function $G_u:\Xset_1\times\Xset_2\rightarrow\Sset$.  
Thus, by (\ref{eq:MtlambdaJConvc}),
\begin{align}
 \E l(G_U(X_1,X_2))= \sum_{ u, x_1,x_2,s} P_U(u) P_{X_1,X_2|U}(x_1,x_2|u) J_u(s|x_1,x_2) l(s) =\Lambda \,.
\label{eq:MtlambdaJConvEG}
\end{align}
Recall that we have defined $U$ in the converse proof as a uniformly distributed variable over $\Uset=[1:n]$.
Now,  consider the following jamming strategy. First, the jammer selects from the codebooks a pair of codewords $(\tX_1^n,\tX_2^n)$ uniformly at random. Then, given $\tX_1^n=x_1^n$ and $\tX_2^n=x_2^n$, 
the jammer chooses the state sequence $S^n=\left( G_i(x_{1,i},x_{2,i}) \right)_{i=1}^n$. 
Observe that given pair of codewords, $\tX_1^n=x_1^n$ and $\tX_2^n=x_2^n$, 
\begin{align}
 l^n(S^n)   =&  \frac{1}{n} \sum_{i=1}^n l(G_i(x_{1,i},x_{2,i})) 
=    \E l(G_U(X_{1},X_{2})) 
	=\Lambda \,,
\end{align}
where the last equality is 
due to 
(\ref{eq:MtlambdaJConvEG}).
Thus, the state sequence satisfies the state constraint.
Now, observe that the jamming strategy $S^n=\left( G_i(\tX_{1,i},\tX_{2,i}) \right)_{i=1}^n$ is equivalent to 
$S^n\sim\tq(s^n)$ as in (\ref{eq:MconvFtq}).
Thus, by part 1 of Lemma~\ref{lemm:Gubner}, 
we have that 
$\err(\tq,\code_n)\geq \frac{1}{4}$, hence both users cannot achieve a positive rate.

Next, consider the case where  $\tLambda_1(P_{U,X_1})=\Lambda$.  Assume to the contrary that $R_1>0$,
 and let $J_{1,u}(s|x_1)$ be distributions that achieves the minimum in (\ref{eq:MtlambdaJ1}), \ie
\begin{align}
\tLambda_1(P_{U,X_1})=&
\sum_{ u,x_1,s} P_U(u) P_{X_1|U}(x_1|u) J_u(s|x_1) l(s) =\Lambda \,.
\label{eq:MtlambdaJConv1d}
\end{align}
By assumption, every $J_{1,u}(s|x_1)$ 
has a $0$-$1$ law, 
\begin{align}
J_{1,u}(s|x_1)=\begin{cases}
1 &\text{if $s=G_{1,u}(x_1)$},\\
0 &\text{otherwise}
\end{cases} \,,
\end{align}
for some deterministic function $G_{1,u}:\Xset_1\rightarrow\Sset$.  
Thus, by (\ref{eq:MtlambdaJConv1d}),
\begin{align}
\E l(G_{1,U}(X_1))= \sum_{ u,x_1} P_U(u) P_{X_1|U}(x_1|u) J_{1,u}(s|x_1) l(s) =\Lambda \,.
\label{eq:MtlambdaJConvEG1e}
\end{align}
Now,  suppose the jammer selects from the codebook of User 1, a codeword $\tX_1^n$ uniformly at random. Then, given $\tX_1^n=x_1^n$, 
the jammer chooses the state sequence $S^n=\left( G_{1,i}(x_{1,i}) \right)_{i=1}^n$. Hence, 
For every given codeword $\tX_1^n=x_1^n$, 
\begin{align}
 l^n(S^n)   =&  \frac{1}{n} \sum_{i=1}^n l(G_{1,i}(x_{1,i}))=\E l(G_{1,U}(X_{1})) \,.
\end{align}
Thus, by (\ref{eq:MtlambdaJConvEG1e}), we have that
$l^n(S^n)=\Lambda$ with probability $1$. This means that the state sequence satisfies the state constraint.
Now, observe that the jamming strategy $S^n=\left( G_{1,i}(\tX_{1,i}) \right)_{i=1}^n$ is equivalent to 
$S^n\sim\tq_1(s^n)$ as in (\ref{eq:MconvFtq1}).
Thus, by part 2 of Lemma~\ref{lemm:Gubner}, 
we have that 
$\err(\tq_1,\code_n)\geq \frac{1}{4}$, hence $R_1=0$. 
By symmetry, we have that $\tLambda_2(P_{X_2})=\Lambda$ implies that $R_2=0$.
\qed

\section{Analysis of Example~\ref{example:BSMAC}}
\label{app:BSMAC}
Let $\avmac$ be the arbitrarily varying binary symmetric MAC, with two independent binary symmetric channels, as in Example~\ref{example:BSMAC}.

We begin with the random code capacity region.
To show achievability, set 
$U=\emptyset$, $X_1\sim\text{Bernoulli}(\omega_1)$, $X_2\sim\text{Bernoulli}(\omega_2)$, and observe that
\begin{align}
&I_q(X_1;Y_1,Y_2|X_2)\geq I_q(X_1;Y_1|X_2)  \stackrel{(a)}{=} I_q(X_1;Y_1) \,,\nonumber\\
&I_q(X_2;Y_1,Y_2|X_1)\geq I_q(X_2;Y_2|X_1)  \stackrel{(b)}{=} I_q(X_2;Y_2) \,,
\end{align} 
and
\begin{align}
 I_q(X_1,X_2;Y_1,Y_2)=& H(X_1)+H(X_2)-H_q(X_1,X_2|Y_1,Y_2) \nonumber\\
\stackrel{(c)}{\geq} & H(X_1)+H(X_2)-H_q(X_1|Y_1,Y_2)-H_q(X_2|Y_1,Y_2) \nonumber\\
\stackrel{(d)}{\geq} & I_q(X_1;Y_1)+I_q(X_2;Y_2) \,,
\end{align}
where $(a)$ holds since $X_2$ is independent of $(X_1,S_1,Y_1)$;
$(b)$ holds since $X_1$ is independent of $(X_2,S_2,Y_2)$;
$(c)$ is due to the independence bound on entropy \cite[Theorem 2.6.6.]{CoverThomas:06b};
and $(d)$ holds since conditioning reduces entropy \cite[Theorem 2.6.5.]{CoverThomas:06b}.
Therefore, based on Theorem~\ref{theo:MrCav}, $(R_1,R_2)$ is achievable for
\begin{align}
R_1\leq \min_{q(s_1):\E S_1\leq\Lambda} I_q(X_1;Y_1)
=\min_{0\leq q_1\leq \Lambda} \left[  h(\omega_1*q_1)-h(q_1) \right] \,,\\
R_2\leq \min_{q(s_2):\E S_2\leq\Lambda} I_q(X_2;Y_2) =\min_{0\leq q_1\leq \Lambda} \left[  h(\omega_2*q_2)-h(q_2) \right] \,.
\end{align}
Since $h(\omega*t)-h(t)$ is a convex-$\cup$ function over $0\leq t\leq 1$ with minimum at $t=\frac{1}{2}$,
we have that
\begin{align}
\MrCav\supseteq
\left\{
\begin{array}{lrl}
(R_1,R_2) \,:\; & R_1 		\leq&   h(\omega_1*\lambda)-h(\lambda)  \,, \\
								& R_2 		\leq&   h(\omega_2*\lambda)-h(\lambda)  
\end{array}
\right\} \,,
\end{align}
which completes the achievability proof.
To prove the converse part, we observe that the rate of User 1 is bounded by
\begin{align}
&\min_{q(s_1,s_2):\E S_1+\E S_2\leq\Lambda} I_q(X_1;Y_1,Y_2|X_2,U) \leq 
I_q(X_1;Y_1,Y_2|X_2,U) \Big|_{ \substack{S_1\sim\text{Bernoulli}(\lambda) \\ S_2=0  } }
\nonumber\\
=& I_q(X_1;Y_1|X_2,U) \Big|_{ \substack{S_1\sim\text{Bernoulli}(\lambda) \\ S_2=0  } }
=h(p_1*\lambda)-h(\lambda) \leq h(\omega_1*\lambda)-h(\lambda) \,,
\end{align}
with $X_1\sim\text{Bernoulli}(p_1)$, for $0\leq p_1\leq \plimit_1$, where the first equality holds since $S_2=0$ implies that $Y_2=X_2$, and the last inequality holds since $h(\alpha*t)$ is a concave-$\cap$ function over $0\leq t\leq 1$ with maximum at $t=\frac{1}{2}$. 
Similarly, the rate of User 2 is bounded by
\begin{align}
&\min_{q(s_1,s_2):\E S_1+\E S_2\leq\Lambda} I_q(X_2;Y_1,Y_2|X_1,U) \leq 
h(\omega_2*\lambda)-h(\lambda) \,.
\end{align} 
This proves that the random code capacity region of the AVMAC in in Example~\ref{example:BSMAC}  is given by 
(\ref{eq:MrCavBSMAC}).

Moving to the deterministic code capacity region, we first compute $L^*$, $L_1^*$ and $L_2^*$.
For every $P_{X_1,X_2}$,
\begin{align}
&\Psi(P_{X_1,X_2})=\min_{0\leq\alpha_1,\alpha_2\leq 1} (\alpha_1*p_1+\alpha_2*p_2) 
=\min(p_1,1-p_1)+\min(p_2,1-p_2) \,,  \\
&\Psi_1(p_1)=\min_{0\leq\alpha_1\leq 1} \alpha_1*p_1=\min(p_1,1-p_1) \,, \\
&\Psi_2(p_2)=\min_{0\leq\alpha_2\leq 1} \alpha_2*p_2=\min(p_2,1-p_2) \,,
\end{align}
where we have used the notation  $p_1=P_{X_1}(1)=1-P_{X_1}(0)$ and $p_2=P_{X_2}(1)=1-P_{X_2}(0)$.
Therefore,
\begin{align}
& L^*=\max_{ \substack{ 0\leq p_1\leq \plimit_1 \,,\\ 0\leq p_2\leq \plimit_2 }} \Psi(P_{X_1,X_2})
=\omega_1+\omega_2 \,,
\label{eq:BSMACls} \\
& L_1^*=\max_{0\leq p_1\leq \plimit_1}\Psi_1(p_1) =\omega_1   \,,
\label{eq:BSMACls1}\\
& L_2^*=\max_{0\leq p_2\leq \plimit_2}\Psi_2(p_2) =\omega_2   \,.
\label{eq:BSMACls2}
\end{align}
Observe that based on the result on the random code capacity region, we have that for  $\Lambda\geq \frac{1}{2}$, or equivalently, $\lambda=\frac{1}{2}$, the capacity region is given by $\MCavc=\MrCav=\{(0,0)\}$, in agreement with (\ref{eq:MCavcBSMACa})-(\ref{eq:MCavcBSMACd}).
Henceforth, assume that $\Lambda<\frac{1}{2}$.

The first case to consider is $\plimit_1>\Lambda$ and $\plimit_2>\Lambda$. 
The converse part is immediate, since the deterministic code capacity region is always bounded by the random code capacity region. As for the direct part, we are going to show that under the assumption that $\Lambda< \frac{1}{2}$, we have that $L^*$, $L_1^*$ and $L_2^*$ are greater than $\Lambda$. Indeed, if $\plimit_k\geq \frac{1}{2}$, then $\omega_k=\frac{1}{2}>\Lambda$, for $k=1,2$. Otherwise, if $\plimit_k< \frac{1}{2}$, then $\omega_k=\plimit_k>\Lambda$, for $k=1,2$.
Therefore, in both cases, we have by (\ref{eq:BSMACls})-(\ref{eq:BSMACls2}) that 
$L^*>\Lambda$, $L_1^*>\Lambda$, and $L_2^*>\Lambda$, which corresponds to Case a) in Definition~\ref{def:MICavc}. It is further inferred that taking $p_1=\omega_1$ and $p_2=\omega_2$, we have that 
$\tLambda(p_1,p_2)=\omega_1+\omega_2>\Lambda$ and
$\tLambda_k(p_k)=\omega_k>\Lambda$ for $k=1,2$. It follows that this inputs distribution is legitimate, in the sense that it belongs to the optimization set $\overline{\pSpace}_{\plimit_1,\plimit_2,\Lambda}(\Uset\times\Xset_1\times\Xset_2)$ (see definition in (\ref{eq:MpI})).
 This completes the achievability proof for the first case, because we have already seen in the achievability proof of the random code capacity region above, that the inputs distribution $p_1=\omega_1$ and $p_2=\omega_2$ achieves the region in the RHS of (\ref{eq:MCavcBSMACa}), with $U=\emptyset$. 

Next, we consider the second case, $\plimit_1\leq\Lambda$ and $\plimit_2>\Lambda$. 
 Assuming that $\Lambda< \frac{1}{2}$, we have that $\plimit_1<\frac{1}{2}$, hence 
$L_1^*=\omega_1=\plimit_1\leq \Lambda$. As for $L_2^*$, if $\plimit_2\geq \frac{1}{2}$,
then $\omega_2=\frac{1}{2}>\Lambda$, and if $\plimit_2< \frac{1}{2}$, then
$\omega_2=\plimit_2>\Lambda$. It follows that $L^*=\omega_1+\omega_2>\Lambda$ and
$L_2^*=\omega_2>\Lambda$, but $L_1\leq\Lambda$,  which corresponds to Case b) in Definition~\ref{def:MICavc}.
Furthermore, the input distributions $p_1=\omega_1$ and $p_2=\omega_2$ belong to the maximization set in 
(\ref{eq:Mcasea}), as $\tLambda(p_1 p_2)=\omega_1+\omega_2>\Lambda$ and $\tLambda_2(p_2)=\omega_2>\Lambda$.
Thus, User 2 can achieve rates below
\begin{align}
&\min_{q(s_1,s_2) : \E S_1+\E S_2\leq \Lambda} I_q(X_2;Y_1,Y_2|X_1) \geq
\min_{q(s_1,s_2) : \E S_1+\E S_2\leq \Lambda} I_q(X_2;Y_2|X_1)\nonumber\\ =&
\min_{0\leq q_2 \leq \Lambda} \left[ h(\omega_2*q_2)-h(q_2) \right]=
h(\omega_2*\lambda)-h(\lambda) \,.
\end{align}
It is also the highest rate achievable for User 2, since
\begin{align}
R_2\leq& \min_{q(s_1,s_2) : \E S_1+\E S_2\leq \Lambda} I_q(X_2;Y_1,Y_2|X_1)
\leq I_q(X_2;Y_1,Y_2|X_1) \Big|_{ \substack{ S_1=0 \\ S_2\sim\text{Bernoulli}(\lambda)  } }
\nonumber\\
=& I_q(X_2;Y_2) \Big|_{ \substack{ S_1=0 \\ S_2\sim\text{Bernoulli}(\lambda)  } }
\leq \max_{0\leq p_2\leq \plimit_2} h(p_2*\lambda)-h(\lambda) = h(\omega_2*\lambda)-h(\lambda) \,,
\end{align}
following the same considerations as in the derivation of the random code capacity region above. 
The third case, $\plimit_1>\Lambda$ and $\plimit_2\leq\Lambda$, follows by symmetry.

In the fourth case, $\plimit_1<\Lambda$ and $\plimit_2<\Lambda$, we have that
\begin{align}
&L^*_1=\omega_1\leq\plimit_1<\Lambda \,,\nonumber\\
&L^*_2=\omega_2\leq\plimit_2<\Lambda \,,
\end{align}
as in Case d) in Definition~\ref{def:MICavc}. Thus, the capacity region is 
$\MCavc=\{(0,0)\}$.  
\qed

\section{Analysis of Example~\ref{example:GaussMAC}}
\label{app:GaussMAC}
Deriving the random code capacity region is straightforward. To show achievability, we set $U=\emptyset$, $X_1\sim\mathcal{N}(0,\plimit_1)$, and 
$X_2\sim\mathcal{N}(0,\plimit_2)$. Then, since Gaussian noise is known to be the worst additive noise under variance constraint \cite[Lemma II.2]{DiggaviCover:01p}, we have that
\begin{align}
&\min\limits_{F_S \,:\; \E S^2\leq\Lambda} I(X_1;Y|X_2)=\frac{1}{2}\log\left(1+\frac{\plimit_1}{\Lambda+\sigma^2}\right) \,,\\
&\min\limits_{F_S \,:\; \E S^2\leq\Lambda} I(X_2;Y|X_1)=\frac{1}{2}\log\left(1+\frac{\plimit_2}{\Lambda+\sigma^2}\right) \,,\\
&\min\limits_{F_S \,:\; \E S^2\leq\Lambda} I(X_1,X_2;Y)=\frac{1}{2}\log\left(1+\frac{\plimit_1+\plimit_2}{\Lambda+\sigma^2}\right) \,.
\end{align}
 This proves achievability. 
To prove the converse part,  observe that the rate of User 1 is bounded by 
\begin{align}
&\min_{F_S:\E S^2\leq\Lambda} I(X_1;Y|X_2,U) \leq  
I(X_1;Y|X_2,U) \Big|_{ S\sim \mathcal{N}(0,\Lambda)  } \leq \frac{1}{2}\log\left( 1+\frac{\plimit_1}{\Lambda+\sigma^2 } \right) \,,
\end{align}
where the last inequality follows as in the converse proof of the classical Gaussian MAC $Y=X_1+X_2+\tZ$, with
$\tZ\sim\mathcal{N}(0,\Lambda +\sigma^2)$  \cite{Wyner:74p}.
The bounds on the rate of User 2 and on the sum rate are proved in the same manner. We have thus determined the random code capacity region.

We move to the deterministic code capacity region, as in
Theorem~\ref{theo:MCavc}. 
 First, we calculate the thresholds $L^*$, $L_1^*$, and $L^*_2$. 
Based on Definition~\ref{def:Msymmetrizable}, the Gaussian AVMAC is symmetrized-$\Xset_1\times\Xset_2$ by 
a conditional pdf
$\varphi(s|x_1,x_2)$ if 
\begin{align}
\label{eq:GPsymmetrizableMj}
\int_{-\infty}^\infty  \varphi(s|\tx_1,\tx_2) f_{Z}(y-x_1-x_2-s) \,ds=
\int_{-\infty}^\infty  \varphi(s|x_1,x_2) f_{Z}(y-\tx_1-\tx_2-s) \,ds
 \,,\; 
\forall\, x_1,x_2,\tx_1,\tx_2,y\in\mathbb{R} \,,
\end{align}
where $f_{Z}(z)=\frac{1}{\sqrt{2\pi\sigma^2}} e^{-z^2/2\sigma^2}$. 
%
In particular, observe that (\ref{eq:GPsymmetrizableMj}) holds for $\varphi(s|x_1,x_2)=\delta(s-x_1-x_2)$, where $\delta(\cdot)$ is the Dirac delta function. In other words, the channel is symmetrized by a distribution $\varphi(s|x_1,x_2)$ which gives probability $1$ to 
$S=x_1+x_2$.
The minimal state cost $\tLambda(F_{X_1}F_{X_2})$ for the jammer to symmetrize-$\Xset_1\times\Xset_2$, for the input distribution 
$f_{X_1}(x_1)f_{X_2}(x_2)$, is the contiuous version of (\ref{eq:MtlambdaJ}), 
\begin{align}
\tLambda(F_{X_1}F_{X_2})=\min\, \int_{-\infty}^\infty  \int_{-\infty}^\infty  \int_{-\infty}^\infty
  f_{X_1}(x_1)f_{X_2}(x_2)\varphi(s|x_1,x_2)s^2 \,ds \,dx_1 \,dx_2
 \,,
\label{eq:GPlambdaTMj}
\end{align}
where the minimization is over all conditional pdfs $\varphi(s|x_1,x_2)$ that symmetrize-$\Xset_1\times\Xset_2$ the channel, that is, satisfy (\ref{eq:GPsymmetrizableMj}). Similar expressions can be written for the individual state costs $\tLambda_1(F_{X_1})$ and $\tLambda_2(F_{X_2})$ as the continuous versions of (\ref{eq:MtlambdaJ1}) and (\ref{eq:MtlambdaJ2}), respectively.
The following lemma states that the minimal state cost for joint symmetrizability is the total input power, and the minimal state cost for individual symmetrizability is the input power of the corresponding transmitter.
\begin{lemma}
\label{lemm:GPscostPM}
For zero mean random variables $X_1$ and $X_2$,
\begin{align}
\Psi(F_{X_1} F_{X_2})=& \E X_1^2+ \E X_2^2  \,,\\
\Psi_1(F_{X_1})=& \E X_1^2 \,,\\
\Psi_2(F_{X_2})=& \E X_2^2 \,.
\end{align}
\end{lemma}
\begin{proof}[Proof of Lemma~\ref{lemm:GPscostPM}]
First, we evaluate $\Psi(F_{X_1} F_{X_2})$.
Observe that by (\ref{eq:GPlambdaTMj}), the Gaussian AVMAC is symmetrized-$\Xset_1\times\Xset_2$ by a conditional pdf 
$\varphi_{x_1,x_2}(s)=\varphi(s|x_1,x_2)$ if
\begin{align}
\label{eq:GPsymmetrizableEqMj}
\int_{-\infty}^\infty  \varphi_{0,0}(s)f_{Z}(y-x_1-x_2-s) \,ds=
\int_{-\infty}^\infty  \varphi_{x_1,x_2}(s)f_{Z}(y-s)\,ds
 \,, 
\end{align}
for all $x_1,x_2,y\in\mathbb{R}$.
By substituting $z=y-x_1-x_2-s$ in the LHS, and $\bar{z}=y-s$ in the RHS, this is equivalent to 
\begin{align}
\label{eq:GPsymmetrizableEq2}
\int_{-\infty}^\infty  \varphi_{0,0}(y-x_1-x_2-z)f_{Z}(z)\, dz=
\int_{-\infty}^\infty  \varphi_{x_1,x_2}(y-\bar{z})f_{Z}(\bar{z})\, d\bar{z} \,.
\end{align}
For every $x_1,x_2\in\mathbb{R}$, define the random variable $\oS(x_1,x_2)\sim\varphi_{x_1,x_2}$.
We note that the RHS is the convolution of the pdfs of the random variables $Z$ and 
$\oS(x_1,x_2) $, while the LHS is the convolution of the pdfs of the random variables $Z$ and 
$\oS(0,0)+x_1+x_2 $. This is not surprising since the channel output $Y$ is  a sum of independent random variables, and thus the pdf of $Y$ is a convolution of pdfs.
It follows that $\varphi_{0,0}(y-x_1-x_2)=\varphi_{x_1,x_2}(y)$, and
 by plugging $s$ instead of $y$, we have that $\varphi_{x_1,x_2}$ symmetrizes-$\Xset_1\times\Xset_2$ the Gaussian AVMAC if and only if 
\begin{align}
\varphi_{x_1,x_2}(s)=\varphi_{0,0} (s-x_1-x_2) \,.
\label{eq:GPvarphiSymmEqM}
\end{align} 
Then, the corresponding state cost satisfies
\begin{align}
&\int_{-\infty}^\infty\int_{-\infty}^\infty\int_{-\infty}^\infty f_{X_1}(x_1)f_{X_2}(x_2) \varphi_{x_1,x_2}(s) s^2 
\, dx_1 \,dx_2 \, ds
\nonumber\\
=& \int_{-\infty}^\infty\int_{-\infty}^\infty\int_{-\infty}^\infty f_{X_1}(x_1)f_{X_2}(x_2) \varphi_{0,0} (s-x_1-x_2) s^2 
\, dx_1 \,dx_2 \, ds
\nonumber\\
=& \int_{-\infty}^\infty\int_{-\infty}^\infty\int_{-\infty}^\infty f_{X_1}(x_1)f_{X_2}(x_2) \varphi_{0,0} (a) (a+x_1+x_2)^2 \, da
\, dx_1 \,dx_2
\nonumber\\
=& \int_{-\infty}^\infty
\left[ \int_{-\infty}^\infty  \int_{-\infty}^\infty (x_1+x_2+a)^2
f_{X_1}(x_1)f_{X_2}(x_2) \, dx_1 \, dx_2
\right]
 \varphi_{0,0} (a)  \, da 
\label{eq:GPscostEq1Mj}
\end{align}
where the second equality follows by the integral substitution of $a=s-x_1-x_2$.
Observe that the bracketed integral can be expressed as 
$
\E[(X_1+X_2+a)^2]=\E X_1^2+\E X_2^2+a^2 
$. 
Thus, 
\begin{align}
 \int_{-\infty}^\infty  \int_{-\infty}^\infty (x_1+x_2+a)^2
f_{X_1}(x_1)f_{X_2}(x_2) \, dx_1 \, dx_2
\geq& \E X_1^2+ \E X_2^2 \,.
\label{eq:GPCscostTrMj}
\end{align}
The last inequality holds for any $\varphi_{x_1,x_2}$ which symmetrizes-$\Xset_1\times\Xset_2$ the channel. 
Now, observe that (\ref{eq:GPvarphiSymmEqM}) holds
 for $\hat{\varphi}_{x_1,x_2}(s)=\delta(s-x_1-x_2)$, where $\delta(\cdot)$ is the Dirac delta function, 
hence $\hat{\varphi}_{x_1,x_2}$ symmetrizes-$\Xset_1\times\Xset_2$ the channel.
In addition, since $\hat{\varphi}_{0,0}$ gives probability $1$ to $S=0$, we have that 
 (\ref{eq:GPCscostTrMj}) holds with equality for $\hat{\varphi}_{x_1,x_2}$,
and thus, $\Psi( F_{X_1}F_{X_2} )=\E X_1^2+\E X_2^2$. 

Next, consider $\Psi_1(F_{X_1})$. The Gaussian AVMAC is symmetrized-$\Xset_1|\Xset_2$ by a conditional pdf 
$\varphi'_{x_1}(s)=\varphi_1(s|x_1)$ if
 \begin{align}
\label{eq:GPsymmetrizableEqM1}
\int_{-\infty}^\infty  \varphi'_{0}(s)f_{Z}(y-x_1-x_2-s) \,ds=
\int_{-\infty}^\infty  \varphi'_{x_1}(s)f_{Z}(y-x_2-s)\,ds
 \,, 
\end{align}
for all $x_1,x_2,y\in\mathbb{R}$. 
By substituting $y'=y-x_2$, $z=y'-x_1-s$ in the LHS, and $\bar{z}=y'-s$ in the RHS, this is equivalent to 
\begin{align}
\label{eq:GPsymmetrizableEq2M1}
\int_{-\infty}^\infty  \varphi'_{0}(y'-x_1-z)f_{Z}(z)\, dz=
\int_{-\infty}^\infty  \varphi'_{x_1}(y'-\bar{z})f_{Z}(\bar{z})\, d\bar{z} \,.
\end{align}
As earlier, it follows that 
 $\varphi'_{x_1}$ symmetrizes-$\Xset_1|\Xset_2$ 
if and only if 
$
\varphi'_{x_1}(s)=\varphi'_{0} (s-x_1) 
$. 
By similar derivation as in (\ref{eq:GPscostEq1Mj}), the corresponding state cost satisfies
\begin{align}
\int_{-\infty}^\infty\int_{-\infty}^\infty f_{X_1}(x_1) \varphi'_{x_1}(s) s^2 
\, dx_1  \, ds
=& \E X_1^2+
 \int_{-\infty}^\infty
a'^2
 \varphi'_{0} (a')  \, da' \geq \E X_1^2 \,, 
\label{eq:GPscostEq1M1}
\end{align}
with equality for $\varphi'_{x_1}(s)=\delta(s-x_1)$. Hence, $\tLambda_1( F_{X_1})=\E X_1^2$, and by symmetry, $\Psi_2( F_{X_2})=
\E X_2^2$. This completes the proof of Lemma~\ref{lemm:GPscostPM}.
\end{proof}
Going forward with the derivation of the deterministic code capacity region,
 we have by Lemma~\ref{lemm:GPscostPM} that the thresholds
defined in (\ref{eq:Lstar})-(\ref{eq:Lstar2}) are given by
\begin{align}
L^*=&\max_{F_{X_1}F_{X_2}: \E X_1^2\leq \plimit_1 , \E X_2^2\leq\plimit_2} \tLambda(F_{U,X_1,X_2}) =\plimit_1+\plimit_2 \,, \\
L^*_1=&\max_{F_{X_1}\E X_1^2\leq \plimit_1 } \tLambda_1(F_{U,X_1}) =\plimit_1 \,,\\ 
L^*_2=&\max_{F_{X_2}\E X_2^2\leq \plimit_2 } \tLambda_2(F_{U,X_2}) =\plimit_2 \,.
\end{align}
We can now complete the derivation by applying Theorem~\ref{theo:MCavc} to cases where $\plimit_k\neq\Lambda$ for $k=1,2$.

If $\plimit_1>\Lambda$ and $\plimit_2>\Lambda$, then $L^*>\Lambda$, $L^*_1>\Lambda$, and $L^*_2>\Lambda$, which corresponds to Case a) in 
Definition~\ref{def:MICavc}.
We have seen that the random code capacity region is achieved with the input distribution specified by  $U=\emptyset$, $X_1\sim\mathcal{N}(0,\plimit_1)$, and $X_2\sim\mathcal{N}(0,\plimit_2)$, which is in the set $\overline{\pSpace}_{\plimit_1,\plimit_2,\Lambda}(\Uset\times\Xset_1\times\Xset_2)$ (see (\ref{eq:MpI})). It follows that the capacity region is the same as the random code capacity region,
\ie $\MCavc=\MrCav$, as in (\ref{eq:MCavcGaussa}).

%

For the Gaussian AVMAC, as opposed to the scenario discussed in Remark~\ref{rem:SingleDiffM}, the cases where one of the users has zero capacity can be derived from the single user results. 
This occurs as the minimal state cost $\tLambda_1(F_{U,X_1})$ for symmetrizability-$\Xset_1$, given in Lemma~\ref{lemm:GPscostPM},  is the same as the minimal state cost for single user symmmetrizability of the Gaussian AVC (see \cite{CsiszarNarayan:88p}). Now, based on Csisz\'ar and Narayan's results on the single user Gaussian AVC \cite{CsiszarNarayan:91p}, we have the following. If $\plimit_1\leq\Lambda$ and $\plimit_2>\Lambda$, then the individual capacities of User 1 and User 2 are $C_1=0$ and $C_2=\frac{1}{2}\log\left( 1+\frac{\plimit_2}{\Lambda+\sigma^2} \right)$, respectively, which implies (\ref{eq:MCavcGaussb}). Similarly, if $\plimit_1>\Lambda$ and $\plimit_2\leq\Lambda$, then the individual capacities are $C_1=\frac{1}{2}\log\left( 1+\frac{\plimit_1}{\Lambda+\sigma^2} \right)$ and $C_2=0$, which results in (\ref{eq:MCavcGaussc}).
If $\plimit_1\leq\Lambda$ and $\plimit_2\leq\Lambda$, then $C_1=C_2=0$, hence, (\ref{eq:MCavcGaussd}) follows. 
\qed

\end{appendices}

\section*{Acknowledgment}
We gratefully thank Vinod M. Prabhakaran (Tata Institute of Fundamental Research) and Yiqi Chen (Shanghai Jiao Tong University)
 for useful discussions and valuable comments.

\bibliography{references2}{}
 
\end{document}